\documentclass[12pt]{article}
\usepackage{amsmath, amssymb}
\usepackage{mathtools}
\usepackage{amsthm}
\usepackage{bbm}

\usepackage{amsmath,amsfonts,bm}
\usepackage{xcolor}
\usepackage[normalem]{ulem}

\usepackage{graphicx}
\usepackage{caption}
\usepackage{subcaption}

\usepackage{algorithm,algorithmic}

\usepackage{amsthm}

\usepackage{footmisc}
\usepackage{float}

\theoremstyle{plain}
\newtheorem{theorem}{Theorem}[section]
\newtheorem{proposition}[theorem]{Proposition}
\newtheorem{lemma}[theorem]{Lemma}
\newtheorem{corollary}[theorem]{Corollary}
\theoremstyle{definition}

\theoremstyle{remark}
\newtheorem{remark}[theorem]{Remark}

\def\eqref#1{(\ref{#1})}

\def\1{\bm{1}}

\def\eps{{\epsilon}}

\def\vbeta{{\bm{\beta}}}

\def\vp{{\bm{p}}}

\def\vu{{\bm{u}}}

\def\vx{{\bm{x}}}
\def\vy{{\bm{y}}}
\def\vz{{\bm{z}}}

\def\mF{{\bm{F}}}

\def\mM{{\bm{M}}}

\DeclareMathAlphabet{\mathsfit}{\encodingdefault}{\sfdefault}{m}{sl}
\SetMathAlphabet{\mathsfit}{bold}{\encodingdefault}{\sfdefault}{bx}{n}

\def\gN{{\mathcal{N}}}

\def\sP{{\mathbb{P}}}

\def\sR{{\mathbb{R}}}

\newcommand{\R}{\mathbb{R}}

\DeclareMathOperator*{\argmax}{arg\,max}

\usepackage[titletoc,toc,title]{appendix}
\usepackage{float}
\usepackage[margin=1in]{geometry}     
\usepackage{multicol}     
\usepackage{graphicx}
\usepackage{setspace}   
\usepackage{url}                        
\usepackage[sort]{natbib}               
\bibpunct{(}{)}{;}{a}{,}{,}             
\usepackage{mathptmx}
\usepackage{caption}
\usepackage{adjustbox}
\usepackage{tikz}
\usepackage{bm}
\usepackage[backref=page]{hyperref}
\usepackage{etoolbox}
\makeatletter
\patchcmd{\BR@backref}{\newblock}{\newblock(Cited in page~}{}{}
\patchcmd{\BR@backref}{\par}{)\par}{}{}
\makeatother

\usepackage{cancel}

\hypersetup{
	colorlinks = true,
	linkcolor = blue,
	citecolor = blue,
	runcolor = blue,
	urlcolor = blue
}

\newcommand*{\Scale}[2][4]{\scalebox{#1}{$#2$}}

\title{Competition and Collusion in Two-Sided Markets with an Outside Option\thanks{This work was partially supported by 
NSF awards DMS 2124913 and 2427955. { We sincerely thank Prof.~Özlem Bedre Defolie and the anonymous reviewers for their invaluable feedback and insights, which have significantly enhanced the quality of our work.}}}
\author{Cristian Chica, Yinglong Guo, and Gilad Lerman\thanks{School of Mathematics, University of Minnesota. Email addresses: chica013@umn.edu, guo00413@umn.edu, lerman@umn.edu.}}

\begin{document}

\maketitle
\begin{abstract}

We introduce pricing formulas for competition and collusion models of two-sided markets with an outside option. For the competition model, we find conditions under which prices and consumer surplus may increase or decrease if the outside option utility increases. Therefore, neglecting the outside option can lead to either overestimation or
underestimation of these equilibrium outputs. Comparing collusion to competition, we find that in cases of small cross-side externalities, collusion results in decreased normalized net deterministic utilities, reduced market participation and increased price, on both sides of the market. Additionally, we observe that as the number of platforms increases in the competition model, market participation rises. Profits, however, decrease when the net normalized deterministic utility is sufficiently low but increase when it is high. Furthermore, we identify specific conditions that quantify the change of price and consumer surplus when the competition increases.



\end{abstract}
\begin{center}
    
\end{center}

\section{Introduction}
\label{sect:introduction}
Platform businesses have immensely grown in the last several decades due to the widespread adoption of communication technologies.\footnote{According to the United States Census Bureau, the percentage of US citizens reporting owning a computer has grown from 8\% in 1984 to 89\% in 2016 (see, e.g., \cite{RyanCamilleCensus}).} For example, the sales of Amazon, which is a platform, have grown from \$148 millions in 1997 to \$386 billions in 2020  \citep{wells2018amazon}. 
Platforms facilitate the interaction between different types of users, such as buyers and sellers (Amazon and eBay), drivers and riders (Uber and Lyft) and content creators and consumers (YouTube, Twitch and Spotify). Their business model has become very popular, but its careful study is still in an early stage, where the first research works are from the beginning of this century (see, e.g., \cite{rochet2003platform} and \cite{caillaud2003chicken}). There are still many open questions and, in particular, a complete model of platform competition is still far from reach. 

A common yet limiting assumption in platform modeling is full market coverage, meaning that in equilibrium all users join at least one platform. While this assumption is strong, most models incorporate it because it leads to explicit equilibrium pricing formulas (see, e.g., \cite{tan2021effects}). In this work, we relax this assumption by considering a model in which $N$ horizontally differentiated platforms compete across two market sides—buyers and sellers (collectively referred to as users)—who can either join one of the platforms (single-homing) or choose not to participate, a choice referred to as the outside option.

Dating apps offer a clear example of a market with a significant outside option and platform competition. Many users still prefer traditional, non-priced methods of meeting partners---such as through friends or at school---underscoring the importance of the outside option. This market also includes numerous competing platforms. 
While users typically multi-home across apps (see, e.g., the last table in \cite{datingreport24}\footnote{The table shows the percentage of people in different age groups who have ever used specific dating sites or apps among those who have used any dating site or app. It indicates that, on average, individuals have used about two different apps, with this average decreasing with age. However, since people may use different apps at different times, the number of apps used simultaneously is likely lower.}), some platforms encourage behavior closer to single-homing.\footnote{For example, Hinge markets itself as ``designed to be deleted,'' while loyalty-encouraging subscription models, like Bumble Boost, and curated apps, like The League or JDate, may encourage users to stick with a single app.}
For tractability, some economic models assume single-homing in this context (see, e.g., \cite{halaburda2018competing,gal2020market}). We thus use this market to ground some of our theoretical results.
Ride-sharing services (e.g., Uber and Lyft) offer another example of a market with a strong outside option—public transportation, scooters, or e-bikes—but with pronounced multi-homing, as users switch between platforms depending on availability, pricing, or convenience.

We develop pricing formulas for our model and express them in terms of the equilibrium normalized net deterministic utility that platforms provide to users, i.e., the difference between the deterministic utility of users joining one platform and the deterministic utility of the outside option. This allows us to transition from a space of prices to a space of utilities in the spirit of \cite{armstrong2001competitive}.

We utilize these pricing formulas to study competition and collusion between platforms with an outside option. We first establish sufficient conditions for the existence and uniqueness of a symmetric Nash equilibrium and a collusive equilibrium in this setting. We further show that under small cross-side externalities,\footnote{Cross-side externalities capture 
the benefits that users on one side of the market derive from interaction with users on the other sides of the market. When these externalities are positive, platforms are confronted with the ``chicken \& egg'' problem: to attract buyers, the platform must have a large base of sellers, who will join the platform if and only if there are many buyers in the platform (see \cite{caillaud2003chicken}).} 
the normalized net deterministic utilities and market participation are smaller in collusion than in competition; furthermore, the prices on both sides are bigger in collusion than in competition. 
We also study the effect of increasing competition on the outputs of the competitive Nash equilibrium. In this case of increasing competition, we specify regimes for the decrease or increase of both prices and consumer surplus (these regimes depend on the size of the user's heterogeneity in tastes, the number of platforms and the size of network externalities). We demonstrate how these results can shed light on the pricing strategies observed in dating apps and how they change depending on the heterogeneity of the population and the size of network externalities.\footnote{Dating markets with heterogeneous populations are dominated by general apps like Tinder and Bumble, while more homogeneous populations tend to use niche apps like The League and JDate.}  We further show in this case that market participation always increases, and profits decrease if the net normalized deterministic utility is sufficiently small and increase if the net normalized deterministic utility is sufficiently large.

The size of the outside option determines the sign of the net deterministic utility in a nonlinear fashion. Indeed, we show that there exists a critical threshold for the deterministic outside option utility such that above this threshold users only receive negative deterministic utility, and below this threshold the sign of the net deterministic utility depends on the relative size of the heterogeneity in user's tastes versus the within-side externalities. 

Moreover, we show that when the outside option increases, prices and consumer surplus may increase or decrease, based on the relative size of the heterogeneity in user's tastes versus the within-side externalities. In particular, we show that a model of platform competition that omits the outside option may overestimate or underestimate the true equilibrium price.

Our pricing formulas imply the following standard results for platforms, accounting for an outside option (see further references and details in 
Section \ref{sect:thm_result_gumbel}): (i) platforms hold market power and charge users in a way that is directly proportional to user's heterogeneity in tastes; (ii) if the within-side externalities\footnote{Within-side externalities capture negative within-side competition effects and positive collaboration effects between users on the same side of the market. For example, competition between content creators in the same platform (e.g., \textit{YouTube}, \textit{Twitch} and others) and collaboration between open source programmers (e.g., \textit{C++}, \textit{Python}, \textit{Linux} and others). } 
are positive, platforms subsidize users on one side of the market by an amount that is proportional to the joining population on this side of the market; (iii) if the cross-side externalities are positive, platforms subsidize users on one side of the market with an amount proportional to the joining population on the other side of the market. The alignment of our results with existing ones suggests that the standard platform pricing strategy can be more general than previously understood.

\textbf{Related Literature.} The study of two-sided markets has emerged in the last few decades. The earlier works of \cite{rochet2003platform,rochet2006two}, \cite{caillaud2003chicken} and \cite{armstrong2006competition} laid out the modeling foundations of two-sided markets with network externalities. These works shed some light on how equilibrium outputs of platform competition and platform monopoly depend on: (i) the size of the network externalities relative to user's heterogeneity in tastes; (ii) users being able to join either one or two platforms (i.e., having a single-home or multi-home option). \cite{weyl2010price} and \cite{white2016insulated} placed an emphasis on platform competition with insulated tariffs, which allow platforms to choose participation rates rather than prices. 

\cite{tan2021effects} modeled competition between $N\geq2$ platforms serving multiple sides of a market. In this setting, customers are heterogeneous and modeled through random idiosyncratic preferences. Under general conditions for the probability distribution of idiosyncratic preferences, they characterized a symmetric subgame perfect equilibrium. Our model generalizes their approach by incorporating an outside option and conducting an extensive analysis for specific probability distributions. This generalization enables us to study the effects of the outside option on equilibrium outputs and collusion in two-sided markets. 
We remark that \cite{tan2021effects} pricing formulas are similar to the ones in \cite{armstrong2006competition}, \cite{white2016insulated}, \cite{jullien2019information} and \cite{chica2021exclusive} in the sense that they explicitly depend on the parameters. 
Due to the challenges of the outside option in our model, our pricing formulas are implicitly determined by the equilibrium net deterministic utility and cannot be explicit like the previous formulas. Similar to \cite{tan2021effects}, but under the assumption of an outside option, we find that increasing competition can either raise or lower equilibrium prices and consumer surplus—depending on the relative strength of network externalities versus user heterogeneity in tastes—and may also increase or decrease platform profits, depending on a condition involving the normalized net deterministic utility. However, we show that market participation always increases, whereas \cite{tan2021effects} assumes $100\%$ market participation. Moreover, our collusion analysis is novel and cannot be addressed in markets with full participation.

\cite{cohen2022competition} focused on the particular setting and idiosyncrasies of ride-sharing services (e.g., \textit{Uber} and \textit{Lyft}).
One of their interesting results is that under collusion riders pay a larger price and drivers receive a lower wage than under competition. This result is similar to one of our results (to see this one should note that the wage in their model is a negative price in our model). Nevertheless, the modeling choice for customer demand is different in both works. In their model, network externalities are endogenous, whereas ours are exogenous and thus follow the traditional setting that stems from \cite{armstrong2006competition}. As such, we can examine the effects of the network externalities on the equilibrium outputs. Another major difference between these works is the type of solution obtained. They inductively solve a sequence of problems, which approximates the Nash equilibrium, and their final solution is a limit of the former solutions. We characterize the best-response of a platform that deviates from the symmetric Nash equilibrium and directly study properties of this equilibrium. The major advantage of our approach is that it allows us to study the effects of competition on the equilibrium outputs, because we can differentiate these outputs. 

Our work also pertains to the literature on platform collusion. On the theoretical side, \cite{dewenter2011semi} compared between competition, semi-collusion (i.e., collusion in only one side of the market) and full collusion in a special model for the newspapers market, where the two sides of the market were represented by advertisers and readers. 
Comparing full collusion to competition, they found out that for the advertisers market, participation is lower and prices are higher under full collusion. This result is similar to one of ours, as we show that for both sides of the markets, all sizes of within-side externalities and relatively small cross-side externalities, collusion always leads to less market participation and higher prices.\footnote{Note that we focus on the study of full collusion, whereas \cite{dewenter2011semi} and \cite{lefouili2020collusion} also study semi-collusion in two-sided markets and show that if the cross-side externalities are positive and sufficiently large, then semi-collusion may benefit users on the collusive side and harm users on the competitive side. Furthermore, \cite{dewenter2011semi} also show that on the readers side, the collusion price may be lower than the competitive price if the competition in the advertising market is high and the newspaper market is large. We remark that our demand specification is different and, in particular, we do not incorporate a parameter for the relative size of the markets.} 
 
Other relevant works that include an outside option are: \cite{jeitschko2020platform}, which studied how consumers and firms endogenously choose between different homing options or the outside option; \cite{correia2019horizontal}, which examined the welfare effects of horizontal mergers between multi-sided platforms while incorporating an outside option for consumers; \cite{tremblay2023cournot}, which analyzed Cournot platform competition in two-sided markets with indirect network effects, where both consumers and sellers have an outside option; \cite{peitz2023asymmetric} studied a model of asymmetric platform oligopoly while allowing partial user participation, i.e., outside options; and \cite{teh2023multihoming} study the effects of allowing multi-homing for both sides of the market while also incorporating outside options.

\textbf{Organization of the Article:}
Section \ref{sect:platform_competition} formulates our platform models of competition and collusion. Section \ref{sect:thm_result_gumbel} solves the models via backward induction. Section \ref{sect:collusionvscompetition}  
compares the outputs of colluding and competing market models. Section \ref{sect:effectsCompetition} examines the effects of increasing competition on the equilibrium quantities of the competition model. Section \ref{sect:econ_poli_discussion} examines the economic implications of our main results.
Section \ref{sect:conclusions} concludes this work.

\section{The Platform Model with an Outside Option}
\label{sect:platform_competition}

We formulate a platform competition model by following previous works, such as \cite{white2016insulated} and  \cite{tan2021effects}. Let $N\geq2$ be the number of horizontally differentiated platforms in the market. Each index $i\in\mathcal{N} := \{1,\dots, N\}$ represents a platform competing in two different sides of a market. We index each side by $k\in \{b,s\}$, where $b$ and $s$ represent buyers and sellers, respectively. Each platform $i$ sets prices for each side of the market, denoted by $\boldsymbol{p}^i = (p_b^i,p_s^i)$. The endogenous mass of users on each side of the market subscribed to platform $i$ is denoted by $\vx^i=(x_b^i,x_s^i)\in [0,1]^2$. For $i=0$, $\vx^0=(x_b^0,x_s^0)\in [0,1]^2$ denotes the mass of users not participating in the market. 

Users on side $k$ of the market have idiosyncratic preferences for platforms and for the outside option. These preferences are captured by the i.i.d.~random variables $\varepsilon_k^i\sim F_k(\cdot)$, where $k\in\{b,s\}$, $i\in\gN\cup\{0\}$ ($i\in\gN$ for the platforms and $i=0$ for the outside option) and $F_k(\cdot)$ is a differentiable probability distribution. 

The game consists of two stages. In stage 1, platforms strategically choose prices to maximize profits. In this article, we study two scenarios in stage 1: (i) The competition scenario, where firms compete and maximize individual profits; (ii) The collusion scenario, where firms collude and jointly maximize profits. In stage 2, given the prices determined by the platforms, users on each side of the market choose whether to participate or not and if they participate they also choose which platform to join. The game is solved using backward induction and we thus first describe the details of the second stage and then the first one.

$(i)$ \textbf{Stage 2 (users' interactions):} 
Any user on side $k\in\{b,s\}$ of the market who joined platform $i\in\mathcal{N}$ receives the following utility:
\begin{equation}\label{uki}
    \hat{u}_k^i := \varepsilon_k^i - p_k^i + \phi_k(\boldsymbol{x}^i),
\end{equation}
where $\varepsilon_k^i$ represents the idiosyncratic utility the user enjoys; $p_k^i$ is the price paid by the user to access services provided by the platform (it was determined in stage 1); and $\phi_k:[0,1]^2\longrightarrow\R$ is a Lipschitz differentiable function so that $\phi_k(\boldsymbol{x}^i)$ captures the network benefits that the user receives from all other users who are also joining platform $i$. We further denote the deterministic component of the utility by $$u_k^i := - p_k^i + \phi_k(\boldsymbol{x}^i).$$ If a user does not join any platform, it receives the utility 
\begin{equation}
    \label{uk0}
    \hat{u}_k^0 = \varepsilon_k^{0} + u_k^{0},
\end{equation}
where $u_k^{0}\in\mathbb{R}$ is a constant representing the deterministic outside utility.\footnote{Most models of platform competition leave out the analysis of the outside utility option that users have. By doing so, they cut out from the profit maximization process the trade-off between market participation and competition. In this article, we study this trade-off.}  Note that users will choose the platform that maximizes their utility, i.e., they will join $j\in \argmax_{i\in \mathcal{N}\cup\{0\}} \{\hat{u}_k^i\}$. It follows that the mass of users from side $k$ joining platform $i$ solves the equation
\begin{equation}\label{xki}
    x_k^i = \mathbb{P}\left(\hat{u}_k^i>\max_{j\in \mathcal{N}\cup\{0\}\setminus \{i\}} \{\hat{u}_k^j\}\right) \quad \forall k\in\{b,s\},\  i\in\mathcal{N}\cup\{0\}.
\end{equation}
Proposition \ref{prop:existence_xki} below implies that \eqref{xki} has at least one solution for any set of prices $\{(p_b^i,p_s^i)\}_{i=1}^N$, and also establishes a sufficient condition for a unique solution of  (\ref{xki}). 

$(ii)$ \textbf{Stage 1 (platforms' optimization):} We consider the following two scenarios: a competitive market and a collusive market.

(a) A competing market: Platform $i$, for each $i\in\mathcal{N}$, sets prices $\{p^i_b,p_s^i\}$ that maximize individual profits, i.e., platform $i$ solves
\begin{equation}\label{pi}
\max_{\{p^i_b,p_s^i\}}\text{  } \pi^i\left(p_b^i,p^i_s\right) , \text{ where }  \pi^i\left(p_b^i,p^i_s\right) :=  x^i_bp^i_b+x^i_sp^i_s,\end{equation}
and $x_k^i$ is implicitly defined by (\ref{xki}). We adopt the standard assumption of zero marginal cost for serving users on sides $b$ and $s$. A Nash equilibrium corresponding to \eqref{pi} is referred to as a \textit{competitive Nash equilibrium} (CNE).

(b) A colluding market: The colluding platforms act as a single platform trying to maximize joint profits across all sides of the market by charging one price on every side of the market; i.e., they solve 
\begin{equation}
    \max_{\substack{p_b,p_s}} \ 
    \Pi_{\text{tot}}(p_b,p_s), \text{ where }
    \Pi_{\text{tot}}(p_b,p_s) 
 := \sum_{i=1}^N\left(x^i_bp_b+x^i_s p_s\right).\label{pim}
\end{equation}
As in the competitive case, we assumed that the marginal costs of serving sellers and buyers are zero. We refer to any maximizer of (\ref{pim}) as \textit{collusive equilibrium} (CE).\footnote{While, in general, collusion can be any situation where two or more platforms jointly make decisions, in this article, we focus on the worst-case-scenario, where all platforms collude.}

In order to fully quantify the equilibrium outcomes, we make two additional assumptions: \\
(a) Assumption I: The idiosyncratic preferences, $\{\varepsilon_k^i\}_{k\in\{b,s\},  i\in\gN\cup\{0\}}$, are i.i.d.~Gumbel distributed with parameters $(\mu_k,\beta_k)$. That is, for $k\in\{b,s\}$, $\beta_k>0$ and $\mu_k\in \R$, 
the distribution $F_k(\cdot)$ is
\begin{equation}\label{eqn:cdf_gumbel}
    F_k(z) = e^{-e^{\frac{\mu_k -z}{\beta_k }}}.
\end{equation}
We claim that this assumption is natural since it gives rise to the classical Logit model (see \cite{werden1996use}, \cite{anderson1992logit}), which describes the demand of heterogeneous consumers for a set of differentiated goods (see \cite{berry1994estimating}, \cite{conlon2020best}, \cite{besanko1998logit}). To support this claim, we note that the central equation in this work is \eqref{xki}, which can be rewritten using \eqref{uk0} and the alternative variables \(\theta_{k}^i = \varepsilon_k^i - \varepsilon_k^0\), \(i \in \gN \cup \{0\}\), \(k \in \{b,s\}\) as follows:
$$
x_k^i = \mathbb{P}\left(\theta_{k}^{i} + u_{k}^{i} \geq \max_{j=0,1,\ldots,N,j\neq i} \left\{ \theta_{k}^{j} + u_{k}^{j} \right\} \right), \quad \forall k \in \{b,s\},\  i \in \gN \cup \{0\}.
$$
We further note that \(\theta_k^i \sim \textnormal{Logistic}(0,\beta_k)\) for \(i \in \gN\) and thus conclude the claim.

(b) Assumption II: The function $\phi_k(\vx)$ is linear, i.e., it can be represented by multiplying with the real-valued matrix $\Phi \in \mathbb{R}^{2 \times 2}$
\begin{equation}\label{def:Phi}
    \left[\begin{array}{c}
   \phi_b(x_b,x_s)  \\
     \phi_s(x_b,x_s)
    \end{array}\right] = \underbrace{\left[\begin{array}{cc}
    \phi_{bb}  & \phi_{bs}  \\
     \phi_{sb}    & \phi_{ss}
    \end{array}\right]}_{ \Phi}\left[\begin{array}{c}
    x_b   \\
     x_s
    \end{array}\right].
\end{equation}

We remark that the following results do not require these additional assumptions:  
The existence and uniqueness of the solution of (\ref{xki}) (see Proposition \ref{prop:existence_xki} in Section \ref{sect:thm_result_gumbel}) and the derivation of the first- and second-order conditions for both (\ref{pi}) and (\ref{pim}) (see Lemma \ref{lemma:foc_cneA}, \ref{lemma:soc_cneA}, \ref{lemma:foc_ce} and \ref{lemma:soc_ce} in Appendix \ref{appendixA}). We note that many of the results presented in this paper can be extended to other probability distributions of economic interest. In the Online Appendix, we demonstrate this for the exponential distribution with two platforms. In particular, we explicitly derive the first-order condition for (\ref{pi}) using Mathematica.  
Additionally, we present numerical simulations supporting results similar to those shown in Propositions \ref{prop:dpkduk0}, \ref{prop:dpikduk0}, and \ref{prop:dCSkuk0}, but using the exponential distribution. It is important to recognize that each probability distribution requires special treatment, and the analysis of the Gumbel distribution is already quite lengthy and complex.

\section{Equilibrium}
\label{sect:thm_result_gumbel}
We solve our model using backward induction. We first 
study the solution to (\ref{xki}) and show that 
for any set of prices, $\{(p_b^i,p_s^i)\}_{i=1}^N$, there is a well-defined set of market shares, $\{(x_b^i,x_s^i)\}_{i=0}^N$, that solve (\ref{xki}) and under a certain condition they are unique.
Next, we
characterize the symmetric CNE of \eqref{pi} (i.e., the CNE such that $p_k^ i=p_k^*$ for each $i\in\gN$) and the CE of \eqref{pim}. At last, 
we interpret the resulting equilibrium pricing and market share formulas.

\subsection*{Stage 2 Solution: Users' Maximization}\label{subsection:users_maximization}
We establish sufficient conditions for the existence and uniqueness for \eqref{xki}.
We recall that \eqref{xki} captures the users' dynamics when prices change. 
Let $\vu=(u^0,u^1,\dots,u^N)\in \R^{N+1}$ and  for $k\in\{b, s\}$ and $i\in\gN\cup\{0\}$ define 
\begin{equation}\label{def:Qki}
    T_k^i(\vu) := \sP(\epsilon_k^i > \max_{j\neq i}(\epsilon_k^j +u^j-u^i)).
\end{equation}
In view of (\ref{def:Qki}), (\ref{xki}) can be rewritten as
\begin{equation}
    \label{xki_Tki}
    x_k^i= T^i_k(u_k^0,\phi_k(\vx^1)-p_k^1,\dots ,\phi_k(\vx^N)-p_k^N). 
\end{equation}
It follows that a vector $\vx = (x_b^0,x_s^0,\dots ,x_b^N,x_s^N)$ solves (\ref{xki}) if and only if it is a fixed point of the map $\Sigma:[0,1]^{2(N+1)}\longrightarrow [0,1]^{2(N+1)}$ given by 
\begin{equation}\label{def:sigma}
    \Sigma(\vx) := (T_b^0(\vu_b),T_s^0(\vu_s),\dots,T_b^N(\vu_b),T_s^N(\vu_s)),
\end{equation}
where $\vu_k = (u_k^0,\phi_k(\vx^1)-p_k^1,\dots ,\phi_k(\vx^N)-p_k^N)$ and $\vx^i = (x_b^i,x_s^i)$. Proposition \ref{prop:existence_xki} below, shows that (\ref{def:sigma}) always has at least one fixed point and thus (\ref{xki}) always has a solution. This proposition also provides sufficient conditions for the uniqueness of this fixed point and the solution of (\ref{xki}).  Its formulation requires the following Lipschitz-type constants:
\begin{equation}\label{def_MT_Mphi}
    \begin{split}
        M_T &:=  \max_{k\in\{b, s\},i\in\gN\cup\{0\}}\sup_{\vu\in\R^{N+1}}\sum_{l=0}^N\left|\frac{\partial T_k^i(\vu)}{\partial u^l}\right|, \textnormal{ and}\\
        M_\phi &:= \max_{k\in\{b, s\}} \sup_{(x_b,x_s)\in[0,1]^2}\sum_{l\in \{b,s\}}\left|\frac{\partial \phi_k(x_b,x_s)}{\partial x_l}\right|.
    \end{split}
\end{equation}
We remark that ${\partial T_k^i(\vu)}/{\partial u^l}$ captures the user's sensitivity to changes in utility levels. Similarly, ${\partial \phi_k(x_b,x_s)}/{\partial x_l}$ measures how externalities change when more people join one specific platform. 

\begin{proposition}[Existence and Uniqueness of Market Shares]\label{prop:existence_xki}
For any prices $\{(p_b^i,p_s^i)\}_{i=1}^N\subset\mathbb{R}^2$, there exists a solution to (\ref{xki}), $\vx = (x_b^0,x_s^0,x_b^1,x_s^1,\cdots,x_b^N,x_s^N)$, such that for each $k\in\{b,s\}$, $\sum_{i=0}^Nx_k^i=1$. Moreover, if $M_T M_\phi<1$, where $M_T$ and $M_{\phi}$ are given by (\ref{def_MT_Mphi}), then the solution of (\ref{xki}) is unique. 
\end{proposition}

This Proposition provides sufficient conditions for the mapping $\{(p_b^i,p_s^i)\}_{i=1}^N\mapsto  \{(x_b^i,x_s^i)\}_{i=0}^N$ to be well-defined. Its proof is in Appendix \ref{appendixA}.\footnote{We remark that while  Proposition \ref{prop:existence_xki} is used to identify symmetric Nash equilibria, it should not be restricted to symmetric market shares. 
When proving the existence of symmetric Nash equilibria (see Proposition \ref{prop:existence_gumbel}), it is necessary to consider all possible deviations from the equilibrium path, including those that lie off the symmetric path.}

\subsection*{Stage 1 Solution: Platforms' Optimization}\label{subsection:platforms_maximization}


We establish sufficient conditions for the existence and uniqueness of symmetric solutions of (\ref{pi}) and (\ref{pim}).
We first focus on symmetric solutions for (\ref{pi}).
For this purpose, we use the following transformation: 
\begin{equation}
    \label{def:zk}
    z_k := \frac{u_k-u_k^0}{\beta_k}, \
    \text{ for } k\in\{b,s\}. 
\end{equation}
We note that $u_k-u_k^0 \equiv -p_k+\phi_k(\vx)-u_k^0$ captures the difference between the deterministic utility of users (sellers or buyers) joining one platform and the deterministic utility of the outside option. We remark that in the symmetric case any platform charges the price $p_k$ for $k \in\{b,s\}$ and the market shares are given by $\vx = (x_b,x_s)$. The Gumbel distribution parameter $\beta_k$ is a measure of the standard deviation of the idiosyncratic preference $\varepsilon_k^i$ and it captures the degree of heterogeneity in users' tastes.\footnote{Note that in general, if $\varepsilon_k^i\sim G(\mu_k,\beta_k)$, then 
$\textnormal{Var}[\varepsilon_k^i]=\frac{\pi^2}{6}\beta_k^2$ and thus the standard deviation of $\varepsilon_k^i$ is $\frac{\pi}{\sqrt{6}}\beta_k$.}  Throughout the article, we will refer to $z_k$ as the \textit{normalized net deterministic utility}
of users on side $k$ of the market.

We can write the first-order condition (FOC) of (\ref{pi}) as a function of $z_k$.\footnote{It is a known fact that attempting to solve (\ref{pi}) by means of an FOC with respect to prices $\{(p_b^i,p_s^i)\}_{i=1}^N$ produces a non-tractable system of equations (see, e.g., \cite{tan2021effects} and \cite{chica2021exclusive}). By contrast, our proofs in the appendix take derivatives with respect to $\{x_b^i, x_s^i\}_{i=1}^N$. By Proposition \ref{prop:existence_xki} and the implicit function theorem, there is a well-defined locally 1-1 mapping from  $\{(x_b^i,x_s^i)\}_{i=1}^N$ to $\{(p_b^i,p_s^i)\}_{i=1}^N$.}  

\begin{proposition}[FOC of (\ref{pi})]\label{prop:FOC_gumbel}
    Suppose there is a symmetric equilibrium $(p_b^*,p_s^*)$ solution of (\ref{pi}) with market shares $(x_b^*,x_s^*)$. If one platform unilaterally deviates from this symmetric CNE, the FOC that characterizes its best-response is given by
    \begin{equation}
     \vbeta \vz=\left(\Phi -H(\vz)\right)\Omega(\vz) - \vu_0 ,\label{FOCs_z}
    \end{equation}
    where $\vbeta=\textnormal{diag}(\beta_b,\beta_s)$, $\vz=(z_b,z_s)$, $\vu_0=(u_b^0,u_s^0)$, $\Phi$ is the externalities network matrix defined in \eqref{def:Phi}, $\Omega(\vz)=(\omega(z_b),\omega(z_s))^T$ with $\omega:\R\longrightarrow(0,\frac{1}{N})$  such that $\omega(z):= \frac{1}{e^{-z}+N}$, and $H(\vz)$ is a $2x2$ matrix defined as 
\begin{equation}
    \label{def:H_matrix}
  H(\vz):= \left[\begin{matrix}
      L_bd_b K_s + h_b-\phi_{bb}  & -\phi_{sb} (d_s L_b + 1) \\
      -\phi_{bs} (d_b L_s + 1) &L_sd_s K_b + h_s -\phi_{ss} 
  \end{matrix}\right],
\end{equation}
where $L_k$, $d_k$ and $h_k$ for each $k\in\{b,s\}$ are given by
\begin{equation}
\begin{split}
L_k & = \frac{(N-1)\beta_k}{J_\phi}(1 + Ne^{z_k}),\\
d_k & = \beta_k(1 + N e^{z_k}),\\
h_k &= \beta_k(1 + e^{z_k}) (e^{-z_k} + N), \\
K_k & = \phi_{kk} - \beta_k (1+Ne^{z_k})(e^{-z_k} + N - 1),\\
J_\phi &= K_bK_s - \phi_{sb}\phi_{bs}. 
\end{split}\label{eqn_z_Lk}
\end{equation}
\end{proposition}

Let us assume $\vz^\ast = (z^\ast_b, z^\ast_s)^T$ is the unique solution of (\ref{FOCs_z}) and it satisfies a corresponding second order condition. 
We discuss below (see Proposition \ref{prop:existence_gumbel}) sufficient conditions for this assumption.
We use $\vz^\ast$ to characterize the symmetric equilibrium solution of (\ref{pi}),  $\vp^*=(p_b^*,p_s^*)^T$, with market shares $\vx^*=(x_b^*,x_s^*)^T$. 
By applying \eqref{def:Qki} and \eqref{xki_Tki} evaluated at 
$u_k^i = u_k^* = -p_k^* + \phi _k(\vx^*)$, where $i\in \mathcal{N}$ and $k\in\{b,s\}$, one can show (see \eqref{eqn:proof_xki_omega} in Appendix \ref{appendixA}) that 
\begin{equation}
\label{eq:x=Omegaz}
x_k^* = \omega(z_k^*) \equiv \frac{1}{e^{-z_k^*}+N} \ \textnormal{ and thus } \ \vx^* = \Omega(\vz^*).
\end{equation} 
We further note that (\ref{def:zk}) implies that $\vbeta \vz^\ast = -\vp^*+\Phi\vx^*-\vu_0$. Combining the latter equation, (\ref{FOCs_z})
and \eqref{eq:x=Omegaz}, the symmetric CNE of (\ref{pi})  is given by 
\begin{equation}\label{eqn_pz}
\begin{split}
\vp^*
&= H(\vz^*)\Omega(\vz^*) \textnormal{ and }\\
\vx^* &= \Omega(\vz^*).
\end{split}
\end{equation}

In order to ensure that \eqref{eqn_pz} yields the symmetric CNE, we next establish a sufficient condition for (\ref{FOCs_z}) to have a unique solution that satisfies a corresponding second order condition.
It uses the following function
\begin{equation}\label{def:RNphi_kk_proof}
    f(N):= \frac{2(N-1)}{N^{2}},
\end{equation}
where we note that $f$ approaches 0 as $N\to\infty$. It also uses the notation $B_\epsilon(0)$ for the ball in $\R^2$ of radius $\eps>0$ around the origin.
\begin{proposition}[Existence and uniqueness of the symmetric CNE]
    \label{prop:existence_gumbel}
Suppose that $N\geq2$ and for each $k\in\left\{ b,s\right\}$, $\left(\phi_{kk},\beta_{k}\right)$ satisfies 
\begin{equation}\label{condition_existence_beta}
\textnormal{either }
\left(\phi_{kk}\leq0\textnormal{ and }\beta_{k}>0\right)\textnormal{ or }\left(\phi_{kk}>0\textnormal{ and }\beta_{k}>f\left(N\right)\phi_{kk}\right).
\end{equation}
Then, there exists $\epsilon>0$ such that for any $(\phi_{bs},\phi_{sb}) \in B_\epsilon(0)$ there is a unique solution of (\ref{FOCs_z}) and this solution satisfies a second order condition.\footnote{We clarify that the $\epsilon$ in Proposition \ref{prop:existence_gumbel} depends on $(\phi_{bb},\phi_{ss},\beta_b,\beta_s,N,u_b^0,u_s^0)$, but for simplicity we denote it by $\epsilon$. We use the same convention in other places in this article where a similar condition with $\epsilon$ appears.} Furthermore, (\ref{eqn_pz}) yields the unique symmetric CNE of (\ref{pi}).
\end{proposition}

Proposition \ref{prop:existence_gumbel} guarantees the existence and uniqueness of a symmetric CNE for a large family of the parameters  $\{\phi_{kk},\beta_k\}_{k\in\{b,s\}}$. In particular, if the within-side externalities (i.e., those that reflect interactions of the same sides of the market), $\phi_{kk}$, are negative, then existence and uniqueness of a solution for \eqref{FOCs_z} is guaranteed for any size of heterogeneity in user's tastes, $\beta_k$. On the other hand, if the within-side externalities are positive, then existence and uniqueness is only ensured for relatively large sizes of heterogeneity in user's tastes (i.e.,  $\beta_k>f(N)\phi_{kk}$). Recall that as the number of platforms $N$ grows to infinity, $f(N)$ approaches 0. Thus, even for positive within-side externalities, existence and uniqueness of a solution for \eqref{FOCs_z} is guaranteed for any size of $\beta_k$, provided that the number of platforms in the market is large enough. Some form of the latter condition appears in many studies of platform competition (see, e.g., \cite{anderson1992discrete}, \cite{armstrong2006competition}, and \cite{tan2021effects}). This condition ensures that network effects do not always dominate idiosyncratic preferences when users are charged non-zero prices (see, e.g., \cite{chica2021exclusive}). Figure \ref{fig:zk_existence} below shows the region described by (\ref{condition_existence_beta}) when $N=4$. 

\begin{figure}[H]
    \centering
    \caption{Region of $(\phi_{kk},\beta_k)$ that guarantees a unique symmetric CNE when $N=4$ according to Proposition \ref{prop:existence_gumbel}.}
     \label{fig:zk_existence}
    \includegraphics[scale=0.35]{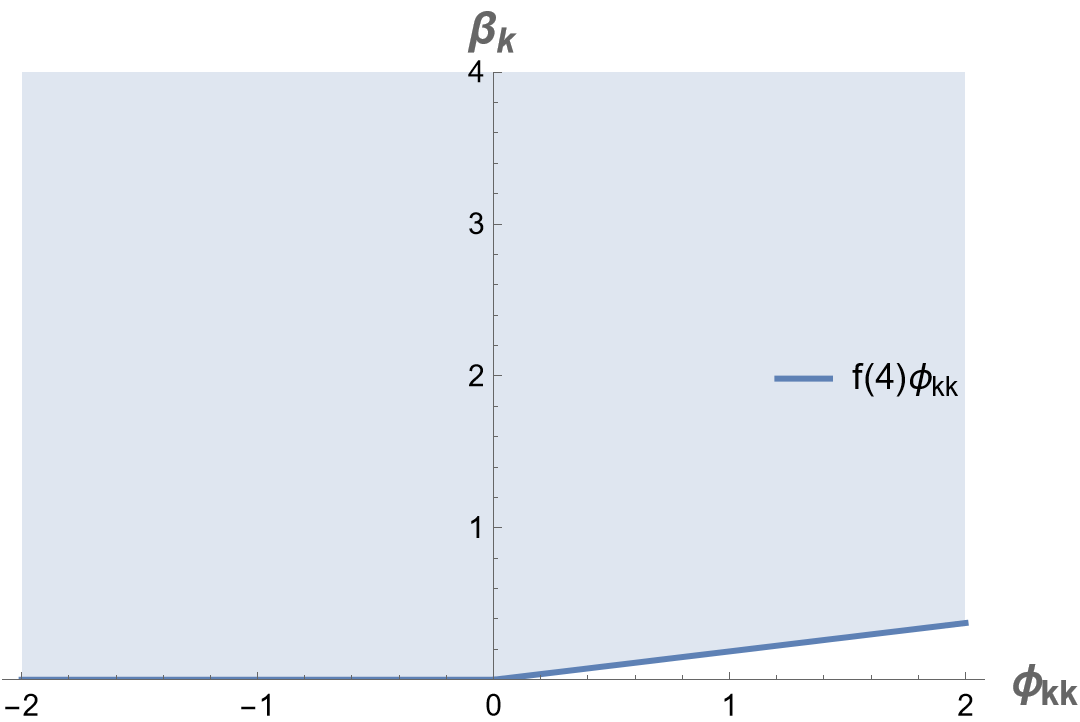}
\end{figure}

Next, we focus on the solution of (\ref{pim}). We first establish the FOC of (\ref{pim}) as a function of $z_k$. 
\begin{proposition}[FOC of (\ref{pim})]\label{prop:FOC_gumbel_collusion}
    The FOC of \eqref{pim} is given by
    \begin{equation}
    \vbeta \vz=(\Phi -H^\text{C}(\vz))\Omega(\vz) - \vu_0 ,\label{eqn:gumbel_mono_b}
\end{equation}
    where $\vbeta$, $\vz$, $\vu_0$, $\Phi$, $\Omega(\vz)$ were defined in Proposition \ref{prop:FOC_gumbel}, and $H^\text{C}(\vz)$ is a $2x2$ matrix defined by 
\begin{equation}
    \label{def:Hm_matrix}
  H^\text{C}(\vz):= \left[\begin{array}{cc}
\frac{\beta_{b}(1+Ne^{z_{b}})^{2}}{e^{z_{b}}}-\phi_{bb} & -\phi_{sb}\\
-\phi_{bs} & \frac{\beta_{s}(1+Ne^{z_{s}})^{2}}{e^{z_{s}}}-\phi_{ss}
\end{array}\right].
\end{equation}
\end{proposition}

Let us assume $\vz^\text{C} = (z^\text{C}_b, z^\text{C}_s)^T$ is the unique solution of \eqref{eqn:gumbel_mono_b} and it satisfies a corresponding second order condition (we provide sufficient conditions for these assumptions in Proposition \ref{coro:mono_uniqueness} below). 
Following the same derivation of \eqref{eqn_pz} (see the proof of Proposition~\ref{prop:FOC_gumbel_collusion} in Appendix \ref{appendixA}), one can show that the CE solution of \eqref{pim}, $\vp^\text{C}=(p_b^\text{C},p_s^\text{C})^T$ and the corresponding market shares, $\vx^\text{C}=(x_b^\text{C},x_s^\text{C})^T$,  satisfy
\begin{equation}\label{eqn_pm}
\begin{split}
\vp^\text{C} &= H^\text{C}(\vz^\text{C})\Omega(\vz^\text{C}) \textnormal{ and }\\
\vx^\text{C} &= \Omega(\vz^\text{C}).
\end{split}
\end{equation}
In order to ensure that \eqref{eqn_pm} yields the CE, we establish sufficient conditions for (\ref{eqn:gumbel_mono_b}) to have a unique solution that satisfies a corresponding second order condition. 



%

\begin{proposition}[Existence and uniqueness of the CE]\label{coro:mono_uniqueness} For any $u_b^0,u_s^0\in \sR$, $\Phi\in \mathbb{R}^{2x2}$, $\beta_b, \beta_s>0$ and $N\geq 2$,
there exists a solution for (\ref{eqn:gumbel_mono_b}). Moreover, if for each $k\in\{b,s\}$, $(\phi_{kk}, \beta_k)$ satisfies
\begin{equation}
    \textnormal{either } (\phi_{kk} \leq 0 \textnormal{ and } \beta_k>0) \ \textnormal{ or } \  (\phi_{kk} >0 \textnormal{ and } \beta_k > \frac{8\phi_{kk}}{27N}), \label{eqn:condition_unique_mono}
\end{equation}
then there exists $\varepsilon > 0$ such that for any $(\phi_{bs}, \phi_{sb}) \in B_{\varepsilon}(0)\subset \sR^2$, the solution for (\ref{eqn:gumbel_mono_b}) is unique, it satisfies a corresponding second order condition and (\ref{eqn_pm}) yields the unique CE of \eqref{pim}. 
\end{proposition}

The proof of Proposition \ref{coro:mono_uniqueness} implies that $f(N)\phi_{kk}$, which was used in Proposition \ref{prop:existence_gumbel}, is strictly bigger than ${8\phi_{kk}}/{(27N)}$ for all $\phi_{kk}>0$. It follows that, if $(\phi_{kk},\beta_k)$ satisfies (\ref{condition_existence_beta}) for each $k\in\{b,s\}$, then it also satisfies (\ref{eqn:condition_unique_mono}) and consequently there exists unique solutions $\vz^*$ and $z^\text{C}$ to (\ref{FOCs_z}) and (\ref{eqn:gumbel_mono_b}), respectively.  Section \ref{sect:collusionvscompetition} will compare these two solutions assuming \eqref{condition_existence_beta}  is satisfied.

\subsection*{Interpretation and Implications of the Resulting Pricing Formulas}\label{subsection_discussionCNE}
We discuss the pricing formulas \eqref{eqn_pz} and \eqref{eqn_pm} 
for the  competition and collusion models. We first relate them to common pricing competition models.
Both formulas are expressed in terms of  the equilibrium normalized net deterministic utility, $z_k$, that platforms provide to users on both sides of the market. They thus remind the formulation in \cite{armstrong2001competitive},  where multiple firms compete in a utility space, instead of a space of prices. When solving (\ref{pi}), firms internalize competition for users in terms of the utility they can provide w.r.t.~(with respect to) the outside utility. The optimal vector utility, $\vz^*$, provided by the competing platforms is determined so that some users are always excluded from the market.\footnote{Note that the equilibrium market share satisfies $x_k^* = \omega(z_k^*)<1/N$ (see \eqref{eqn_pz} and the definition of $\omega(\cdot)$ in Proposition \ref{prop:FOC_gumbel}). This condition excludes the participation of some users.} A similar result is obtained for the colluding case, while excluding a larger portion of participants, as shown below in Proposition \ref{prop:CNEvsCE}. Therefore, our models also imply the standard result that the output is not optimally distributed among users when there is price competition or collusion (see, e.g., \cite{varian1989price}, \cite{armstrong1996multiproduct} and \cite{rochet1998ironing}). 

Our pricing formulas \eqref{eqn_pz} and \eqref{eqn_pm} generalize many of the standard results in the platform's literature 
for the case of an outside option.
We emphasize some of these generalizations: (i) For CE,  the 
term ${\beta_{k}(1+Ne^{z_{k}})^{2}}/{e^{z_{k}}}$,
which appears in the 
diagonal of the matrix (\ref{def:Hm_matrix}), captures the platform's market power (see \cite{perloff1985equilibrium}). It implies that in equilibrium platforms charge users on side $k$ of the market proportionally to the platform's differentiation parameter $\beta_k$ (see \cite{tan2021effects} and  \cite{chica2021exclusive}). 
(ii) For CNE and CE, assume that the within-side externalities are positive (i.e., $\phi_{kk}\geq 0$). Then, from the diagonal of  (\ref{def:H_matrix}) and (\ref{def:Hm_matrix}), platforms subsidize users on side $k$ by an amount that is proportional to the joining population on this side of the market (i.e., they subsidize users on side $k$ with $\phi_{kk}x_k^*$ and $\phi_{kk}x_k^\text{C}$, respectively, for the competing and colluding models). If these externalities are negative (i.e., $\phi_{kk}<0$), the opposite result is true (see \cite{bardey2014competition}). (iii) For CNE and CE, assume positive cross-side externalities, that is, $\phi_{lk}\geq 0$ for each $l,k\in \{b,s\}$, $l\neq k$. Then, the off-diagonal terms of (\ref{def:H_matrix}) and (\ref{def:Hm_matrix}) imply that platforms subsidize users on side $k$ with an amount directly affected by the joining population on the other side of the market (i.e., platforms subsidize an amount $\phi_{lk}x_l^*$ to users on side $k$). 

\section{Competition and Collusion in Two-sided Markets with an Outside Option}\label{sect:collusionvscompetition}

We 
compare the colluding and competing market models
by studying the main properties of and differences between the pricing formulas (\ref{eqn_pz}) and (\ref{eqn_pm}). 
We first assume competition and
characterize the markets in which users receive positive and negative normalized net deterministic utility, $z_k^*$ (see Proposition \ref{prop:sign_zk} and Corollary \ref{coro:sign_zk}). 
We also characterize the sign of $z_k^*$ under perfect competition (i.e., as $N\to\infty$) and show that platforms charge a price that is equal to user's heterogeneity in tastes while covering the entire market (see Corollary \ref{coro:sign_zk_N}). We then study the effects of the outside option on the change of prices, profits and consumer surplus. In particular, we show that when the outside option increases: (i) prices on side $k$ may increase or decrease (see Proposition \ref{prop:dpkduk0}); (ii) profits decrease (see Proposition \ref{prop:dpikduk0}); and (iii) consumer surplus may increase or decrease (see Proposition \ref{prop:dCSkuk0}).
Next,  
we assume collusion and characterize markets in which $z_k^\text{C}$ is positive or negative (see Proposition \ref{prop:sign_zkm} and Corollaries \ref{coro:gammavsgammaC} and \ref{coro:sign_zkm}). Finally,  
we compare the equilibrium quantities of competition and collusion (see Proposition \ref{prop:CNEvsCE}). 

\subsection*{The Sign of the Net Deterministic Utility Under Competition}
\label{sec:sign_z*}
In CNE, a positive (negative) $z_k^*$ implies that the deterministic utility that users enjoy in equilibrium from joining a given platform is larger (smaller) than the deterministic utility of the outside option.  
For this reason, we first study the sign of $z_k^*$ as given by the solution of (\ref{FOCs_z}). 
The following proposition shows sufficient conditions to partition the region described by (\ref{condition_existence_beta}) into two regions: $\{z_k^*<0\}$ and $\{z_k^*>0\}$, which we demonstrate in Figure \ref{fig:zk_sign} for two different values of $u_k^0$. The indifference region $\{z_k^*=0\}$ is described by a curve $\beta_k = \gamma(N,\phi_{kk},u_k^0)$ in the plane $(\phi_{kk},\beta_k)$, where $\gamma$ is defined as follows:
\begin{equation}\label{def:gamma}
    \gamma(N,\phi_{kk},u_k^0):= \frac{\left(2\phi_{kk}-Nu_{k}^{0}\right)+\sqrt{\left(2\phi_{kk}-Nu_{k}^{0}\right)^{2}+4\phi_{kk}\left(u_{k}^{0}-\frac{2\phi_{kk}}{N+1}\right)}}{2\left(N+1\right)}.
\end{equation}
We remark that the clustering of the sign of $z_k^*$ according to this proposition requires a local bound on the cross-side externalities.

\begin{figure}[H]
    \caption{Classification of the sign of $z_k^*$ based on $(\phi_{kk},\beta_k)$, $k\in\{b,s\}$, according to Proposition \ref{prop:sign_zk}, where $N=4$ and $u_k^0 = -1$ (left) or $u_k^0 = 0.5$ (right). The red and blue regions correspond to negative and positive $z_k^*$, respectively.}
     \label{fig:zk_sign}
     \centering
        \includegraphics[width=0.35\linewidth]{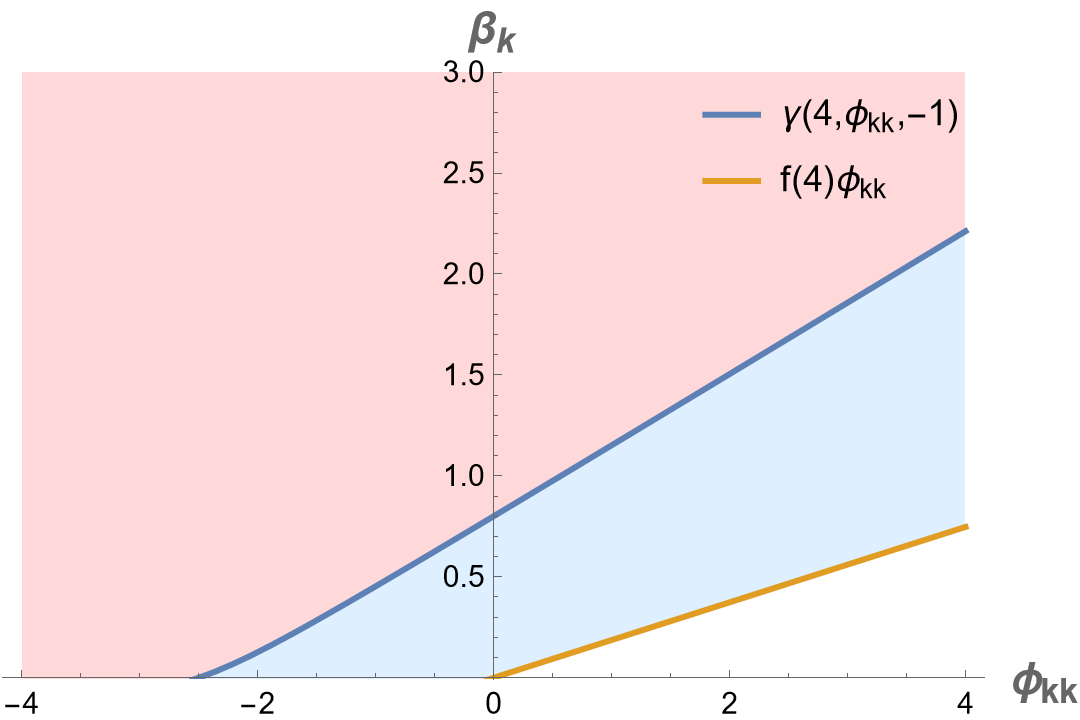}
        \includegraphics[width=0.35\linewidth]{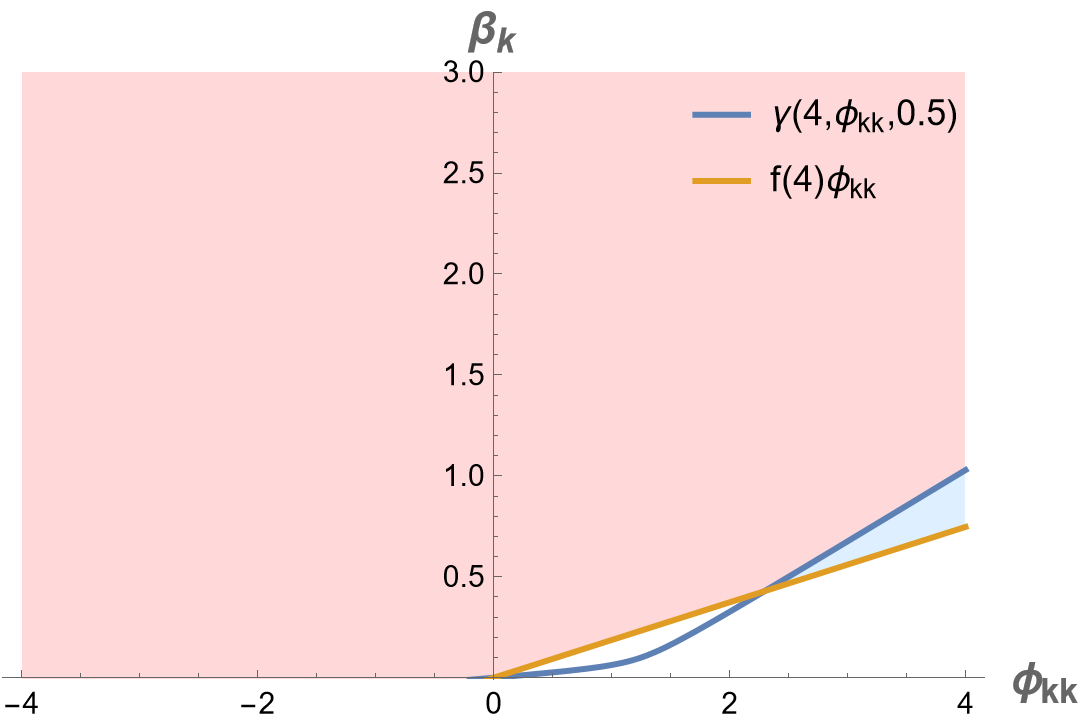}
 \end{figure}

\begin{proposition}[The sign of $z_k^*$]\label{prop:sign_zk} 
Suppose that $N\geq2$ and for each $k\in\left\{ b,s\right\}$, $\left(\phi_{kk},\beta_{k}\right)$ satisfies (\ref{condition_existence_beta}). 
Then we can further partition the domain specified in  (\ref{condition_existence_beta}) into two regions, which cluster the sign of $z_k^*$ as long as a local condition on the cross-side externalities hold:\begin{itemize}
    \item[(i)] If $\beta_{k}>\gamma(N,\phi_{kk},u_k^0)$,  
    then there exists $\varepsilon>0$ such that for any $(\phi_{bs},\phi_{sb})\in B_\varepsilon(0)\subset \mathbb{R}^2$, 
    $z_{k}^{*}<0$. 
    \item[(ii)] $\beta_{k}<\gamma\left(N,\phi_{kk},u_{k}^{0}\right)$, 
    then there exists $\varepsilon>0$ such that for any $(\phi_{bs},\phi_{sb})\in B_\varepsilon(0)\subset \mathbb{R}^2$,  $z_{k}^{*}>0$.
\end{itemize}
\end{proposition}

We first clarify the economic meaning of this proposition. If user's heterogeneity in tastes is large enough (i.e., $\beta_{k}>\gamma(N,\phi_{kk},u_k^0)$), then it is a standard result 
that platforms extract consumer surplus by charging a price that leads to a negative normalized net deterministic utility, i.e., $z_k^*<0$ (see \cite{anderson1992logit}, \cite{tan2021pricing},  \cite{chica2021exclusive} and others).\footnote{This result is due to the fact that highly heterogeneous users are less responsive to price and demand effects.} On the other hand, if the user's heterogeneity in tastes, $\beta_k$, is small enough (i.e., $\beta_{k}<\gamma\left(N,\phi_{kk},u_{k}^{0}\right)$), then users receive positive normalized net deterministic utility, i.e., $z_k^*>0$. 

We identify a critical threshold for the deterministic outside utility so that above this threshold, (ii) in Proposition~\ref{prop:sign_zk} is not feasible. 
\begin{corollary}[The sign of $z_k^*$ for large values of $u_k^0$]\label{coro:sign_zk}
Case (ii) in Proposition~\ref{prop:sign_zk} is not feasible if $u_k^0 \geq \tilde{u}_k^0(N, \phi_{kk})$, where $\tilde{u}_k^0(N, \phi_{kk})$ is the critical threshold for the deterministic outside utility and it is defined in (\ref{def:tilde_uk0}) in the Appendix~\ref{appendixA}.
\end{corollary}

This corollary implies 
that if the deterministic outside utility is sufficiently large (i.e., $u_k^0\geq\tilde{u}_k^0(N,\phi_{kk})$), the CNE leads to a negative net deterministic utility for any size of heterogeneity in user's tastes satisfying (\ref{condition_existence_beta}). In other words, only if the deterministic outside utility is relatively small (i.e., $u_k^0<\tilde{u}_k^0(N,\phi_{kk})$), users with relatively weak preferences (i.e., $\beta_k<\gamma(N,\phi_{kk},u_k^0)$) receive positive net deterministic utility.

Next, we show that in the case of \textit{perfect competition} (i.e., the limiting case $N\to\infty$), the sign of $z_k^*$ can be characterized by the sign of $u_k^0$ and the size of $\beta_k$.

\begin{corollary}[CNE under perfect competition]\label{coro:sign_zk_N}
For each $k\in\{b,s\}$, any $u_k^0\in \sR$, $\Phi\in \mathbb{R}^{2x2}$ and $\beta_k>0$, under perfect competition (i.e., when $N\to \infty$), 
\begin{equation}
    \label{gamma_limit}
    \lim_{N\to \infty}z_k^* \begin{cases}
        
        > 0, & \textnormal{ if } u_k^0< 0 \textnormal{ and } \ \beta_k < -u_k^0;\\
        < 0, & \textnormal{ if } (u_k^0< 0 \textnormal{ and } \ \beta_k > -u_k^0) \textnormal{ or }\  u_k^0\geq 0.
\end{cases}
\end{equation}

Moreover, as $N\to\infty$, $p_k^*\to \beta_k$, $x_k^*\to 0$ and $Nx_k^*\to 1$.
\end{corollary}
Under perfect competition, platforms charge a price that is equal to the user's heterogeneity in tastes for that side of the market, i.e., $p_k^* = \beta_k$. Moreover, the equilibrium market participation on side $k$ of the market is complete, i.e., $Nx_k^*= 1$. When the deterministic outside option utility is positive, users receive negative normalized net deterministic utility for any size of $\beta_k$ under perfect competition. When the deterministic outside option utility is strictly negative, users receive positive normalized net deterministic utility if and only if the heterogeneity of users' taste is small, i.e. $\beta_k < -u_k^0$.
We demonstrate (\ref{gamma_limit}) in Figure \ref{fig:zk_sign2} for two different values (negative and positive) of $u_k^0$ and a sufficiently large $N$.

 \begin{figure}[H]
    \caption{Classification of the sign of $z_k^*$ based on $(\phi_{kk},\beta_k)$, $k\in\{b,s\}$, for large $N$ according to Corollary \ref{coro:sign_zk_N}, where $N=200$ and $u_k^0 = -1$ (left) or $u_k^0 = 1$ (right). The red and blue regions correspond to negative and positive $z_k^*$, respectively.}
     \label{fig:zk_sign2}
                \centering
        \includegraphics[width=0.35\linewidth]{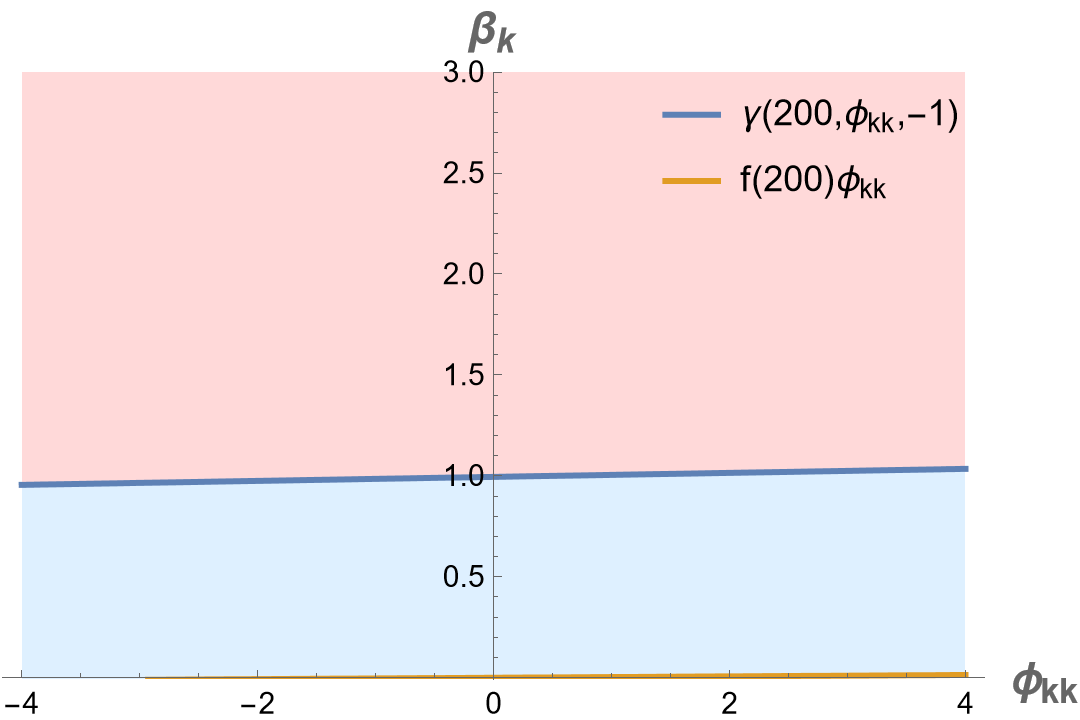}
\includegraphics[width=0.35\linewidth]{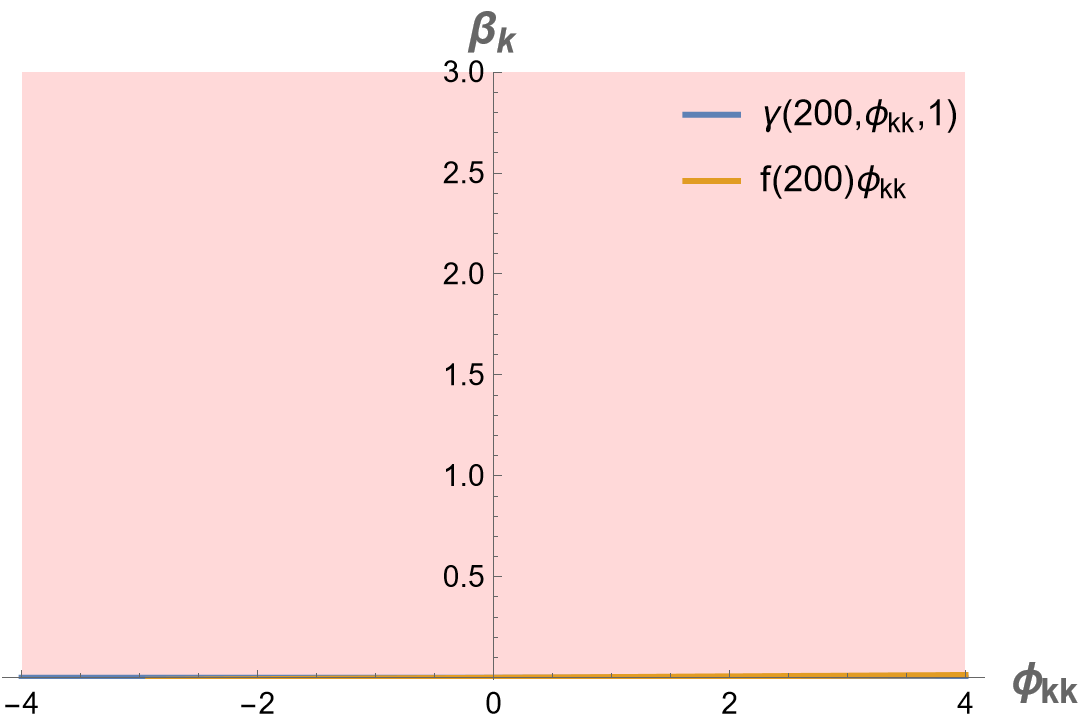}
 \end{figure}

\subsection*{The Effects of the Outside Option on the CNE}

The following proposition provides sufficient conditions to characterize the sign of $\partial p_k^*/\partial u_k^0$. It shows that the effect of the outside option on $p_k^*$ is nonlinear. It uses the following quantities:
\begin{equation}\label{def:fuk0}
\begin{split}
    g_{p,u}(N)&:=\frac{\left(N+\sqrt{\left(N-1\right)\left(N+3\right)}+1\right)}{2N}\ \textnormal{ and }\\
     f_{p,u}(N)&:=\frac{1}{2} \left(\sqrt{\frac{N-2}{N}}+1\right).
\end{split}
\end{equation}

\begin{proposition}(The sign of $\partial p_k^*/\partial u_k^0$)\label{prop:dpkduk0}
Suppose that $N\geq2$ and for each $k\in\left\{ b,s\right\}$, $\left(\phi_{kk},\beta_{k}\right)$ satisfies (\ref{condition_existence_beta}). 
Then we can further partition the domain specified in  (\ref{condition_existence_beta}) into two regions, which cluster the sign of $\partial p_k^*/\partial u_k^0$ as long as a local condition on the cross-side externalities hold:\begin{itemize}
    \item[(i)] If either $(\phi_{kk}\leq0$ and $\beta_{k}>0)$ or $(\phi_{kk}>0$ and $\beta_{k}>g_{p,u}\left(N\right)\phi_{kk})$, 
    then there exists $\varepsilon>0$ such that for any $(\phi_{bs},\phi_{sb})\in B_\varepsilon(0)\subset \mathbb{R}^2$, 
    $\partial p_k^*/\partial u_k^0<0$. 
    \item[(ii)] If $\phi_{kk}>0$, $N\geq 3$ and $f\left(N\right)\phi_{kk}<\beta_{k}<f_{p,u}\left(N\right)\phi_{kk}$, 
    then there exists $\varepsilon>0$ such that for any $(\phi_{bs},\phi_{sb})\in B_\varepsilon(0)\subset \mathbb{R}^2$, 
    $\partial p_k^*/\partial u_k^0>0$.
\end{itemize}
Moreover, if $(\phi_{bs},\phi_{sb})=0$, then
\begin{equation}\label{limitspkuk0}
    \begin{split}
\lim_{u_{k}^{0}\to-\infty}p_{k}^{*}&=\frac{N}{N-1}\beta_{k}-\frac{\phi_{kk}}{N-1}=: p_{k,u}, \textnormal{ and}\\
\lim_{u_{k}^{0}\to\infty}p_{k}^{*}&=\beta_{k}=: p_{k,E}.
\end{split}
\end{equation}
Consequently, $p_k^*\in (p_{k,E} , p_{k,u})$ in case (i), and $p_k^*\in (p_{k,u}, p_{k,E})$ in case (ii). 
\end{proposition}

We clarify the economic meaning of this proposition. We first note it implies $p_k^*\to p_{k,u}$ when $u_{k}^{0}\to-\infty$, which coincides with the equilibrium price in a platform competition model with no outside option. It also implies $p_k^*\to p_{k,E}=\beta_k$ when  $u_{k}^{0}\to\infty$, which represents the efficient price, that is, the price, $\beta_k$, under perfect competition, expressed in Corollary \ref{coro:sign_zk_N}. 
In part (i), the incorporation of an outside option into the platform competition model decreases the equilibrium price w.r.t.~the no outside option model, which is an expected result. Thus, if the cross-side externalities are sufficiently small and the within-side externalities are either negative or positive with relatively large user's heterogeneity in tastes, then users are compensated by an amount equal to $p_{k,u}-p_k^*$. Moreover, in this case, users always pay a price that is bigger than the efficient price, $p_{k,E}$.   
On the other hand, in part (ii), incorporating an outside option increases the equilibrium price w.r.t~the no outside option model, which is non-trivial. Therefore, under sufficiently small cross-side externalities and positive within-side externalities, users with relatively small heterogeneity in tastes pay a premium w.r.t~the model with no outside option, which is quantified by $p_k^*-p_{k,u}$. Moreover, users always pay a price that is smaller than the efficient price $p_{k,E}$. 

\begin{remark}[Price overestimation vs.~underestimation]
If a given population can be parameterized using the region of parameters described by either (i) or (ii) of Proposition~\ref{prop:dpkduk0}, then a model of platform competition that omits the outside option will either overestimate or underestimate, respectively, the true equilibrium price. 
\end{remark}

The following proposition shows sufficient conditions to determine the sign of $\partial \pi_k^*/\partial u_k^0$. It uses the following quantity:
\begin{equation}
    \label{def:g_pi_u}
    g_{\pi,u}(N):=\sqrt{\frac{N-1}{N^3}}+\frac{1}{N}. 
\end{equation}

\begin{proposition}(The sign of $\partial \pi_k^*/\partial u_k^0$)\label{prop:dpikduk0}
If $N\geq 2$ and 
$$\textnormal{either } (\phi_{kk}\leq0 \textnormal{ and } \beta_{k}>0) \textnormal{ or } (\phi_{kk}>0 \textnormal{ and } \beta_{k}>g_{\pi,u}\left(N\right)\phi_{kk}),$$ 
then there exists $\varepsilon>0$ such that for any $(\phi_{bs},\phi_{sb})\in B_\varepsilon(0)\subset \mathbb{R}^2$,  $\partial \pi_k^*/\partial u_k^0<0$. Moreover, if $(\phi_{bs},\phi_{sb})=0$, then
\begin{equation}\label{limitspikuk0}
    \begin{split}
\lim_{u_{k}^{0}\to-\infty}\pi_{k}^{*}&=\frac{\beta_{k}}{N-1}-\frac{\phi_{kk}}{\left(N-1\right)N}=: \pi_{k,u}, \textnormal{ and}\\
\lim_{u_{k}^{0}\to\infty}\pi_{k}^{*}&=0=: \pi_{k,E}. 
\end{split}
\end{equation}
\end{proposition}

Note that even though, by part (ii) of Proposition \ref{prop:dpkduk0}, prices may increase with the outside option utility, Proposition \ref{prop:dpikduk0} shows that profits are always decreasing w.r.t.~$u_k^0$. This happens because market participation is always decreasing w.r.t.~$u_k^0$. Therefore, it is not a surprise that profits are decreasing as a function of $u_k^0$.

The following proposition provides sufficient conditions to determine the sign of the derivative of the consumer surplus w.r.t.~the outside option utility. More specifically, it uses the 
equilibrium consumer surplus on side $k$ of the market, $CS_k^*$, which is 
defined as follows (see \cite{tan2021effects} for the case without an outside option):
\begin{equation}
\label{eqn:equilibrium_CS}
CS_{k}^{*}:=\mathbb{E}\left[\max_{i=0,...,N}\epsilon_{k}^{i}\right]-p_{k}^{*}+\phi_{k}\left(x_{b}^{*},x_{s}^{*}\right),
\end{equation}
where  $\mathbb{E}\left[\max_{i=0,...,N}\epsilon_{k}^{i}\right]$ is the expected maximum idiosyncratic utility.\footnote{Note that $\max_{i=0,...,N}\{ \epsilon_{k}^{i}\} \sim G(\mu_{k}+\beta_{k}\ln\left(N+1\right),\beta_{k})$ and thus  
$\mathbb{E}[\max_{i=0,...,N}\epsilon_{k}^{i}]=[\mu_{k}+\beta_{k}\ln\left(N+1\right)]+\beta_{k}\gamma$,
where $\gamma$ denotes the Euler-Mascheroni constant. This quantity captures the product variety of the market.} It also uses the 
function $f_{CS,u}(N)$ given in Appendix \ref{appendixA} (see \eqref{def:fCSu}).

\begin{proposition}[The sign of $\partial (CS_{k}^{*})/\partial u_k^0$]\label{prop:dCSkuk0}
The effect of the outside option on the change of consumer surplus can be clustered into the following two regions: 
\begin{itemize}
    \item[(i)] If $N\geq2$ and either $(\phi_{kk}\leq 0 \textnormal{ and } \beta_k>0) \textnormal{ or } (\phi_{kk}> 0 \textnormal{ and } \beta_k>2\phi_{kk})$, then there exists $\epsilon>0$ such that for any $(\phi_{bs},\phi_{sb})\in B_{\epsilon}(0)$, $\partial (CS_{k}^{*})/\partial u_k^0>0$.
\item[(ii)] If $N\geq 2$, $\phi_{kk}>0$, $f\left(N\right)\phi_{kk}<\beta_{k}< f_{CS,u}\left(N\right)\phi_{kk}$, then there exists $\epsilon>0$ such that for any $(\phi_{bs},\phi_{sb})\in B_{\epsilon}(0)$, $\partial (CS_{k}^{*})/\partial u_k^0<0$.
\end{itemize}
\end{proposition}

Part (i) of this proposition implies that the incorporation of an outside option into the platform competition model may increase the consumer surplus, or equivalently, the consumer welfare, w.r.t.~the no outside option model. On the other hand, 
Part (ii) implies that incorporating an outside option may decrease the equilibrium consumer surplus w.r.t~the no outside option model. While part (i) is standard, part (ii) is surprising.

\subsection*{The Sign of the Net Deterministic Utility Under Collusion}
\label{sec:sign_zC}

The following proposition quantifies the sign of $z_k^\text{C}$ in the collusion case of (\ref{pim}). In particular, it claims that the indifference region $\{z_k^\text{C}=0\}$ is described by a curve $\beta_k = \gamma^\text{C}(N,\phi_{kk},u_k^0)$ in the plane $(\phi_{kk},\beta_k)$, where $\gamma^\text{C}$ is defined as follows:
\begin{equation}\label{def:gammaC}
    \gamma^\text{C}(N,\phi_{kk},u_k^0):=\frac{2\phi_{kk}-u_{k}^{0}\left(N+1\right)}{\left(N+1\right)^{2}}.
\end{equation}

\begin{proposition}[The sign of $z_k^\text{C}$]\label{prop:sign_zkm} 
Suppose that $N\geq2$ and for each $k\in\left\{ b,s\right\}$, $\left(\phi_{kk},\beta_{k}\right)$ satisfies \eqref{eqn:condition_unique_mono}. 
Then we can further partition the domain specified in  \eqref{eqn:condition_unique_mono} into two regions, which cluster the sign of $z_k^\text{C}$ as long as a local condition on the cross-side externalities hold:\begin{itemize}
    \item[(i)] If $\beta_{k}>\gamma^\text{C}(N,\phi_{kk},u_k^0)$,
    then there exists $\varepsilon>0$ such that for any $(\phi_{bs},\phi_{sb})\in B_\varepsilon(0)\subset \mathbb{R}^2$, 
    $z_{k}^\text{C}<0$. 
    \item[(ii)] $\beta_{k}<\gamma^\text{C}\left(N,\phi_{kk},u_{k}^{0}\right)$, 
    then there exists $\varepsilon>0$ such that for any $(\phi_{bs},\phi_{sb})\in B_\varepsilon(0)\subset \mathbb{R}^2$,  $z_{k}^\text{C}>0$.
\end{itemize}
\end{proposition}

The interpretation of Proposition \ref{prop:sign_zkm} is very similar to that of Proposition \ref{prop:sign_zk}.
When the user's heterogeneity in tastes is small (i.e., $\beta_k < \gamma^\text{C}(N, \phi_{kk}, u_k^0)$), then users receive $z_k^\text{C}> 0$. On the other hand, if $\beta_k$ is large (i.e., $\beta_{k}>\gamma^\text{C}(N,\phi_{kk},u_k^0)$), users receive $z_k^\text{C}<0$. Figure \ref{fig:zk_signC} demonstrates the regions described in Proposition \ref{prop:sign_zkm} for two different values of $u_k^0$.

 \begin{figure}[H]
    \caption{Classification of the sign of $z_k^\text{C}$ based on $(\phi_{kk},\beta_k)$, $k\in\{b,s\}$, according to Proposition \ref{prop:sign_zkm}, where $N=4$ and $u_k^0 = -1$ (left) or $u_k^0 = 0.5$ (right). The red and blue regions correspond to negative and positive $z_k^\text{C}$, respectively.}
     \label{fig:zk_signC}
                \centering
        \includegraphics[width=0.35\linewidth]{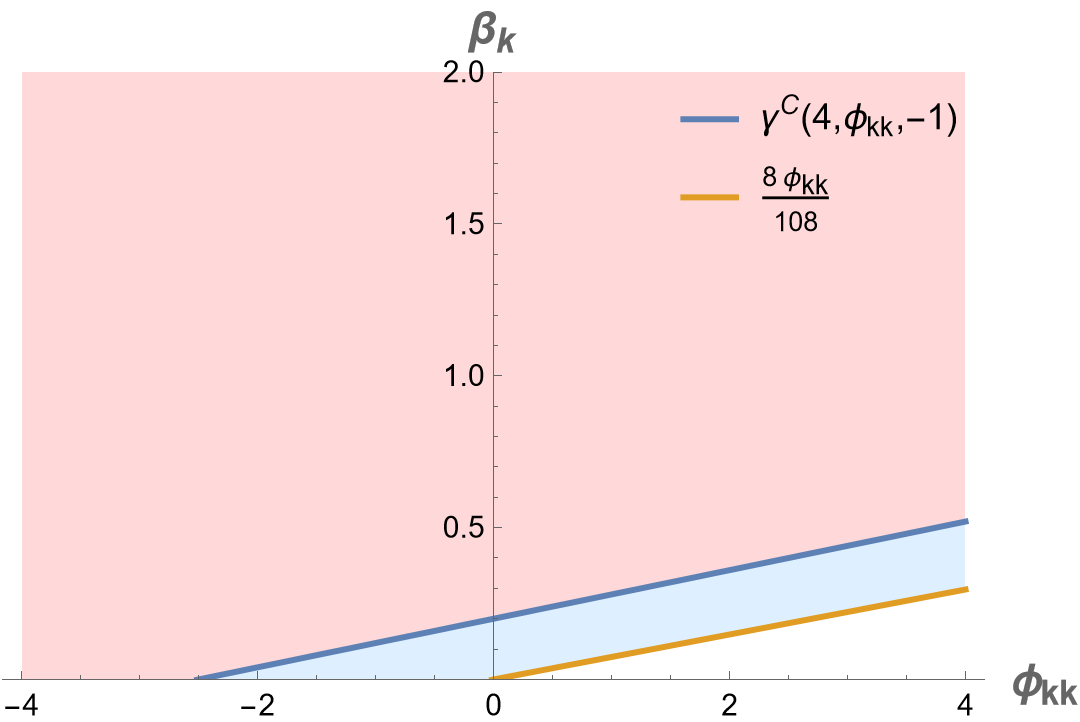}
\includegraphics[width=0.35\linewidth]{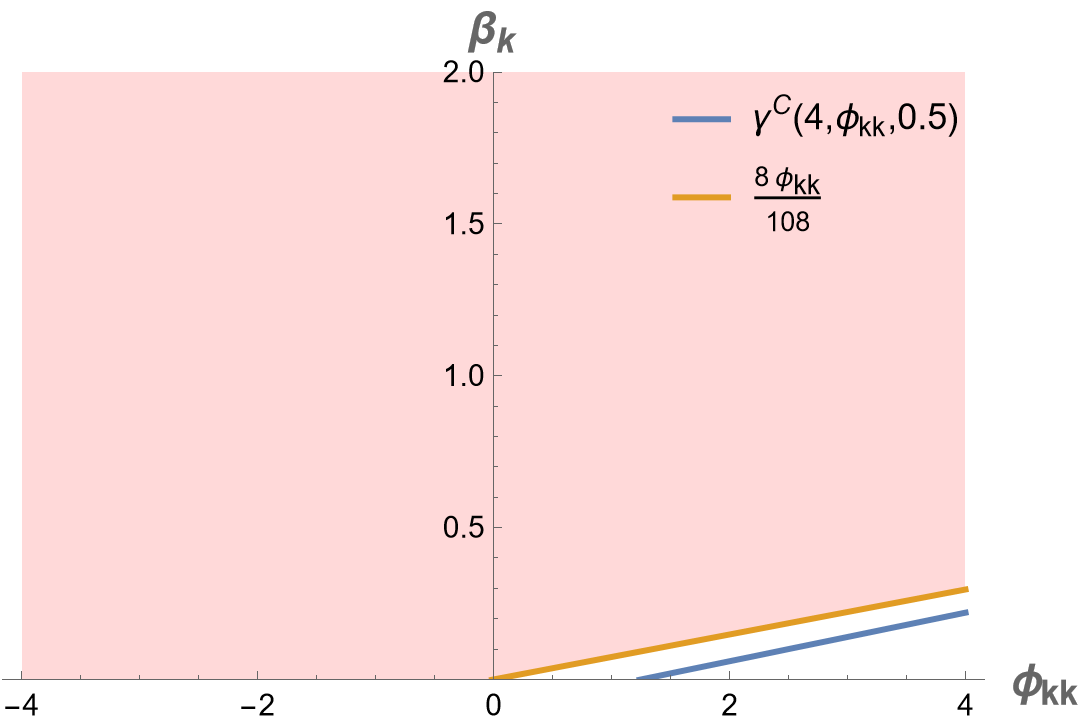}
 \end{figure}

\begin{corollary}[$\gamma^C(N, \phi_{kk}, u_k^0)$ vs $\gamma(N, \phi_{kk}, u_k^0)$]\label{coro:gammavsgammaC}
   If $N\geq 2$ and $\gamma(N, \phi_{kk}, u_k^0)\geq 0$, then $\gamma(N, \phi_{kk}, u_k^0)\geq \gamma^C(N, \phi_{kk}, u_k^0)$.
\end{corollary}

By Corollary \ref{coro:gammavsgammaC}, in order to have a positive normalized net deterministic utility $z_k$ in CE, the size of the user's heterogeneity in tastes must be smaller than in CNE. Moreover, in CE, we also identify a critical threshold for the outside utility such that above this threshold, the condition of (ii) in Proposition~\ref{prop:sign_zkm} is not feasible. 

\begin{corollary}[The sign of $z_k^\text{C}$ for large values of $u_k^0$]\label{coro:sign_zkm}
Case (ii) in Proposition~\ref{prop:sign_zkm} is not feasible if $u_k^0 \geq \tilde{u}_k^\text{C}(N, \phi_{kk})$, where $\tilde{u}_k^\text{C}(N, \phi_{kk})$ is the critical threshold for the outside utility and it is defined in (\ref{def:tilde_uk0m}) in the Appendix~\ref{appendixA}.
\end{corollary}

\subsection*{Economic Outputs in Competitive vs. Collusive Markets}
\label{sec:compare_z*_zC}

The following proposition compares the normalized net deterministic utility, market participation and prices in competitive and collusive markets.

\begin{proposition}[Competition vs Collusion Outputs]\label{prop:CNEvsCE}
Suppose that $N\geq 2$ and for each $k\in\{b,s\}$, $(\phi_{kk},\beta_k)$ satisfies \eqref{condition_existence_beta}. Then, there exists $\epsilon>0$ such that for any $\varphi_{1}=\left(\phi_{bs},\phi_{sb}\right)\in B_{\epsilon}\left(0\right)\subset\mathbb{R}^{2}$, in equilibrium:
\begin{itemize}
    \item[(i)] the normalized net deterministic utility for users on side $k$ is bigger under competition than under collusion (i.e., $z_{k}^{*}>z_{k}^\text{C}$);
    \item[(ii)] the market participation is bigger under competition than under collusion (i.e., $Nx_{k}^{*}>Nx_{k}^\text{C}$);
    \item[(iii)] the price charged on side $k$ of the market is smaller under competition than under collusion (i.e., $p_{k}^{*}<p_{k}^\text{C}$).
\end{itemize}
\end{proposition}

Part (i) of the above proposition agrees with the standard collusion literature (see, e.g., \cite{bishop1960duopoly}, \cite{varian1989price}, \cite{brander1985tacit} among others) in which users receive the lowest normalized net deterministic utility under collusion. 
Part (ii) is a direct corollary of part (i). 
Indeed, Proposition \ref{prop:FOC_gumbel} implies that $\omega(\cdot)$ is monotonically increasing.  Thus, combining \eqref{eqn_pz}, \eqref{eqn_pm} and part (i) of Proposition \ref{prop:CNEvsCE} leads to part (ii) as follows: $Nx_{k}^{*} = N\omega(z_k^*)>N\omega(z_k^\text{C})=Nx_{k}^\text{C}$.  
To explain part (iii), we use (\ref{def:zk}) and decompose the difference between collusion and competition prices as follows: 
\begin{equation}
\label{price_decomposition}
       \vp^\text{C}-\vp^* = \Phi(\vx^\text{C} - \vx^*) + \vbeta (\vz^\ast -\vz^\text{C} ).
\end{equation}
competition. 
A careful combination of this formula with parts (i) and (ii) of Proposition \ref{prop:CNEvsCE}, the assumption $\left(\phi_{bs},\phi_{sb}\right)\in B_{\epsilon}\left(0\right)$ (or for simplicity $(\phi_{bs}, \phi_{sb})=0$) and the observation that if $\phi_{kk}\geq 0$ for $k\in \{b,s\}$ then $ \beta_k > f(N)\phi_{kk}$ (see \eqref{condition_existence_beta}) leads to part (iii). We thus note that the above detailed analysis regarding the deterministic net utility is valuable for deriving broader economic implications.

Our results can be compared to other ones on platform collusion. \cite{dewenter2011semi} study collusion and competition following the idiosyncrasies of a newspapers market with two firms. They find that for small cross-side network externalities the collusive price is higher than the competitive price. We generalize this result by incorporating the outside option utility, $u_k^0$, the within-side network externalities, $\phi_{kk}$, and by having $N$ horizontally differentiated platforms.  
Part (iii) of our result also has the same conclusion as \cite{cohen2022competition}, who in the context of ride-sharing services (e.g., \textit{Uber} and \textit{Lyft}), show that under collusion, riders pay a larger price and workers receive a lower wage than under competition (note that the wage in their model is a negative price in our model). 
Nevertheless, \cite{cohen2022competition} assume a different model for the user's utility function, which is tailored for their specific setting of prices and wages.

\section{The Effects of Increasing Competition on the CNE}\label{sect:effectsCompetition}
We study how increasing competition (i.e., increasing $N$) affects four CNE quantities: 
price, market participation, consumer surplus, and profit. We first establish sufficient conditions for the derivative $\partial p_{k}^{*}/\partial N$ to be either positive or negative (see Proposition \ref{prop:dpk_dN}). We thus specify regions where competition can lead to increasing or decreasing prices. We also establish sufficient conditions for $\partial (N x_{k}^{*})/\partial N$ to be positive and consequently for increasing market participation under competition (see Proposition \ref{prop:dNxk_dN}). We further formulate sufficient conditions to have increasing and decreasing consumer surplus, i.e., to have positive and negative $\partial (CS_{k}^{*})/\partial N$ (see Proposition \ref{prop:dCSk_dN_positive}). Finally, we establish sufficient conditions for the derivative $\partial \pi_{k}^{*}/\partial N$ to be either positive or negative  (see Proposition \ref{prop:dpik_dN_negative}). That is, we specify regions where  competition can lead to increasing or decreasing profits on side $k$ of the market.


\textbf{The effect of competition on prices.} We study the sign of $\partial p_{k}^{*}/\partial N$. We first clarify the difficulty in estimating the latter derivative.
In view of (\ref{eqn_pz}), the equilibrium vector price is $\boldsymbol{p}^*=H(\vz^*)\Omega(\vz^*)$. As shown in \eqref{def:H_matrix} and \eqref{eq:x=Omegaz}, the matrix $H$ and the vector $\Omega$ directly depend on $N$. However,  \eqref{FOCs_z} and the definitions of $H$ and $\Omega$ imply that $\vz^*$ is an implicit function of $N$. 
It is thus hard to determine the sign of $\partial p_{k}^{*}/\partial N$. Nevertheless, when the cross-side externalities are sufficiently small, the following proposition establishes sufficient conditions to determine the sign of $\partial p_k^*/\partial N$. It uses the functions $g_{p}(N)$ and $f_{p}\left(N\right)$ defined in \eqref{def:rp1} and \eqref{def:rp2}, respectively, of Appendix \ref{appendixA}. We note that $g_{p}(N)$ and $f_{p}\left(N\right)$ approach 0 and 1, respectively, as $N\to\infty$.

\begin{proposition}[Regions where competition decreases/increases prices]\label{prop:dpk_dN}
The effect of competition on the change of prices can be clustered into the following two regions: 
\begin{itemize}
    \item[(i)]
        Assume that  $N \geq 2$ and either $(\phi_{kk}\leq 0$ and  $\beta_k>g_{p}(N)\phi_{kk})$  or  $(\phi_{kk}>0$  and  $\beta_k>\phi_{kk})$. Then, there exists $\epsilon>0$ such that for any $(\phi_{bs},\phi_{sb})\in B_{\epsilon}(0)$, $\partial p_{k}^{*}/\partial N<0$.
    
    \item[(ii)] Assume that $N\geq3$, $\phi_{kk}>0$ and $f\left(N\right)\phi_{kk}<\beta_{k}<f_p\left(N\right)\phi_{kk}$. Then, there exists $\epsilon>0$ such that for any $(\phi_{bs},\phi_{sb})\in B_{\epsilon}(0)$, $\partial p_{k}^{*}/\partial N>0$.

\end{itemize}

\end{proposition}



The first part of Proposition  \ref{prop:dpk_dN} 
agrees with traditional results, where 
a sufficiently large user's heterogeneity in tastes implies the 
decrease of the equilibrium prices with the increase of competition, i.e., $\partial p_{k}^{*}/\partial N<0$ (see \cite{anderson1992logit}).
On the other hand, the second part of Proposition  \ref{prop:dpk_dN} agrees with a recent and less conventional result, where positivity of  the within-side externalities, $\phi_{kk}$, and sufficiently small user's heterogeneity in tastes, $\beta_k$ imply the increase of prices with the increase of competition, i.e.,  $\partial p_{k}^{*}/\partial N>0$ 
(see, \cite{tan2021effects}). 
The proposition carefully quantifies the thresholds on the user's heterogeneity in tastes that yield different signs of $\partial p_{k}^{*}/\partial N$. The resulting regions are demonstrated below in Figure \ref{fig:dpkdN} for $N=4$. The two regions are subsets of the region specified in \eqref{condition_existence_beta} and when $N \to \infty$ the union of the former two regions approaches the latter region (because $g_p(N) \to 0$ and $f_{p}(N) \to 1$ as $N\to\infty$).


 \begin{figure}[H]
    \centering
    \caption{Classification of the sign of $\partial p_{k}^{*}/\partial N$ based on $(\phi_{kk},\beta_k)$, $k\in\{b,s\}$, according to Proposition \ref{prop:dpk_dN}, where $N=4$. The red and blue regions correspond to negative and positive $\partial p_{k}^{*}/\partial N$, respectively.}
     \label{fig:dpkdN}
    \includegraphics[scale=0.35]{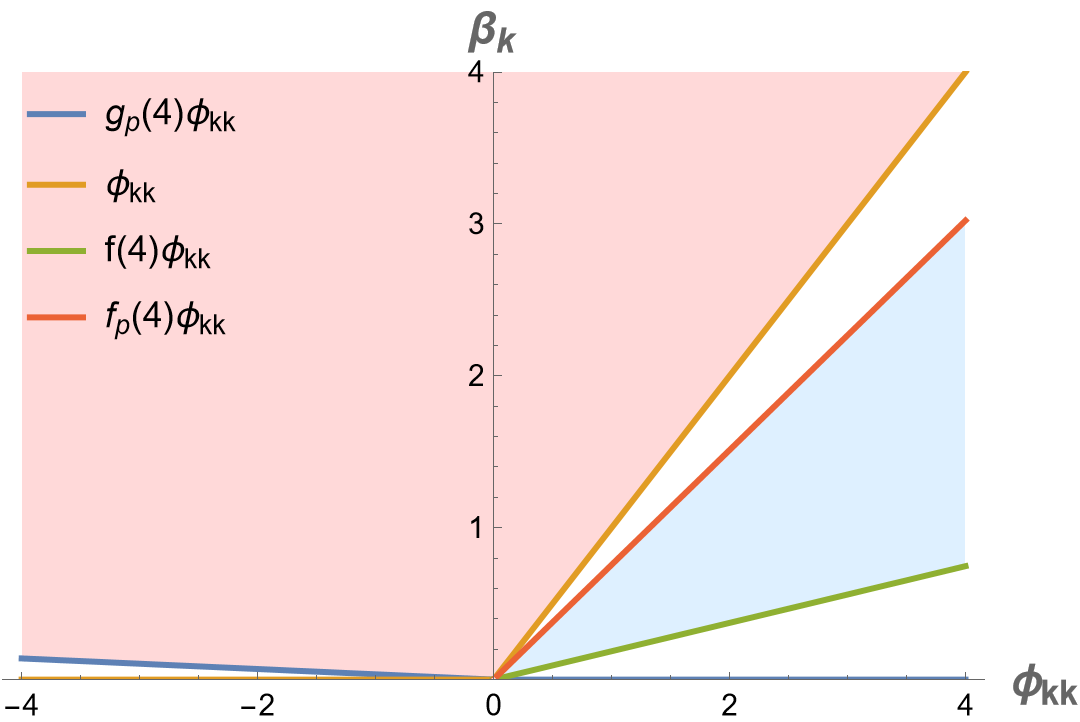}
\end{figure}

\textbf{The effect of competition on market participation.} The equilibrium market participation on side $k$ of the market is given by $Nx_k^*$, where $x_k^*$ is given by \eqref{eq:x=Omegaz}. The following proposition provides sufficient conditions for positive  $\partial (Nx_{k}^{*})/\partial N$. It uses the following quantity:
\begin{equation}\label{def:rNx}
    g_{x}(N):=\frac{\left(2N^{2}-2N+1\right)}{N\left(N^{2}-N+1\right)}.
\end{equation}

\begin{proposition}[Competition increases market participation]\label{prop:dNxk_dN}
If $N\geq 2$ and 
$$
\textnormal{either } (\phi_{kk}\leq 0 \textnormal{ and } \beta_k>0) \textnormal{ or } (\phi_{kk}>0 \textnormal{ and } \beta_k>g_{x}(N)\phi_{kk}),$$ 
 then there exists  $\epsilon>0$  such that for any $(\phi_{bs},\phi_{sb})\in B_{\epsilon}(0)$, $\partial (Nx_{k}^{*})/\partial N>0$.
\end{proposition}

Most models of platform competition leave out the analysis of the outside utility option. By doing so, they assume full market coverage,\footnote{Corollary \ref{coro:sign_zk_N} shows that full market coverage occurs when $N \to \infty$, however, when the number of platforms is finite, we find this assumption unrealistic.} and thus cannot study the effect of competition on market participation. Proposition  \ref{prop:dNxk_dN} fills this gap and its region 
of positive $\partial (Nx_{k}^{*})/\partial N$ is demonstrated below in Figure \ref{fig:dNxkdN} when $N=4$.
We note that the region described by Proposition \ref{prop:dpk_dN} part (ii) intersects with the region described by Proposition  \ref{prop:dNxk_dN}. Thus, when the within-side externalities are sufficiently large (relative to the user’s heterogeneity in tastes) 
then both prices and market participation increase with competition. At last, we note that the region described by Proposition  \ref{prop:dNxk_dN} is a subset of the region described by \eqref{condition_existence_beta} and they coincide as $N \to \infty$.


\begin{figure}[H]
    \centering
    \caption{Demonstration of the region of positive $\partial (Nx_{k}^{*})/\partial N$  based on $(\phi_{kk},\beta_k)$, $k\in\{b,s\}$ (in blue), according to Proposition \ref{prop:dNxk_dN}, where $N=4$. }
     \label{fig:dNxkdN}
    \includegraphics[scale=0.35]{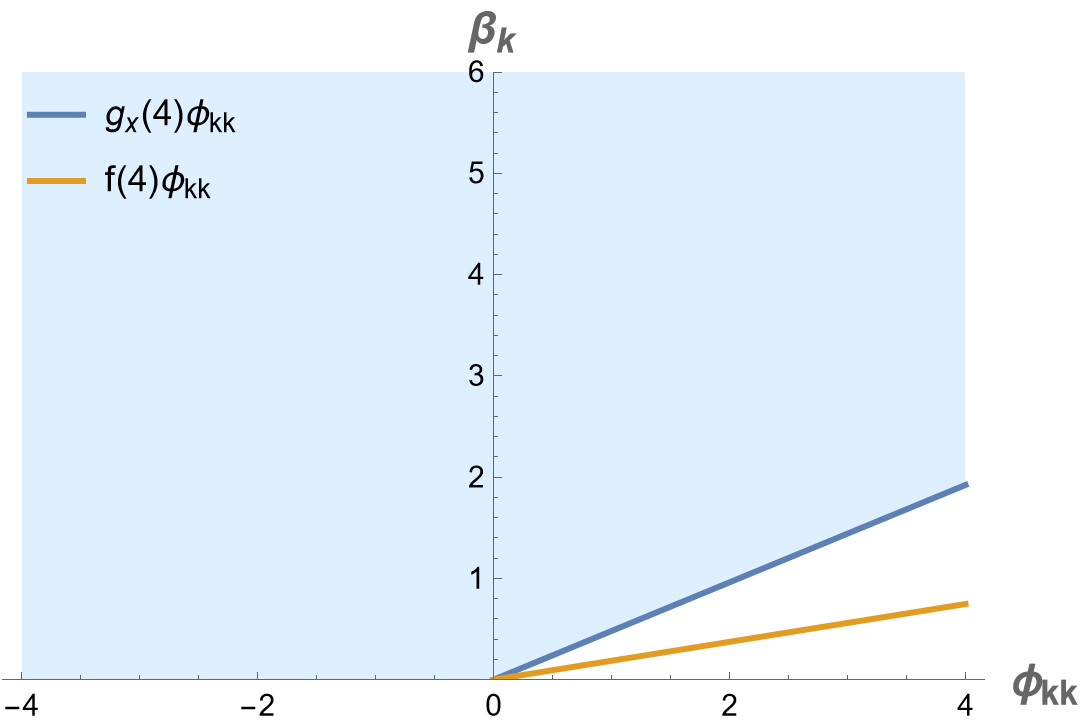}
\end{figure}

\textbf{The effect of competition on consumer surplus.} 
The equilibrium consumer surplus on side $k$ of the market, $CS_k^*$, is defined above in \eqref{eqn:equilibrium_CS}.  The following proposition provides sufficient conditions to determine the sign of $\partial (CS_{k}^{*})/\partial N$. It uses the following quantities:
    \begin{equation}
       \begin{split}\label{def:rcsk1_rcsk2}
           g_{CS}\left(N\right)&:=\frac{2N^{3}-N+1}{N^{2}\left(N^{2}-N+2\right)}
       \end{split}
    \end{equation}
and $f_{CS}\left(N\right)$ which is given by \eqref{def:f_CS} in Appendix \ref{appendixA}.

\begin{proposition}[Regions where competition decreases/increases consumer surplus]\label{prop:dCSk_dN_positive}
The effect of competition on the change of consumer surplus can be clustered into the following two regions: 
\begin{itemize}
    \item[(i)] If $N\geq2$ and either $(\phi_{kk}\leq 0 \textnormal{ and } \beta_k>0) \textnormal{ or } (\phi_{kk}> 0 \textnormal{ and } \beta_k>g_{CS}\left(N\right)\phi_{kk})$, then there exists $\epsilon>0$ such that for any $(\phi_{bs},\phi_{sb})\in B_{\epsilon}(0)$, $\partial (CS_{k}^{*})/\partial N>0$.
\item[(ii)] If $N\geq 7$, $\phi_{kk}>0$, $f\left(N\right)\phi_{kk}<\beta_{k}<\min\left\{ f_{CS}\left(N\right)\phi_{kk},\gamma\left(N,\phi_{kk},u_{k}^{0}\right)\right\}$ and $z_{k}^{*}<\frac{1}{5}\ln 2$, then there exists $\epsilon>0$ such that for any $(\phi_{bs},\phi_{sb})\in B_{\epsilon}(0)$, $\partial (CS_{k}^{*})/\partial N<0$.
\end{itemize}

\end{proposition}

Part (i) of Proposition  \ref{prop:dCSk_dN_positive} 
agrees with traditional results, where consumer surplus increases with increased competition. For example, \cite{hsu2005welfare} consider the Bertrand competition model with substitute goods and show that competition increases consumer surplus. The region in part (i) of Proposition  \ref{prop:dCSk_dN_positive} has small cross-side externalities and its within-side externalities are either negative or positive and small with respect to the user's taste heterogeneity. Part (ii), on the other hand, shows sufficient conditions for decreasing consumer surplus with increased competition. This result agrees with a result from \cite{tan2021effects}, where in markets that are relatively concentrated with a few platforms, consumer surplus decreases as competition increases. Moreover, in the asymptotic regime as $N$ goes to infinity, the region in part (ii) disappears (because $g_{CS}(N)\to 0$ and $f_{CS}(N)\to 0$ as $N\to \infty$) and such behavior is also observed in \cite{tan2021effects}. Note that the region in  Part (ii) of Proposition \ref{prop:dCSk_dN_positive} has positive within-side externalities, small user's heterogeneity in tastes relative to the within-side externalities, positive but small normalized net deterministic utility relative to the number of platforms, and small cross-side externalities. Figure \ref{fig:dCSkdN} demonstrates the resulting regions (i) and (ii) when $N=4$, while excluding the condition involving $z_k^*$.


\begin{figure}[H]
    \centering
    \caption{ Classification of the sign of $\partial (CS_{k}^{*})/\partial N$ based on $(\phi_{kk},\beta_k)$, $k\in\{b,s\}$, for $N=4$ and $u_k^0 = 0$. The blue and red regions correspond to positive and negative $\partial (CS_{k}^{*})/\partial N$, respectively. For the red region we did not include the bound on $z_k^*$, but we still demonstrate a restricted region.}
     \label{fig:dCSkdN}
    \includegraphics[scale=0.35]{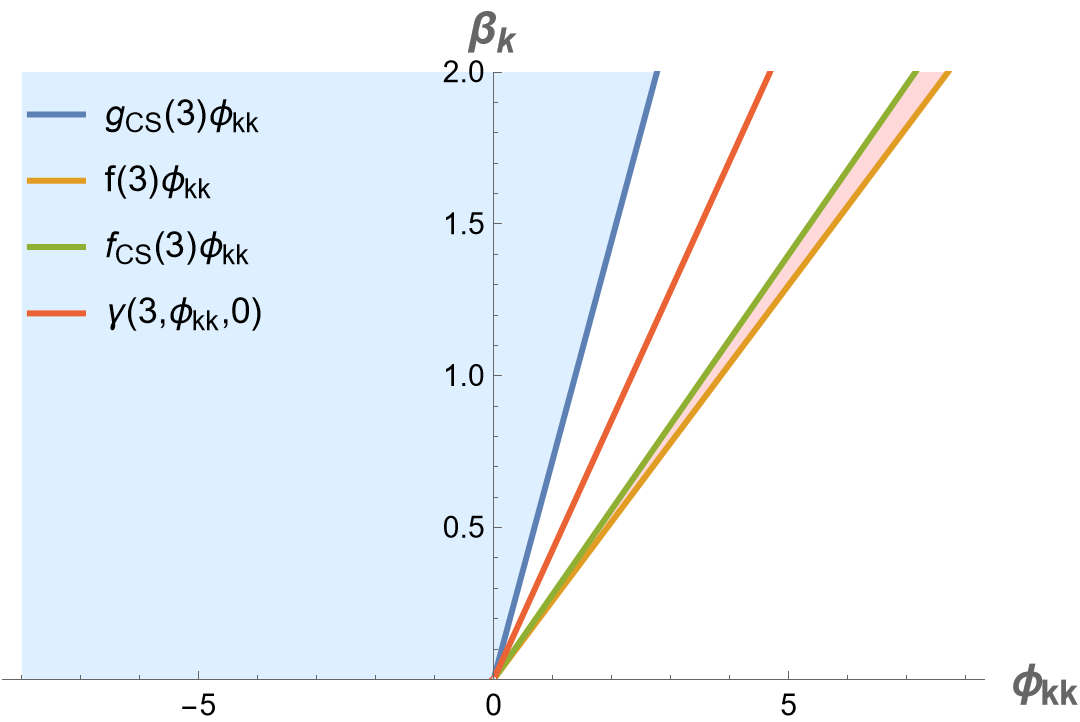}
\end{figure}

\textbf{The effect of competition on profits.} The equilibrium profit quantity, $\pi^*$, is given by
\begin{equation}
\label{eqn:equilibrium_pi}
\pi^{*}:=\sum_{k\in\left\{ b,s\right\} }p_{k}^{*}x_{k}^{*}.
\end{equation}
For each $k\in\left\{ b,s\right\}$, let $\pi_{k}^{*}:= p_{k}^{*}x_{k}^{*}$, the profits on side $k$ of the market. 
The following proposition provides sufficient conditions to determine the sign of $\partial \pi_k^*/\partial N$. It uses the following condition:
\begin{equation}\label{condition_pi_increasing}
\textnormal{either }
\left(\phi_{kk}\leq0\textnormal{ and }\beta_{k}>0\right)\textnormal{ or }\left(\phi_{kk}>0\textnormal{ and }\beta_{k}>g_{\pi}\left(N\right)\phi_{kk}\right),
\end{equation}
where $g_{\pi}\left(N\right)>f\left(N\right)$, $f(N)$ is given by \eqref{def:RNphi_kk_proof} and $g_{\pi}(N)$ is given by \eqref{def:rpi} in Appendix \ref{appendixA}. It also uses the functions $g_{\pi,z}(N,\phi_{kk},u_k^0,\beta_k)$ and $f_{\pi,z}(N,\phi_{kk},u_k^0,\beta_k)$   defined in \eqref{def:R1piz} and \eqref{def:R2piz} of Appendix \ref{appendixA}, respectively.

\begin{proposition}[Regions where competition decreases/increases profits on side $k$]\label{prop:dpik_dN_negative} Assume that $N\geq 2$ and for each $k\in\{b,s\}$, $(\phi_{kk},\beta_k)$ satisfies \eqref{condition_existence_beta}. The effect of competition on the change of profits on side $k$ of the market can be clustered into the following two regions of $z_k^*$:
\begin{itemize}   
 \item[(i)]  If $z_k^*<g_{\pi,z}(N,\phi_{kk},u_k^0,\beta_k)$, then there exists $\epsilon>0$ such that for any $(\phi_{bs},\phi_{sb})\in B_{\epsilon}(0)$, $\partial \pi_{k}^{*}/\partial N<0$. 
\item[(ii)]   If $z_k^*>f_{\pi,z}\left(N,\phi_{kk},u_{k}^{0},\beta_{k}\right)$ and (\ref{condition_pi_increasing}) is satisfied, then there exists $\epsilon>0$ such that for any $(\phi_{bs},\phi_{sb})\in B_{\epsilon}(0)$, $\partial \pi_{k}^{*}/\partial N>0$. 
\end{itemize}
\end{proposition}

Part (i) of Proposition \ref{prop:dpik_dN_negative} shows that in markets where the normalized net deterministic utility from joining the market is relatively small, the increased competition decreases profits. In other words, when the incentive to join the market in equilibrium, $z_k^*$, is small enough, more platforms joining the market reduce the pie for all of the competing platforms. A more interesting result appears in part (ii) when the incentives to join the market are high enough (relatively large $z_k^*$) and thus the increased competition increases profits. This observation aligns with traditional results in the platform's literature (see \cite{tan2021effects}) where the effect of network externalities can reverse the usual link between competition and firm profit (i.e., profits can increase with competition). 


\section{Economic and Policy discussion}
\label{sect:econ_poli_discussion}

We examine the economic implications of some of the results presented in Sections \ref{sect:collusionvscompetition} and \ref{sect:effectsCompetition}. We first discuss how increases in outside option utility and competition influence equilibrium prices and consumer surplus. These findings may inform policy discussions aimed at improving consumer outcomes and market efficiency. We also interpret our mathematical result comparing collusion and competition under small cross-side externalities. For concreteness, we focus on the dating app market, as motivated in the introduction.

\textbf{Scenarios in which an increased outside option or greater competition leads to lower prices and higher consumer surplus:} 
Both parts (i) of Propositions \ref{prop:dpkduk0} and \ref{prop:dCSkuk0} suggest that, under some conditions, such as relatively high heterogeneity, increasing the value of the outside option decreases prices and increases consumer surplus. In the setting of popular dating apps that attract a heterogeneous population---such as Tinder, Bumble, and Hinge---increased preference for traditional partner-finding methods leads to the reduction of dating app prices and the increase in consumer surplus. 
This suggests that apps should adjust pricing strategies, possibly by reducing prices or enhancing sign-up services. Similarly, both parts (i) of Propositions \ref{prop:dpk_dN} and \ref{prop:dCSk_dN_positive} imply that, under some conditions, such as relatively high heterogeneity, increasing competition decreases prices and increases consumer surplus. This finding is well-known in the traditional single-sided competition literature (see, e.g., \cite{tirole1988theory} and \cite{anderson1992logit}). 
In summary, under certain conditions—particularly high heterogeneity—our findings suggest two regulatory mechanisms to decrease prices and increase consumer surplus: enhancing the value of outside options or incentivizing competition.

\textbf{Scenarios in which an increased outside option or greater competition leads to higher prices and lower consumer surplus:} 
Both parts (ii) of Propositions \ref{prop:dpkduk0} and \ref{prop:dCSkuk0} suggest that under different conditions, such as relatively low heterogeneity, an 
increased outside option raises prices and reduces consumer surplus. These results may be exemplified by dating apps that target specific demographics or niches where users are often  homogeneous in their preferences. For example, apps like The League and JDate target more homogeneous segments of the population, and consequently, they can charge higher prices. Therefore, in these apps, users are less sensitive to outside options. Moreover, if subscribers are loyal at a sufficiently high outside option utility, there is no incentive to reduce prices even when this utility increases. 
Similarly, both parts (ii) of Propositions \ref{prop:dpk_dN} and \ref{prop:dCSk_dN_positive} imply that, under some conditions, such as relatively low heterogeneity, increasing competition leads to higher prices and lower consumer surplus. The intuition follows from  \eqref{eqn:equilibrium_CS}. We first note from this equation that consumer surplus is inversely related to price, so we focus on the former. Additionally, we observe that, for a fixed price, consumer surplus  increases with (i) the expected maximum user idiosyncrasy, and (ii) the size of the network externalities. In homogeneous populations, the expected maximum user idiosyncrasy is relatively small, which makes network externalities more pronounced. In this scenario, fewer platforms can amplify network effects more effectively (e.g., instead of having many alternatives to The League or JDate), making them more attractive. As a result, even with increased competition or a higher outside option, users may gravitate toward a smaller number of large platforms to maximize the benefits of network externalities. This dynamic allows these platforms to maintain or increase prices, while consumer surplus remain stagnant or decreases.

\textbf{Population heterogeneity matters for policy:}  The discussion above suggests that in markets like those described in this paper, regulators should carefully assess the level of population heterogeneity when aiming to improve consumer surplus and reduce equilibrium prices. This is because the same policy can have varying effects depending on the degree of heterogeneity.  Specifically, when population heterogeneity is sufficiently high, policies that either promote competition (e.g., reducing entry barriers and enforcing antitrust laws) or improve the outside option (e.g., enhancing public spaces like parks, libraries, and cultural centers) tend to lower equilibrium prices and increase consumer surplus. Conversely, when population heterogeneity is sufficiently low, policies that limit competition (e.g., supporting a dominant platform) or restrict the outside option (e.g., subsidizing part of the cost for some consumers) can help maintain or reduce prices while preserving or increasing consumer surplus. 

\textbf{Collusion under small cross-side externalities:} 
Proposition \ref{prop:CNEvsCE} shows that in cases of small cross-side
externalities, collusion (in comparison to competition) results in decreased normalized net deterministic utilities, reduced market participation and increased price, on both sides of the market. This is intuitive since when the cross-side externalities are sufficiently small---meaning users derive limited benefit from the presence of users on the opposite side---competing platforms have strong incentives to lower prices and attract users. In contrast, colluding platforms internalize each other's pricing decisions and reduce competition, enabling them to raise prices on both sides. This further results in higher net deterministic utilities and greater overall participation for the competition case versus the colluding one. The collusive outcome resembles classic monopoly pricing: platforms extract more surplus at the expense of user welfare, resulting in higher prices and lower market participation compared to the competitive case. In the dating app market, cross-side externalities capture the value one side (e.g., men) derives from a larger presence of the other side (e.g., women) on a given platform. These externalities 
are typically lower in large-scale casual apps like Tinder, Badoo, and Facebook Dating, where user pools are already extensive and the marginal value of new users is diminished. While Proposition~\ref{prop:CNEvsCE} is difficult to verify empirically, we illustrate its logic with a speculative example. During the 2013-2017 period of increasing competition among casual dating apps like Tinder, OkCupid, and Plenty of Fish (POF), prices were lower, user utility was higher, and market participation was larger. In contrast, we hypothesize that the dating app market has shifted in recent years toward reduced competition and arguably increased collusion. Match Group has gained a dominant position through acquisitions of major platforms such as Tinder, POF, OkCupid and Hinge~\citep{gilbert2019antitrust}. In parallel, the adoption of AI-based pricing strategies raises questions about the potential for tacit coordination~\citep{chica2024artificial}. During this period, rising prices have become evident. Moreover, features that were once free are increasingly placed behind paywalls. This results in lower utility for price-sensitive users and may limit participation, despite overall market growth.

Lastly, we note that the impact of increasing or decreasing competition appears both in Proposition \ref{prop:CNEvsCE}, where competition is compared to collusion in an extreme case, and in Propositions \ref{prop:dpk_dN} and \ref{prop:dCSk_dN_positive}, where competition changes by either increasing or decreasing the number of platforms in the market.

\section{Conclusions}
\label{sect:conclusions}

We provided a realistic framework for platform competition and collusion with an outside option. Among our many results, we highlight the following key findings:
\begin{enumerate}
    \item 
When the cross-side externalities are sufficiently small, the normalized net deterministic utilities and market participation are smaller in collusion than competition, and the prices on both sides of the market are higher in collusion than competition.
\item Depending on the size of the user's heterogeneity in tastes, incorporating an outside option may increase or decrease the equilibrium price and consumer surplus w.r.t.~the no outside option model. In particular, a model of platform competition that omits the outside option will either overestimate or underestimate the true equilibrium price. 
\item Depending on the size of the user's heterogeneity in tastes, the number of platforms and the size of network externalities, we also demonstrated when different quantities either decrease or increase with increased competition.\footnote{In particular, when the number of platforms increases, prices decrease if the user's heterogeneity is relatively large compared to the within-side externalities, and increase if there are at least three platforms and the user's heterogeneity is relatively small compared to the within-side externalities; market participation always increases; consumer surplus increases if the user's heterogeneity is relatively large compared to the within-side externalities, and decreases if there are at least three platforms, the user's heterogeneity is relatively small compared to the within-side externalities, and the net deterministic utility is small relative to the number of platforms; and profits decrease if the net normalized deterministic utility is small enough and increase if the net normalized deterministic utility is large enough.}
\end{enumerate}

While the paper uses lengthy mathematical derivation, a basic and fundamental idea is demonstrated in \eqref{price_decomposition}. This equation  decomposes the price gap between collusion and competition into two forces: reduced network benefits from lower participation, and lower user utility under collusion. Together, these explain why prices are higher in the collusive regime. 

There are many open directions for future research. In particular, it would be interesting to extend our model to incorporate the following features: (i) a multi-homing option, i.e., allowing users to join more than one platform; and (ii) platform asymmetries, i.e., allowing for different marginal costs of serving users. Incorporating multi-homing would require introducing an additional decision margin for users, potentially following the frameworks in \cite{chica2021exclusive} or \cite{teh2023multihoming}. For the case of platform asymmetries, one could modify problems \eqref{pi} and \eqref{pim} by introducing a marginal cost $c_i > 0$ for each platform $i \in \{1, \dots, n\}$. Exploring these extensions would likely require a combination of numerical methods and further simplifying assumptions. We view these as promising directions for future work that can build on the foundation laid by the present analysis. Another direction we are currently exploring is the use of our models as an economic framework for analyzing how reinforcement learning algorithms for platform pricing affect equilibrium outcomes. Our models help us assess whether network externalities mitigate or exacerbate the degree of collusion that AI-driven platforms may achieve \citep{chica2024artificial}.

\newpage
\appendix
\section{Appendix}\label{appendixA}
We prove all the stated results in the following order: Proposition \ref{prop:existence_xki}, Proposition~\ref{prop:FOC_gumbel} (for which we first provide various definitions and establish Lemma~\ref{lemma:foc_cneA}), Proposition \ref{prop:existence_gumbel} (for which we first prove Lemma~\ref{lemma:soc_cneA}), Proposition~\ref{prop:FOC_gumbel_collusion} (for which we first prove Lemma~\ref{lemma:foc_ce}), Proposition~\ref{coro:mono_uniqueness} (for which we first prove Lemma~\ref{lemma:soc_ce}), Proposition \ref{prop:sign_zk}, Corollary \ref{coro:sign_zk}, Corollary \ref{coro:sign_zk_N}, Proposition \ref{prop:dpkduk0}, Proposition \ref{prop:dpikduk0}, Proposition \ref{prop:dCSkuk0}, Proposition \ref{prop:sign_zkm}, Corollary \ref{coro:gammavsgammaC},  Corollary \ref{coro:sign_zkm},  Proposition \ref{prop:CNEvsCE}, Proposition \ref{prop:dpk_dN}, Proposition \ref{prop:dNxk_dN}, Proposition \ref{prop:dCSk_dN_positive} and Proposition \ref{prop:dpik_dN_negative}.
In order to save space, we leave some of the lengthy calculations to Mathematica and report them in the supplementary file \emph{Gumbel\_N.nb}.

\begin{proof}[{\bf Proof of Proposition \ref{prop:existence_xki}}]  Let $\{(p_b^i,p_s^i)\}_{i=1}^N\subset\R^2$ be a set of prices. For $i\in\gN\cup\{0\}$ and $k\in\{b,s\}$, set $v_k^i := \hat{u}_k^i- \varepsilon_k^i$. From (\ref{uki}) and (\ref{uk0}), for $i,j\in\mathcal{N}\cup\{0\}$, $i\neq j$, (\ref{xki}) can be rewritten as
\begin{align}
    x_k^i & =  \sP(\hat{u}_k^i > \max_{j\neq i}(\hat{u}_k^j))\notag \\ 
    & = \sP(\epsilon_k^i > \max_{j\neq i}(\epsilon_k^j +v_k^j-v_k^i)),
    \quad k\in\{b,s\}.\label{eqn:xki_outside}
\end{align}
For $i\in\gN\cup\{0\}$ and $k\in\{b,s\}$, we define $T_k^i:\R^{N+1}\longrightarrow [0,1]$ such that 
\begin{equation*}
    \vu=(u^0,u^1,\cdots,u^N) \mapsto T_k^i(\vu) := \sP(\underbrace{\epsilon_k^i > \max_{j\neq i}(\epsilon_k^j +u^j-u^i)}_{:= E^i_k(\vu)}).
\end{equation*}
Note that $E^i_k(\vu)\subset \Omega$ (where $\Omega$ is the domain of the random variables $\{\varepsilon_k^i\}_{i\in\gN\cup\{0\},k\in\{b,s\}}$). In two steps, we show that for any $\vu\in\R^{N+1}$ and $k\in\{b,s\}$, $\sum_{i=0}^NT_k^i(\vu) = 1$. 

Step (i): For any $i\neq j$, the events $E^i_k(\vu)$ and $E^j_k(\vu)$ are disjoint because either $\epsilon_k^i > \epsilon_k^j +u^j-u^i$ or $\epsilon_k^j > \epsilon_k^i +u^i-u^j$, but not both. Then, $
    \sum_{i=0}^NT_k^i(\vu) = \mathbb{P}\left(\cup_{i=0}^NE^i_k\right).$

Step (ii): We show that $\mathbb{P}\left(\cap_{i=0}^N(E^i_k)^c\right)=0$. First, note that the sets $\{\overline{E_k^i}\cap (E_k^i)^c\}_{i\in\gN\cup\{0\},k\in\{b,s\}}$ have $\mathbb{P}$-zero probability, because $\mathbb{P}$ is absolutely continuous with respect to the Lebesgue measure and each of the sets $\overline{E_k^i}\cap (E_k^i)^c$ is contained inside an $N$-dimensional set of $\mathbb{R}^{N+1}$. We claim that $\cap_{i=0}^N(E^i_k)^c\subseteq \cup_{i=0}^N(\overline{E_k^i}\cap (E_k^i)^c)$. 
Let $\omega\in \cap_{i=0}^N(E^i_k)^c$. Then, for all $i\in\gN\cup\{0\}$,
\begin{equation}\label{Ekic}
    \epsilon_k^i \leq \max_{j\neq i}(\epsilon_k^j +u^j-u^i).
\end{equation}
If there exists $i\in\gN\cup\{0\}$ such that \eqref{Ekic} holds with equality, then $\omega\in\overline{E_k^i}\cap (E_k^i)^c$ and the claim holds true. Now we prove that if for all $i\in\gN\cup\{0\}$, \eqref{Ekic} is satisfied with strict inequality, we get a contradiction. By (\ref{Ekic}) with strict inequality, there exists $\theta(0)\in\gN\cup\{0\}$, $\theta(0)\neq 0$ such that $ \epsilon_k^0 < \epsilon_k^{\theta(0)} +u^{\theta(0)}-u^0.$ Similarly, there exists $\theta^2(0)\neq \theta(0)$ such that $ \epsilon_k^{\theta(0)}  <\epsilon_k^{\theta^2(0)}  +u^{\theta^2(0)} -u^{\theta(0)}$. Note that $\theta^2(0)\neq 0$, otherwise, $ \epsilon_k^{\theta(0)}  <\epsilon_k^{0}  +u^{0} -u^{\theta(0)}$ which contradicts the definition of $\theta(0)$. By induction, suppose that for $n\in \mathbb{N}$ and all $0\leq m \leq n $, there exists $\theta^m(0)\in\gN\cup\{0\}$  such that 
\begin{equation}\label{Ekic2}
    \theta^m(0)\notin\{0,\theta(0),\cdots, \theta^{m-1}(0)\} \ \textnormal{ and }\  \epsilon_k^{\theta^{m-1}(0)}  <\epsilon_k^{\theta^m(0)}  +u^{\theta^m(0)} -u^{\theta^{m-1}(0)}.
\end{equation}
By (\ref{Ekic}) with strict inequality, there exists $\theta^{n+1}(0)\in\gN\cup\{0\}$, $\theta^{n+1}(0)\neq \theta^{n}(0)$ such that $ \epsilon_k^{\theta^{n}(0)}  <\epsilon_k^{\theta^{n+1}(0)}  +u^{\theta^{n+1}(0)} -u^{\theta^{n}(0)}$. We claim that $\theta^{n+1}(0)\notin\{0,\theta(0),\cdots, \theta^{n}(0)\}$, otherwise $\theta^{n+1}(0) = \theta^m(0)$ for some $0\leq m \leq n-1$. In this case, by \eqref{Ekic2}
\begin{equation}\begin{split}\label{Ekic3}
     \epsilon_k^{\theta^{n+1}(0)}=\epsilon_k^{\theta^m(0)}  &<\epsilon_k^{\theta^{m+1}(0)}  +u^{\theta^{m+1}(0)} -u^{\theta^m(0)}\\
    &<\epsilon_k^{\theta^{m+2}(0)}  +u^{\theta^{m+2}(0)} -\cancel{u^{\theta^{m+1}(0)}}+\cancel{u^{\theta^{m+1}(0)}} -u^{\theta^m(0)}\\
    &\vdots\\
    &<\epsilon_k^{\theta^{n}(0)}  +u^{\theta^{n}(0)} - u^{\theta^m(0)}.
\end{split}
\end{equation}
Note that (\ref{Ekic3}) contradicts the definition of $\theta^{n+1}(0)$. Then
$\theta^{n+1}(0)\notin\{0,\theta(0),\cdots, \theta^{n}(0)\}$ and \eqref{Ekic2} is satisfied for the next index $n+1$. It follows that  \eqref{Ekic2} holds for any $n\in \mathbb{N}$. The latter is impossible  because there are only $N+1$ different indices inside $\gN\cup\{0\}$. Thus, $\cap_{i=0}^N(E^i_k)^c\subseteq \cup_{i=0}^N(\overline{E_k^i}\cap (E_k^i)^c)$ and  $\mathbb{P}\left(\cap_{i=0}^N(E^i_k)^c\right)=0$. 

Combining steps (i) and (ii), we get that for any $\vu\in\R^{N+1}$ and $k\in\{b,s\}$, $ \sum_{i=0}^2T_k^i(\vu) = 1$.  Now, for $\vx = (x_{0b},x_{0s},x_{1b},x_{1s},\dots,x_{Nb},x_{Ns})\in[0,1]^{2(N+1)}$ and each $ i\in\gN\cup\{0\}$, we introduce the auxiliary functions $\phi_k^i(\vx)$ defined as 
\begin{equation*}
   \phi_k^i(\vx) =  \begin{cases}
    \phi_k^0  & \textnormal{ if } i=0 \\
       \phi_k(x_{ib},x_{is})  & \textnormal{ if } i\geq 1 
    \end{cases}.
\end{equation*}

Similarly, we define
$$\sigma_k^i(\vx) := T_k^i(v_k^0(\vx),\cdots,v_k^N(\vx)), $$
where $v_k^j(\vx)=\phi_k^j(\vx)-p_k^j$ ($p_k^0=0$).
If $\Sigma:[0,1]^{2(N+1)}\longrightarrow[0,1]^{2(N+1)}$ is defined by $$\Sigma(\vx)=(\sigma_b^0(\vx),\sigma_s^0(\vx)\cdots,\sigma_b^N(\vx),\sigma_s^N(\vx)),$$ then solving system (\ref{xki}) is equivalent to finding a fixed point of $\Sigma$, i.e., $\Sigma(\vx)=\vx$. Existence of such a fixed-point is guaranteed by Brouwer's Fixed Point Theorem, as $\Sigma$ is continuous on $[0,1]^{2(N+1)}$. For such a fixed point: $\sum_{i=0}^Nx_k^i = \sum_{i=0}^N\sigma_k^i(\vx) = \sum_{i=0}^NT_k^i(v_k^0(\vx),\cdots,v_k^N(\vx))=1 $.  

To show the uniqueness of the solution of (\ref{xki}), we use the Banach Fixed Point Theorem. Let $\vx,\vy\in [0,1]^{2(N+1)}$, then 
\begin{equation*}
    \begin{split}
        |\sigma_k^i(\vx)-\sigma_k^i(\vy)| &= |T_k^i(v_k^0(\vx),\cdots,v_k^N(\vx))-T_k^i(v_k^0(\vy),\cdots,v_k^N(\vy))|\\
        &\leq \max_j|v^j_k(\vx)-v^j_k(\vy)| \cdot \left( \sup_{\vu\in\R^{N+1}}\sum_{l=0}^N\left|\frac{\partial T_k^i(\vu)}{\partial u^l}\right| \right)\\
        & \leq \max_j|\phi^j_k(\vx)-\phi^j_k(\vy)|\cdot M_T \\
        & \leq M_\phi M_T |\vx-\vy|_\infty, \\
    \end{split}
\end{equation*}
where 
\begin{equation*}
    \begin{split}
        M_T &:=  \max_{k\in\{b,s\},i\in\gN\cup\{0\}}\sup_{\vu\in\R^{N+1}}\sum_{l=0}^N\left|\frac{\partial T_k^i(\vu)}{\partial u^l}\right|, \textnormal{ and}\\
         M_\phi &:= \max_{k\in\{b,s\}} \sup_{(x_b,x_s)\in[0,1]^2}\sum_{l\in \{b,s\}}\left|\frac{\partial \phi_k(x_b,x_s)}{\partial x_l}\right|.
    \end{split}
\end{equation*}

It follows that $\Sigma(\cdot)$ is a (strict) contracting mapping whenever $M_T M_\phi<1$, and uniqueness follows.

\end{proof}


{\bf Preliminary results for the proof of Proposition~\ref{prop:FOC_gumbel}.}
We introduce notation and definitions and establish a useful lemma. Let $\boldsymbol{X}=\left(x_{b}^{1},\cdots,x_{b}^{N},x_{s}^{1},\cdots,x_{s}^{N}\right)$ and   $\boldsymbol{P}=\left(p_{b}^{1},\cdots,p_{b}^{N},p_{s}^{1},\cdots,p_{s}^{N}\right)$ be two vectors in $\mathbb{R}^{2N}$. For $k\in\{b,s\}$, let $\tilde{\vu}_k := (u_{k}^{0}, u_k^1,\dots,u_{k}^{N}) $, where $u_{k}^{i}=\phi_{k}\left(x^i_b,x_s^i\right)-p_{k}^{i}$. Using \eqref{xki_Tki}, we can define a mapping from $\mathbb{R}^{4N}$ to $\mathbb{R}^{2N}$ as 
\begin{equation}
    (\boldsymbol{X}, \boldsymbol{P}) \mapsto \mathcal{T}\left(\boldsymbol{X},\boldsymbol{P}\right):=\left(T_{b}^{1}\left(\tilde{\vu}_b\right)-x_{b}^{1},\cdots,T_{b}^{N}\left(\tilde{\vu}_b\right)-x_{b}^{N},T_{s}^{1}\left(\tilde{\vu}_s\right)-x_{s}^{1},\cdots,T_{s}^{N}\left(\tilde{\vu}_s\right)-x_{s}^{N}\right).
\end{equation}
The Jacobian of \eqref{xki_Tki} w.r.t.~$\boldsymbol{P}$ is defined as
\begin{equation}\label{notation_prop32_jacobian}
    \begin{split}
    \textnormal{det}\frac{\partial\mathcal{T}}{\partial\boldsymbol{P}}\left(\boldsymbol{X},\boldsymbol{P}\right)&:=Q_{b}\left(\boldsymbol{X},\boldsymbol{P}\right)Q_{s}\left(\boldsymbol{X},\boldsymbol{P}\right),\ \textnormal{ where }\\
        Q_{k}\left(\boldsymbol{X},\boldsymbol{P}\right)&:=\left|\begin{array}{ccc}
-\frac{\partial T_{k}^{1}}{\partial u^{1}}(\tilde{\vu}_k) & \dots & -\frac{\partial T_{k}^{1}}{\partial u^{N}}(\tilde{\vu}_k)\\
\vdots & \ddots & \vdots\\
-\frac{\partial T_{k}^{N}}{\partial u^{1}}(\tilde{\vu}_k) & \dots & -\frac{\partial T_{k}^{N}}{\partial u^{N}}(\tilde{\vu}_k)
\end{array}\right|.
    \end{split}
\end{equation}

Under symmetry, for any $i\in\mathcal{N}$ and $k\in\{b, s\}$, we write $p_k^i = p_k$, $x_k^i = x_k$, and $u_k^i = u_k := \phi_k(x_b, x_s) - p_k$. Let $\vu_k := (u_k^0, \phi_k(x_b, x_s) - p_k, \cdots, \phi_k(x_b, x_s) - p_k)^T\in\sR^{N+1}$. For $i,j\in \gN$, $i\neq j$, $k\in\{b,s\}$, we define the functions
\begin{equation}\label{notation_prop32}
    \begin{split}
        & S_k(\vu_k) := \frac{\partial T_k^i}{\partial u^i_k}(\vu_k),\\
    & R_k(\vu_k):= \frac{\partial T_k^i}{\partial u^j_k}(\vu_k),\\
    & J_k(\vu_k) := S_k(\vu_k)\left(S_k(\vu_k) + (N-2)R_k(\vu_k)\right) - (N-1)R_k(\vu_k)^2, \ \text{ and }\\
    & J_\phi(\vu_b, \vu_s) := \left( \frac{\partial \phi_s}{\partial x_s} - \frac{1}{J_s(\vu_s)}S_s(\vu_s)\right)\left(\frac{\partial \phi_b}{\partial x_b} - \frac{1}{J_b(\vu_b)}S_b(\vu_b)\right) - \frac{\partial \phi_s}{\partial x_b}\frac{\partial \phi_b}{\partial x_s}.
    \end{split}
\end{equation}
Whenever there is no room for confusion, we simplify the notations by neglecting the explicit mention of the input $\vu_k$. For example, $T_k^i(\vu_k)$, $S_k(\vu_k)$, $R_k(\vu_k)$, $J_k(\vu_k)$ and $J_\phi(\vu_b, \vu_s)$ are simplified to $T_k^i$, $S_k$, $R_k$, $J_k$ and $J_\phi$ respectively. The following Lemma shows the first-order condition of (\ref{pi}) as a function of $x_k$.

\begin{lemma}[FOC of CNE]\label{lemma:foc_cneA}
If $\textnormal{det}\frac{\partial\mathcal{T}}{\partial\boldsymbol{P}}\left(\vx^\ast,\vp^\ast\right)\neq 0$, then the symmetric Nash equilibrium outputs $\vp^\ast$ and $\vx^\ast$ are solutions of (\ref{xki}) and of the following two equations
\begin{equation}\label{eqn:s_FOC_b}
\begin{split}
    &p_k + \frac{\partial \phi_k}{\partial x_k}x_k + \frac{\partial \phi_l}{\partial x_k}x_l - \frac{1}{J_k} \left(S_k + (N-2)R_k\right) x_k + \frac{1}{J_k^2J_l J_\phi}(N-1)R_k^2S_l x_k  \\
    &+\frac{N-1}{J_kJ_\phi}R_k\left(\frac{1}{J_l}R_l\frac{\partial \phi_l}{\partial x_k}x_l - \frac{1}{J_k}R_k\frac{\partial \phi_l}{\partial x_l}x_k\right)  = 0, \ \textnormal{ for }k,l\in\{b,s\},\ k\neq l.
    \end{split}
\end{equation}
\end{lemma}

The proof of this Lemma does not require assumptions I and II of Section \ref{sect:platform_competition}. Thus, the FOC given by (\ref{eqn:s_FOC_b}) is applicable to idiosyncratic preferences other than Gumbel distribution and to more general externality functions $\phi_k(\vx)$.

\begin{proof}[Proof of Lemma~\ref{lemma:foc_cneA}]



Assume that all platforms follow a symmetric equilibrium where $\vp^i =\vp^\ast= (p_b^\ast, p_s^\ast)$ and $\vx^i =\vx^\ast= (x_b^\ast, x_s^\ast)$. We show that unilateral deviations from this strategy lead to zero gain. Without loss of generality, we assume that the first platform deviates from the symmetric setting. This platform can deviate by either choosing prices $p_k^1\neq p_k^\ast$ or market shares $x_k^1\neq x_k^\ast$. Suppose that $\textnormal{det}\frac{\partial\mathcal{T}}{\partial\boldsymbol{P}}\left(\boldsymbol{X}^\ast,\boldsymbol{P}^\ast\right)\neq 0$, where $$\boldsymbol{X}^\ast = (x_{b}^\ast,\cdots,x_{b}^\ast,x_{s}^\ast,\cdots,x_{s}^\ast)\ \textnormal{ and } \boldsymbol{P}^\ast = (p_{b}^\ast,\cdots,p_{b}^\ast,p_{s}^\ast,\cdots,p_{s}^\ast)$$ belong to $\mathbb{R}^{2N}$. Then, by the Implicit Function Theorem, there exists a neighborhood $B$ of $\boldsymbol{X}^\ast$ in $\mathbb{R}^{2N}$ and a unique differentiable function $\mathcal{P}:B\longrightarrow\mathbb{R}^{2N}$ such that $\mathcal{P}\left(\boldsymbol{X}^{*}\right)=\boldsymbol{P}^{*}$ and 
\begin{equation}\label{eqn:IFT}
    \mathcal{T}\left(\boldsymbol{X},\mathcal{P}\left(\boldsymbol{X}\right)\right)=0\textnormal{ for all }\boldsymbol{X}\in B.
\end{equation} 
From \eqref{pi} and \eqref{eqn:IFT}, we can compute the FOC for this platform w.r.t.~$x_k^1$ as
\begin{align}
    \frac{\partial \pi^1}{\partial x^1_k} \Big|_{\vp^i = \vp^\ast, \vx^i = \vx^\ast, \textnormal{ for } i\neq 1} & = p^1_k + x^1_k\frac{\partial p^1_k}{\partial x_k^1} + x^1_l\frac{\partial p^1_l}{\partial x_k^1} = 0,  \ \textnormal{ for each }k,l \in\{b,s\},\ k\neq l. \label{FOC_k}
\end{align}

To solve \eqref{FOC_k}, we need to compute the following derivatives
\begin{equation}
\frac{\partial p^1_k}{\partial x^1_l},\ \textnormal{ for each }k,l \in\{b,s\}.
\label{dpdx}
\end{equation}
We determine those four partial derivatives $\frac{\partial\vp^1}{\partial \vx^1}$ in \eqref{dpdx} using the definition of $T_k$ in \eqref{def:Qki}. By \eqref{xki_Tki}, for $k\in\{b,s\}$, the vectors of market shares and prices, $(x_k^1,x_k^\ast,\dots,x_k^\ast)$ and $(p_k^1,p_k^\ast,\dots,p_k^\ast)$ satisfy all the following
\begin{align}
    T^1_k(u^0_k, u^1_k, u^2_k, \dots, u^N_k) &= x_k^1 \  \textnormal{ and} \label{eqn:proof_T_functions_1}\\
    T^i_k(u^0_k, u^1_k, u^2_k, \dots, u^N_k) &= x_k^i, \ \text{ for } i\in\{2, 3, \cdots, N\}.
    \label{eqn:proof_T_functions_i}
\end{align}
Taking derivatives w.r.t. $x_b^1$ and $x_s^1$ in \eqref{eqn:proof_T_functions_1} and \eqref{eqn:proof_T_functions_i}, respectively, gives us
\begin{align}    
\frac{\partial T_k^1(u^0_k, u^1_k, u^2_k, \dots, u^N_k)}{\partial x_l^1} &= \delta_{kl} \ \textnormal{ and } \label{T1_eq1}\\
\frac{\partial T_k^i(u^0_k, u^1_k, u^2_k, \dots, u^N_k)}{\partial x_l^1} &= \frac{\partial x_k^i}{\partial x_l^1}, \text{ for } i\in\{2, 3, \cdots, N\},\ k,l \in\{b,s\}, \label{T2_eq1}
\end{align}
where $\delta_{kl}=1$ if $k=l$ and $\delta_{kl}=0$ if $k\neq l$. Note that the system of equations in\eqref{T1_eq1} and \eqref{T2_eq1} includes $4+4(N-1)=4N$ equations. The unknowns are $\frac{\partial p_k^1}{\partial x_l^1}$ and $\frac{\partial x_k^i}{\partial x_l^1}$, for $k, l \in \{ b, s\}$ and $i\in \{2, 3,\dots, N\}$, adding up to $4N$ unknowns. 




Note that $u_k^0$ is a constant and it does not depend on $x_k^1$. By the chain rule, the left hand side of \eqref{T1_eq1} and \eqref{T2_eq1} can be rewritten as
\begin{equation}
\frac{\partial T_k^j(u^0_k, u^1_k, u^2_k, \dots, u^N_k)}{\partial u_1} \frac{\partial u_k^1}{\partial x_l^1} + \sum_{i=2}^N\frac{\partial T_k^j(u^0_k, u^1_k, u^2_k,\dots, u^N_k)}{\partial u_i} \frac{\partial u_k^i}{\partial x_l^1}, \quad \textnormal{ for } k,l \in \{b,s\},\ j\in\{1,2,\dots, N\}.
   \label{pTpx}
\end{equation}
Recall that 
$
u_k^1 = \phi_k(\vx^1) - p^1_k,
$
then, we can explicitly write 
\begin{align}    
\frac{\partial u_k^1}{\partial x_l^1} & = -\frac{\partial p_k^1}{\partial x_l^1} + \frac{\partial \phi_k}{\partial x_l} , \quad \textnormal{ for } k,l \in \{b,s\}.\label{u1_eq1}
\end{align}
On the other hand, for $i\in \{2, 3, \cdots, N\}$, 
$
u_k^i = \phi_k(\vx^i) - p_k^\ast
$. Thus, 
\begin{equation}    
\frac{\partial u_k^i}{\partial x_l^1} = \frac{\partial \phi_k}{\partial x_b}\frac{\partial x_b^i}{\partial x_l^1} + 
\frac{\partial \phi_k}{\partial x_s}\frac{\partial x_s^i}{\partial x_l^1}, \quad \textnormal{ for } k,l \in \{b,s\},\ i\in\{2,\dots, N\}.
\label{pu2px}
\end{equation}
Plugging \eqref{pTpx},  \eqref{u1_eq1} and \eqref{pu2px} into \eqref{T1_eq1} and \eqref{T2_eq1}, we obtain for $k,l\in\{b,s\}$ and $j\in\{2,3,\cdots, N\}$,
\begin{align}    
\frac{\partial T_k^1}{\partial u_1} \left(-\frac{\partial p_k^1}{\partial x_l^1} + \frac{\partial \phi_k}{\partial x_l}\right) + \sum_{i=2}^N\frac{\partial T_k^1}{\partial u_i} \left(\frac{\partial \phi_k}{\partial x_b}\frac{\partial x_b^i}{\partial x_l^1} + 
\frac{\partial \phi_k}{\partial x_s}\frac{\partial x_s^i}{\partial x_l^1}\right) &= \delta_{kl}, \label{pT1bpxb}\\
\frac{\partial T_k^j}{\partial u_1} \left(-\frac{\partial p_k^1}{\partial x_l^1} + \frac{\partial \phi_k}{\partial x_l}\right) + \sum_{i=2}^N\frac{\partial T_k^j}{\partial u_i} \left(\frac{\partial \phi_k}{\partial x_b}\frac{\partial x_b^i}{\partial x_l^1} + 
\frac{\partial \phi_k}{\partial x_s}\frac{\partial x_s^i}{\partial x_l^1}\right) &= \frac{\partial x^j_k}{\partial x_l^1}, \label{pT2bpxb}
\end{align}

There are $4N$ equations with $4N$ variables in the above system. Before solving the system, we want to apply the property of symmetry at equilibrium, where we denote
\begin{equation}
\begin{split}
    \frac{\partial x^j_k}{\partial x^1_l} &=: y_{k,l}, \ \text{ for } j \in \{2, 3,\cdots, N\} \textnormal{ and} \\
    \frac{\partial T_k^i}{\partial u_j} &=: \left\{ 
\begin{array}{cc}
S_k & i = j\\
R_k & i\neq j
\end{array}
\right. , \ \text{ for } i, j \in \{1,\dots,N\}. \label{eqn_y_RS}
\end{split}
\end{equation}
Incorporating \eqref{eqn_y_RS}, we can further reduce the system described by (\ref{pT1bpxb}) and (\ref{pT2bpxb}) into 8 equations with unknowns: $\frac{\partial p_k^1}{\partial x_l^1}$ and $y_{k,l}$, for $k, l \in\{ b, s\}$. The new system is, for $k,l\in\{b,s\}$, 

\begin{align}    
S_k \left(-\frac{\partial p_k^1}{\partial x_l^1} + \frac{\partial \phi_k}{\partial x_l}\right) + \left(N-1\right)R_k \left(\frac{\partial \phi_k}{\partial x_b}y_{b,l} + 
\frac{\partial \phi_k}{\partial x_s}y_{s,l}\right) &= \delta_{kl} \ \textnormal{ and} \label{A_1}\\
R_k \left(-\frac{\partial p_k^1}{\partial x_l^1} + \frac{\partial \phi_k}{\partial x_l}\right) + \left(S_k + \left(N-2\right)R_k\right) \left(\frac{\partial \phi_k}{\partial x_b}y_{b, l} + 
\frac{\partial \phi_k}{\partial x_s}y_{s, l}\right) &= y_{k, l}. \label{A_3}
\end{align}

To solve the above equations, we notice that there are two groups of equations: (A) four equations associated to derivatives w.r.t.~$x_b^1$; (B) four equations associated to derivatives w.r.t.~$x_s^1$. Thus, we solve two linear systems, each one containing four equations and four unknowns. 


Using \eqref{notation_prop32},  we conclude that when $J_k\neq 0$ for any $k\in\{b,s\}$ and $J_\phi\neq 0$,
\begin{align}
    y_{b, b} & = -\frac{1}{J_b J_\phi}R_b\left(\frac{\partial \phi_s}{\partial x_s} -\frac{S_s}{J_s}\right),\label{sln_pxbpxb}\\
    y_{s, b} & = \frac{1}{J_b J_\phi}R_b\frac{\partial \phi_s}{\partial x_b},\label{sln_pxspxb}\\
    y_{b, s} & = \frac{1}{J_s J_\phi}R_s\frac{\partial \phi_b}{\partial x_s},\ \textnormal{ and} \label{sln_pxbpxs}\\
    y_{s, s} & = -\frac{1}{J_s J_\phi}R_s\left(\frac{\partial \phi_b}{\partial x_b} -\frac{S_b}{J_b}\right).\label{sln_pxspxs}
\end{align}
Moreover, it follows that
\begin{align}
    \frac{\partial p_b^1}{\partial x_b^1} & = \frac{\partial \phi_b}{\partial x_b} - \frac{1}{J_b}(S_b + (N-2)R_b) + \frac{N-1}{J_b}R_b\left(-\frac{1}{J_b J_\phi}R_b\left(\frac{\partial \phi_s}{\partial x_s} -\frac{1}{J_s}S_s\right)\right),\label{sln_ppbpxb}\\
    \frac{\partial p_s^1}{\partial x_b^1} & = \frac{\partial \phi_s}{\partial x_b} + \frac{N-1}{J_s}R_s\left(\frac{1}{J_b J_\phi}R_b\frac{\partial\phi_s}{\partial x_b}\right),\label{sln_ppspxb}\\
    \frac{\partial p_b^1}{\partial x_s^1} & =  \frac{\partial \phi_b}{\partial x_s} + \frac{N-1}{J_b}R_b\left(\frac{1}{J_s J_\phi}R_s\frac{\partial\phi_b}{\partial x_s}\right),\ \textnormal{ and} \label{sln_ppbpxs}\\
    \frac{\partial p_s^1}{\partial x_s^1} & = \frac{\partial \phi_s}{\partial x_s} - \frac{1}{J_s}\left(S_s + (N-2)R_s\right) + \frac{N-1}{J_s}R_s\left(-\frac{1}{J_s J_\phi}R_s\left(\frac{\partial \phi_b}{\partial x_b} -\frac{1}{J_b}S_b\right)\right).\label{sln_ppspxs}
\end{align}

Finally, we can plug in \eqref{sln_ppbpxb}, \eqref{sln_ppbpxs}, \eqref{sln_ppspxb} and \eqref{sln_ppspxs} into the FOC \eqref{FOC_k} to obtain
\begin{equation}\label{eqn:k_FOC_k}
\begin{split}
    &p_k + \frac{\partial \phi_k}{\partial x_k}x_k + \frac{\partial \phi_l}{\partial x_k}x_l - \frac{1}{J_k} \left(S_k + (N-2)R_k\right) x_k + \frac{1}{J_k^2J_l J_\phi}(N-1)R_k^2S_l x_k  \\
    &+\frac{N-1}{J_kJ_\phi}R_k\left(\frac{1}{J_l}R_l\frac{\partial \phi_l}{\partial x_k}x_l - \frac{1}{J_k}R_k\frac{\partial \phi_l}{\partial x_l}x_k\right)  = 0, \ \textnormal{ for }k,l\in\{b,s\},\ k\neq l.
    \end{split}
\end{equation}
    
\end{proof}

Lemma~\ref{lemma:foc_cneA} is general enough to accommodate idiosyncratic preferences other than Gumbel distribution and general externality functions $\phi_k(\vx)$. Next, we use Assumptions I and II from Section \ref{sect:platform_competition} and Lemma~\ref{lemma:foc_cneA} to prove Proposition~\ref{prop:FOC_gumbel}.

\begin{proof}[{\bf Proof of Proposition \ref{prop:FOC_gumbel}}]

We want to rewrite the FOC given by \eqref{eqn:s_FOC_b} using Assumptions I and II from Section \ref{sect:platform_competition}. Applying Gumbel distribution, we can derive specific forms for the functions $T_k^i$, $S_k$, and $R_k$ as defined in \eqref{notation_prop32}. By Assumption I, $\{\varepsilon_k^i\}_{k\in\{b,s\},  i\in\gN\cup\{0\}}$ are i.i.d.~Gumbel distributed with distribution
\begin{equation}\label{eqn:cdf_gumbel_proof}
    F_k(z) = e^{-e^{\frac{\mu_k -z}{\beta_k }}}.
\end{equation}
For any $i\in\gN$, the random variable $Y^i:=\epsilon_{k}^{i}+u_{k}^{i}-\max_{j=0,1,\dots,N,j\neq i}\left\{ \epsilon_{k}^{j}+u_{k}^{j}\right\}$ has a logistic distribution, 

$$Y^i \sim \textrm{Logistic}(u_{k}^{i}-\beta_{k}\ln\alpha_{i},\beta_{k}), \ \textnormal{ where } \ \alpha_i:=\sum_{j=0,1,\dots,N,j\neq i}e^{\frac{u_{k}^{j}}{\beta_{k}}}. $$
By \eqref{def:Qki}, we can explicitly write $ T^i_k $ as follows
\begin{align}
    T^i_k(u^0_k, u^1_k, \dots, u^N_k) & = \mathbb{P}\left(\epsilon_{k}^{i}+u_{k}^{i}\geq\max_{j=0,1,...N,j\neq i}\left\{ \epsilon_{k}^{j}+u_{k}^{j}\right\} \right) \notag\\
    & = 1 - F_{Y^i}(0) 
     = 1 - \frac{1}{1 + e^{\left(u^i_k - \beta_k \ln( \alpha_i ) \right)/\beta_k}}.\label{eqn:T_res_gumbel}
\end{align}
 The derivatives of $T_k^i$ can be calculated as 
\begin{align}
    \frac{\partial T^i_k(u^0_k, u^1_k, \dots, u^N_k)}{\partial u^i} & = \frac{\frac{1}{\beta_k} e^{\left(u^i_k - \beta_k \ln( \alpha_i)\right)/\beta_k}}{\left(1 + e^{\left(u^i_k - \beta_k \ln(\alpha_i)\right)/\beta_k}\right)^2} \ \textnormal{ and} \label{eqn:S_k_res_gumbel}\\
    \frac{\partial T^i_k(u^0_k, u^1_k, \dots, u^N_k)}{\partial u^j} & =  \frac{- e^{\left(u^i_k - \beta_k \ln(\alpha_i)\right)/\beta_k} \cdot \frac{\frac{1}{\beta_k}e^{u^j_k/\beta_k}}{\alpha_i} }{\left(1 + e^{\left(u^i_k - \beta_k \ln(\alpha_i)\right)/\beta_k}\right)^2}, \quad \textnormal{ for } j\neq i. \label{eqn:R_k_res_gumbel}
\end{align}

At a symmetric equilibrium,   $\vu^i = \vu = (u_b, u_s)^T, $ for any $i\in\gN$. Then, we can further simplify \eqref{eqn:T_res_gumbel}, \eqref{eqn:S_k_res_gumbel} and \eqref{eqn:R_k_res_gumbel},
\begin{align}
T^i_k & = \frac{e^{u_k/\beta_k}}{e^{u^0_k/\beta_k} + Ne^{u_k/\beta_k}},\label{eqn:T}\\
S_k & = \frac{1}{\beta_k} \frac{e^{u_k/\beta_k}(e^{u^0_k/\beta_k} + (N-1)e^{u_k/\beta_k})}{\left(e^{u^0_k/\beta_k} + Ne^{u_k/\beta_k}\right)^2} \ \textnormal{ and} \label{eqn:Sk}\\
R_k & = -\frac{1}{\beta_k} \frac{e^{2u_k/\beta_k}}{\left(e^{u^0_k/\beta_k} + N e^{u_k/\beta_k}\right)^2}. \label{eqn:Rk}
\end{align}

Using \eqref{def:zk}, we derive the dependence of $\vx$ and $\vp$ on $\vz$ under the symmetric setting as 
\begin{equation}
    x_k^i = T_k^i(\vu) = \frac{e^{u_k/\beta_k}}{e^{u^0_k/\beta_k} + Ne^{u_k/\beta_k}} 
     = \frac{1}{e^{(u_k^0 - u_k)/\beta_k} + N} 
     = \frac{1}{e^{-z_k} + N} =: \omega(z_k). \label{eqn:proof_xki_omega}
\end{equation}
Denoting $\Omega(\vz) := (\omega(z_b), \omega(z_s))^T$, we obtain $\vx = \Omega(\vz)$ in the symmetric equilibrium. Moreover, we can write $\vu - \vu_0 = \vbeta\vz$, $\vu = \Phi \vx - \vp$ and 
\begin{align}\label{eqn:relationship_xp_z}
    \vp & = \Phi \Omega(\vz) - \vbeta \vz - \vu_0.
\end{align}
At this point, we want to use Lemma~\ref{lemma:foc_cneA} to rewrite \eqref{eqn:s_FOC_b} using Assumptions I and II. First, we verify that $\textnormal{det}\frac{\partial\mathcal{T}}{\partial\boldsymbol{P}}\left(\boldsymbol{X},\boldsymbol{P}\right)\neq 0$, see \eqref{notation_prop32_jacobian}, for any pair of symmetric vectors 
$$\boldsymbol{X}=\left(x_{b},\cdots,x_{b},x_{s},\cdots,x_{s}\right) \ \textnormal{and }\ \boldsymbol{P}=\left(p_{b},\cdots,p_{b},p_{s},\cdots,p_{s}\right).$$
Using \eqref{eqn:S_k_res_gumbel}, \eqref{eqn:R_k_res_gumbel} and \eqref{def:zk} into \eqref{notation_prop32_jacobian}, we obtain 
\begin{equation}
\begin{split}
    Q_{k}\left(\boldsymbol{X},\boldsymbol{P}\right)&=\left|\begin{array}{ccc}
-\frac{\partial T_{k}^{1}}{\partial u^{1}} & \dots & -\frac{\partial T_{k}^{1}}{\partial u^{N}}\\
\vdots & \ddots & \vdots\\
-\frac{\partial T_{k}^{N}}{\partial u^{1}} & \dots & -\frac{\partial T_{k}^{N}}{\partial u^{N}}
\end{array}\right|\\
&=\frac{1}{\beta_{k}^{N}}\frac{e^{Nz_{k}}\left(1+\left(N-1\right)e^{z_{k}}\right)^{N}}{\left(1+Ne^{z_{k}}\right)^{2N}}\left|\begin{array}{cccc}
-1 & \frac{e^{z_{k}}}{1+\left(N-1\right)e^{z_{k}}} & \dots & \frac{e^{z_{k}}}{1+\left(N-1\right)e^{z_{k}}}\\
\vdots & -1 & \ddots & \vdots\\
\frac{e^{z_{k}}}{1+\left(N-1\right)e^{z_{k}}} & \frac{e^{z_{k}}}{1+\left(N-1\right)e^{z_{k}}} & \dots & -1
\end{array}\right|\\
& =\frac{1}{\beta_{k}^{N}}\frac{e^{Nz_{k}}\left(-1\right)^{N}}{\left(1+Ne^{z_{k}}\right)^{N+1}}\neq0. 
\end{split}
\end{equation}

We can now apply Lemma~\ref{lemma:foc_cneA} and equations \eqref{eqn:T} to \eqref{eqn:relationship_xp_z} into \eqref{eqn:s_FOC_b} to obtain
\begin{equation}
\vbeta \vz=\left(\Phi -H(\vz)\right)\Omega(\vz) - \vu_0 ,\label{FOCs_z_proof2}
\end{equation}
where $\vu_0=(u_b^0,u_s^0)$, $H(\vz)$ is a $2\times2$ matrix defined as 
\begin{equation}
    \label{def:H_matrix_proof}
  H(\vz):= \left[\begin{matrix}
      L_bd_b K_s + h_b-\phi_{bb}  & -\phi_{sb} (d_s L_b + 1) \\
      -\phi_{bs} (d_b L_s + 1) &L_sd_s K_b + h_s -\phi_{ss} 
  \end{matrix}\right],
\end{equation}
and, for any $k\in\{b,s\}$, $L_k$, $d_k$ and $h_k$ are functions depending on $z_k$ as
\begin{equation}
\begin{split}
L_k & = \frac{(N-1)\beta_k}{J_\phi}(1 + Ne^{z_k}),\\
d_k & = \beta_k(1 + N e^{z_k}),\\
h_k &= \beta_k(1 + e^{z_k}) (e^{-z_k} + N),  
\end{split}\label{eqn_z_Lk_proof}
\end{equation}
and $J_\phi$ is a function of $z_b$ and $z_s$ as
\begin{equation}
\begin{split}
    K_k & = \phi_{kk} - \beta_k (1+Ne^{z_k})(e^{-z_k} + N - 1), \ \text{ for } k \in\{b,s\}, \\
J_\phi &= K_bK_s - \phi_{sb}\phi_{bs}.
\end{split}\label{eqn:Jphi}
\end{equation}

Denoting $\vz^\ast$ to be the solution to \eqref{FOCs_z_proof2} and using \eqref{eqn:proof_xki_omega} and  \eqref{eqn:relationship_xp_z}, we conclude the proposition by noting that the symmetric equilibrium solution of $\eqref{pi}$ is given by $\vx^\ast = \Omega(\vz^\ast)$ and $\vp^\ast = \Phi \Omega(\vz^\ast) - \vbeta\vz^\ast -\vu^0$.
\end{proof}

{\bf Preliminary results for the proof of Proposition~\ref{prop:existence_gumbel}.}
We introduce notation and definitions and establish a useful lemma. Let $j\in\{1,\cdots,N\}$, $k,l$ and $m\in \{b,s\}$, $u_{k,l}:=\frac{\partial\phi_{k}}{\partial x_{b}}y_{b,l}+\frac{\partial\phi_{k}}{\partial x_{s}}y_{s,l}$ where $y_{k,l}$ is given by \eqref{eqn_y_RS}. Moreover, 
\begin{equation}\label{Uklmj}
    \begin{split}
        U_{k,lm}^{j}:= &\frac{\partial}{\partial x_{m}^{1}}\left[\frac{\partial T_{k}^{j}}{\partial u_{1}}\right]\left(-\frac{\partial p_{k}^{1}}{\partial x_{l}^{1}}+\frac{\partial\phi_{k}}{\partial x_{l}}\right)+\sum_{i=2}^{N}\frac{\partial}{\partial x_{m}^{1}}\left[\frac{\partial T_{k}^{j}}{\partial u_{i}}\right]u_{k,l}+
        \\ & \sum_{i=2}^{N}\frac{\partial T_{k}^{j}}{\partial u_{i}}\left(\frac{\partial^{2}\phi_{k}}{\partial x_{m}\partial x_{b}}\frac{\partial x_{b}^{i}}{\partial x_{l}^{1}}+\frac{\partial^{2}\phi_{k}}{\partial x_{m}\partial x_{s}}\frac{\partial x_{s}^{i}}{\partial x_{l}^{1}}\right), 
    \end{split}
\end{equation}
where the derivatives of $T_k^j$ are evaluated at $(u_{k}^{0},u_{k}^{1},\cdots,u_{k}^{N})$, $u_{k}^{1}:=\phi_{k}(\boldsymbol{x}^{1})-p_{k}^{1}$ and  $u_{k}^{i}:=\phi_{k}(\boldsymbol{x}^{i})-p_{k}^{*}$ for $i\geq2$. The following Lemma shows the second-order condition of \eqref{pi} as a function of $x_k$. 

\begin{lemma}[SOC of CNE]\label{lemma:soc_cneA}
The second-order condition of \eqref{pi} is given by 
\begin{equation}\label{eqn:s_SOC_b}
   D_{(x_b^1,x_s^1)}^{2}\pi^{1}=\left[\begin{array}{cc}
2\frac{\partial p_{b}^{1}}{\partial x_{b}^{1}}+x_{b}^{1}\frac{\partial^{2}p_{b}^{1}}{\partial\left(x_{b}^{1}\right)^{2}}+x_{s}^{1}\frac{\partial^{2}p_{s}^{1}}{\partial\left(x_{b}^{1}\right)^{2}} & \left(\frac{\partial p_{b}^{1}}{\partial x_{s}^{1}}+\frac{\partial p_{s}^{1}}{\partial x_{b}^{1}}\right)+x_{b}^{1}\frac{\partial^{2}p_{b}^{1}}{\partial x_{s}^{1}\partial x_{b}^{1}}+x_{s}^{1}\frac{\partial^{2}p_{s}^{1}}{\partial x_{s}^{1}\partial x_{b}^{1}}\\
\left(\frac{\partial p_{s}^{1}}{\partial x_{b}^{1}}+\frac{\partial p_{b}^{1}}{\partial x_{s}^{1}}\right)+x_{b}^{1}\frac{\partial^{2}p_{b}^{1}}{\partial x_{b}^{1}\partial x_{s}^{1}}+x_{s}^{1}\frac{\partial^{2}p_{s}^{1}}{\partial x_{b}^{1}\partial x_{s}^{1}} & 2\frac{\partial p_{s}^{1}}{\partial x_{s}^{1}}+x_{b}^{1}\frac{\partial^{2}p_{b}^{1}}{\partial\left(x_{s}^{1}\right)^{2}}+x_{s}^{1}\frac{\partial^{2}p_{s}^{1}}{\partial\left(x_{s}^{1}\right)^{2}}
\end{array}\right],
\end{equation}
where for $k,m$ and $l\in\left\{ b,s\right\}$, $\frac{\partial p_k^1}{\partial x_l^1}$ is given by \eqref{sln_ppbpxb}-\eqref{sln_ppspxs}, 
\begin{equation}\label{d2pkxm1xl1}
    \frac{\partial^{2}p_{k}^{1}}{\partial x_{m}^{1}\partial x_{l}^{1}}=\frac{\partial^{2}\phi_{k}}{\partial x_{m}\partial x_{l}}+\frac{1}{J_{k}}\left(\left(S_{k}+\left(N-2\right)R_{k}\right)U_{k,lm}^{1}+\left(x_{k,ml}-U_{k,lm}\right)\left(N-1\right)R_{k}\right),
\end{equation}
and $R_k$, $S_k$ and $J_k$ are given by \eqref{notation_prop32}. Moreover, 
\begin{equation}\label{d2xjkxm1xl1}
    \left[\begin{array}{c}
x_{b,ml}\\
x_{s,ml}
\end{array}\right]=\frac{1}{J_{\phi}}\left[\begin{array}{cc}
\left(\frac{\partial\phi_{s}}{\partial x_{s}}-\frac{1}{J_{s}}S_{s}\right) & -\frac{\partial\phi_{b}}{\partial x_{s}}\\
-\frac{\partial\phi_{s}}{\partial x_{b}} & \left(\frac{\partial\phi_{b}}{\partial x_{b}}-\frac{1}{J_{b}}S_{b}\right)
\end{array}\right]\left[\begin{array}{c}
\frac{1}{J_{b}}\left(R_{b}U_{b,lm}^{1}-S_{b}U_{b,lm}\right)\\
\frac{1}{J_{s}}\left(R_{s}U_{s,lm}^{1}-S_{s}U_{s,lm}\right)
\end{array}\right],
\end{equation}
where $U_{k,lm}^{1}$ and $U_{k,lm}:=U_{k,lm}^{j}$ for $j\geq 2$ are given by \eqref{Uklmj}.
\end{lemma}

The proof of this Lemma does not require assumptions I and II of Section \ref{sect:platform_competition}. Thus, the SOC given by (\ref{eqn:s_SOC_b}) is applicable to idiosyncratic preferences other than Gumbel distribution and to more general externality functions $\phi_k(\vx)$.

\begin{proof}[Proof of Lemma~\ref{lemma:soc_cneA}] Differentiating the left-hand side of \eqref{FOC_k} w.r.t.~$x_m^1$, for  $m\in\{b,s\}$, easily yields \eqref{eqn:s_SOC_b}. To obtain \eqref{d2pkxm1xl1} and \eqref{d2xjkxm1xl1}, we differentiate \eqref{pT1bpxb} and \eqref{pT2bpxb} w.r.t.~$x_m^1$. For $m, k, l\in\{b,s\}$ and $j\in\{2,3,\cdots, N\}$, we obtain
\begin{align}    
\frac{\partial T_{k}^{1}}{\partial u_{1}}\left(-\frac{\partial^{2}p_{k}^{1}}{\partial x_{m}^{1}\partial x_{l}^{1}}+\frac{\partial^{2}\phi_{k}}{\partial x_{m}^{1}\partial x_{l}}\right)+\sum_{i=2}^{N}\frac{\partial T_{k}^{1}}{\partial u_{i}}\left(\frac{\partial\phi_{k}}{\partial x_{b}}\frac{\partial^{2}x_{b}^{i}}{\partial x_{m}^{1}\partial x_{l}^{1}}+\frac{\partial\phi_{k}}{\partial x_{s}}\frac{\partial^{2}x_{s}^{i}}{\partial x_{m}^{1}\partial x_{l}^{1}}\right)+U_{k,lm}^{1}&=0, \label{d2T1x1m}\\
\frac{\partial T_{k}^{j}}{\partial u_{1}}\left(-\frac{\partial^{2}p_{k}^{1}}{\partial x_{m}^{1}\partial x_{l}^{1}}+\frac{\partial^{2}\phi_{k}}{\partial x_{m}^{1}\partial x_{l}}\right)+\sum_{i=2}^{N}\frac{\partial T_{k}^{j}}{\partial u_{i}}\left(\frac{\partial\phi_{k}}{\partial x_{b}}\frac{\partial^{2}x_{b}^{i}}{\partial x_{m}^{1}\partial x_{l}^{1}}+\frac{\partial\phi_{k}}{\partial x_{s}}\frac{\partial^{2}x_{s}^{i}}{\partial x_{m}^{1}\partial x_{l}^{1}}\right)+U_{k,lm}^{j}&=\frac{\partial^{2}x_{k}^{j}}{\partial x_{m}^{1}\partial x_{l}^{1}}, \label{d2Tjx1m}
\end{align}
where $U_{k,lm}^{1}$ and $U_{k,lm}^{j}$ are given by \eqref{Uklmj}. Note that the derivatives of $T_k^j$ in \eqref{d2T1x1m} and \eqref{d2Tjx1m}, are evaluated at $(u_{k}^{0},u_{k}^{1},\cdots,u_{k}^{N})$, $u_{k}^{1}:=\phi_{k}(\boldsymbol{x}^{1})-p_{k}^{1}$, where $u_{k}^{i}:=\phi_{k}(\boldsymbol{x}^{i})-p_{k}^{*}$ for $i\geq2$. 
The unknowns are $\frac{\partial^{2}p_{k}^{1}}{\partial x_{m}^{1}\partial x_{l}^{1}}$ and $\frac{\partial^{2}x_{k}^{i}}{\partial x_{m}^{1}\partial x_{l}^{1}}$ for $k,l,m\in\{b,s\}$ and $i\in\{2,\dots,N\}$, adding up to $8+8(N-1)=8N$. Similarly,  the number of equations in the system of equations described in \eqref{d2T1x1m} and \eqref{d2Tjx1m} is $8N$.

Next, we apply the property of symmetry at equilibrium, where we denote  
\begin{equation}\label{Uklm_xkml}
    \begin{split}
        U_{k,lm}^{j}&=:U_{k,lm}\ \textnormal{ and}\\
        \frac{\partial^{2}x_{k}^{j}}{\partial x_{m}^{1}\partial x_{l}^{1}}&=: x_{k,ml}, \textnormal{ for} j\in\{2,\dots,N\}.
    \end{split}
\end{equation}
Incorporating \eqref{notation_prop32}, \eqref{eqn_y_RS} and \eqref{Uklm_xkml}, we can reduce the system described by \eqref{d2T1x1m} and \eqref{d2Tjx1m} into 16 equations with unknowns: $\frac{\partial^{2}p_{k}^{1}}{\partial x_{m}^{1}\partial x_{l}^{1}}$ and $x_{k,ml}$ for $k,l,m\in\{b,s\}$.  For $k, m, l\in\{b,s\}$, the 16 equations are
\begin{equation}
    \begin{split}\label{pk2ml_soc_16eqns}
S_{k}q_k+\left(N-1\right)R_{k}t_k&=-U_{k,lm}^{1}, \\
R_{k}q_k+\left(S_{k}+\left(N-2\right)R_{k}\right)t_k&=x_{k,ml}-U_{k,lm}, 
    \end{split}
\end{equation}
where $q_k=\left(-\frac{\partial^{2}p_{k}^{1}}{\partial x_{m}^{1}\partial x_{l}^{1}}+\frac{\partial^{2}\phi_{k}}{\partial x_{m}^{1}\partial x_{l}}\right)$ and $t_k=\left(\frac{\partial\phi_{k}}{\partial x_{b}}x_{b,ml}+\frac{\partial\phi_{k}}{\partial x_{s}}x_{s,ml}\right)$. Solving the 2x2 system given by \eqref{pk2ml_soc_16eqns} easily yields \eqref{d2pkxm1xl1} and \eqref{d2xjkxm1xl1}.  
\end{proof}

Lemma~\ref{lemma:soc_cneA} is general enough to accommodate idiosyncratic preferences other than Gumbel distribution and general externality functions $\phi_k(\vx)$. Next, we use Assumptions I and II from Section \ref{sect:platform_competition}, Proposition~\ref{prop:FOC_gumbel} and Lemma~\ref{lemma:soc_cneA} to prove Proposition~\ref{prop:existence_gumbel}.

\begin{proof}[{\bf Proof of Proposition~\ref{prop:existence_gumbel}}]
    The proof has two main steps: (i) Verifying sufficient conditions for (\ref{FOCs_z}) to have a unique solution; (ii) Establishing that a second order condition is satisfied. 

Step (i):  Note that \eqref{FOCs_z} is equivalent to 
    \begin{equation}
        \label{FOCs_z2}
        \begin{split}
            \left(2\phi_{bb}-\phi_{ss}d_{b}L_{b}+h_{b}\right)\omega\left(z_{b}\right)+\left(\phi_{sb}+\phi_{bs}+\phi_{sb}d_{s}L_{b}\right)\omega\left(z_{s}\right)-u_{b}^{0}&=\beta_{b}z_{b}\ \textnormal{ and}\\
    \left(\phi_{sb}+\phi_{bs}+\phi_{bs}d_{b}L_{s}\right)\omega\left(z_{b}\right)+\left(2\phi_{ss}-\phi_{bb}d_{s}L_{s}+h_{s}\right)\omega\left(z_{s}\right)-u_{s}^{0}&=\beta_{s}z_{s},
        \end{split}
    \end{equation}
where for any $k\in\{b,s\}$, $L_k$, $d_k$ and $h_k$ are functions depending on $z_k$ given by \eqref{eqn_z_Lk_proof}. We want to find sufficient conditions for \eqref{FOCs_z2} to have a unique solution. First, we write some definitions. Set $\varphi_1=(\phi_{bs},\phi_{sb})$ and $\varphi_2=(\phi_{bb},\phi_{ss})$. For the given parameters in $\psi=\left(\beta_{b},\beta_{s},u_{b}^{0},u_{s}^{0},N\right)$, we define $M:\mathbb{R}^{6}\longrightarrow\mathbb{R}^{2}\ \textnormal{ as }$ 
\begin{equation}\label{def_M}
\begin{split}
M\left(z_{b},z_{s},\varphi_{1},\varphi_{2};\psi\right)&:=\left[\begin{array}{c}
M_{b}\left(z_{b},z_{s},\varphi_{1},\varphi_{2};\psi\right)\\
M_{s}\left(z_{b},z_{s},\varphi_{1},\varphi_{2};\psi\right)
\end{array}\right], \ \textnormal{ where for each } k,j\in \{b,s\},\ j\neq k,  \\
M_{k}\left(z_{b},z_{s},\varphi_{1},\varphi_{2};\psi\right)&:=\left(2\phi_{kk}-\phi_{jj}d_{k}L_{k}+h_{k}\right)\omega\left(z_{k}\right)+\left(\phi_{sb}+\phi_{bs}+\phi_{jk}d_{j}L_{k}\right)\omega\left(z_{j}\right)-u_{k}^{0}-\beta_k z_{k}.
\end{split}
\end{equation}  
We show that under \eqref{condition_existence_beta}, equation \eqref{FOCs_z2} has a unique solution for $\varphi_1 = 0$. From (\ref{def_M}), if we let $\varphi_1 = 0$, we obtain 
\begin{align}\label{eqn_Mk_varphi20}
M_{k}\left(z_{b},z_{s},0,\varphi_{2};\psi\right)=\left(2\phi_{kk}-\phi_{jj}d_{k}L_{k}+h_{k}\right)\omega\left(z_{k}\right)-u_{k}^{0}-\beta_k z_{k}.
\end{align}
Plugging \eqref{eqn_z_Lk_proof} into (\ref{eqn_Mk_varphi20}) gives
\begin{equation}\label{Mk_phi10}
\begin{split}
M_{k}\left(z_{b},z_{s},0,\varphi_{2};\psi\right)=&\frac{-\beta_{k}^{2}\left(1+Ne^{z_{k}}\right)^{3}+\beta_{k}\phi_{kk}e^{z_{k}}\left(\left(2N-1\right)e^{z_{k}}+3\right)\left(1+Ne^{z_{k}}\right)-2e^{2z_{k}}\phi_{kk}^{2}}{\left(1+Ne^{z_{k}}\right)\left(\beta_{k}\left(1+\left(N-1\right)e^{z_{k}}\right)\left(1+Ne^{z_{k}}\right)-e^{z_{k}}\phi_{kk}\right)}\\
&-\beta_k z_{k}-u_{k}^{0}.
\end{split}
\end{equation}
It follows that for $\varphi_1=0,$ $M_k$ does not depend on $z_j$ for $j\neq k$. Under \eqref{condition_existence_beta}, we claim that if $\varphi_1 = 0$, then the three statements below, (i-a)-(i-c), hold true:
\begin{itemize}
    \item[(i-a)] $M_{k}\left(z_{b},z_{s},0,\varphi_{2};\psi\right)$ is continuous on $z_k$ for all $z_k\in \R$.
    \item[(i-b)] $M_{k}\left(z_{b},z_{s},0,\varphi_{2};\psi\right)$ is strictly decreasing in the variable $z_k$ for all $z_k\in\R$.
    \item[(i-c)]The following limits hold true \begin{equation}
    \label{Mk_limits}
    \begin{split}
         \lim_{z_{k}\to-\infty}M_{k}\left(z_{b},z_{s},0,\varphi_{2};\psi\right)&=\infty \ \textnormal{ and } \\ \lim_{z_{k}\to\infty}M_{k}\left(z_{b},z_{s},0,\varphi_{2};\psi\right)&=-\infty.
    \end{split}
\end{equation}
\end{itemize}
Before proving the above claims, note that (i-a), (i-b) and (i-c) combined imply that there is a unique $(z_b^*,z_s^*)\in\R^2$ such that $$M\left(z_{b}^*,z_{s}^*,0,\varphi_{2};\psi\right)=\left[\begin{array}{c}
M_{b}\left(z_{b}^*,z_{s}^*,0,\varphi_{2};\psi\right)\\
M_{s}\left(z_{b}^*,z_{s}^*,0,\varphi_{2};\psi\right)
\end{array}\right]=0.$$
Thus, under \eqref{condition_existence_beta}, equation \eqref{FOCs_z2} has a unique solution for $\varphi_1 = 0$. By (i-b),
\begin{equation}
    \label{Mk_det}
\textnormal{det}\left(\frac{\partial\left(M_{b},M_{s}\right)}{\partial\left(z_{b},z_{s}\right)}\right)\Big|_{\left(z_{b}^{*},z_{s}^{*},0,\varphi_{2};\psi\right)}=\frac{\partial M_{b}\left(z_{b}^{*},z_{s}^{*},0,\varphi_{2};\psi\right)}{\partial z_{b}}\frac{\partial M_{s}\left(z_{b}^{*},z_{s}^{*},0,\varphi_{2};\psi\right)}{\partial z_{s}}>0.
\end{equation}
By \eqref{Mk_det} and the Implicit Function Theorem, there exists $\epsilon>0$ and a unique continuous function $$\left(z_{b}\left(\cdot\right),z_{s}\left(\cdot\right)\right):B_{\epsilon}\left(0,0\right)\longrightarrow\mathbb{R}^{2}$$ such that $\left(z_{b}\left(0,0\right),z_{s}\left(0,0\right)\right)=\left(z_{b}^{*},z_{s}^{*}\right)$. Moreover, for all $\varphi_{1}\in B_{\epsilon}\left(0,0\right)$, 
$$M\left(z_{b}\left(\varphi_{1}\right),z_{s}\left(\varphi_{1}\right),\varphi_{1},\varphi_{2};\psi\right)=0.$$
In particular, under \eqref{condition_existence_beta}, there exists $\epsilon>0$ such that for all $\varphi_{1}\in B_{\epsilon}\left(0,0\right)$, equation \eqref{FOCs_z2} has a unique solution. We now prove that under \eqref{condition_existence_beta}, if $\varphi_1 = 0$, then (i-a)-(i-c) hold true. 

\textit{Proof of (i-a)}. We prove that $M_{k}\left(z_{b},z_{s},0,\varphi_{2};\psi\right)$ is continuous on $z_k$ for all $z_k\in \R$. The auxiliary function $$g\left(z_{k}\right):= e^{-z_{k}}\left(1+\left(N-1\right)e^{z_{k}}\right)\left(1+Ne^{z_{k}}\right)$$ has a unique minimum at $z_k^0=-\frac{1}{2}\ln(N\left(N-1\right))$, with $g\left(z_k^0\right)=2\sqrt{N\left(N-1\right)}+2N-1$. It follows that if 
\begin{equation}\label{condition_continuity_Mk}
    \beta_{k}>\frac{\phi_{kk}}{(2\sqrt{N\left(N-1\right)}+2N-1)},
\end{equation} 
then, the denominator of the fraction in (\ref{Mk_phi10}) is never zero. We further prove (see the Mathematica file \emph{Gumbel\_N.nb}) that for all $\phi_{kk}>0$,
\begin{equation}
    f(N)\phi_{kk}\geq \frac{\phi_{kk}}{(2\sqrt{N\left(N-1\right)}+2N-1)} 
\end{equation}
(recall that $f(N)$ is stated in the proposition and given by \eqref{def:RNphi_kk_proof}). 
Thus, if $(\phi_{kk},\beta_k)$ satisfies (\ref{condition_existence_beta}), then \eqref{condition_continuity_Mk} is satisfied and $M_{k}\left(z_{b},z_{s},0,\varphi_{2};\psi\right)$ is continuous on $z_k$ for all $z_k\in \R$. 

\textit{Proof of (i-b)}. We prove that for each $k$, $M_{k}\left(z_{b},z_{s},0,\varphi_{2};\psi\right)$ is strictly decreasing in the variable $z_k$ for all $z_k\in\R$. In the supplementary file \emph{Gumbel\_N.nb}, we show that the partial derivative of $M_{k}\left(z_{b},z_{s},0,\varphi_{2};\psi\right)$ w.r.t. $z_k$ can be written as 
\begin{equation}\label{Mk_derivative}
    \frac{\partial M_{k}\left(z_{b},z_{s},0,\varphi_{2};\psi\right)}{\partial z_{k}}=-\frac{\sum_{m=0}^{6}a_{m}e^{mz_{k}}}{\left(1+Ne^{z_{k}}\right)^{2}\left(\beta_{k}\left(1+\left(N-1\right)e^{z_{k}}\right)\left(1+Ne^{z_{k}}\right)-e^{z_{k}}\phi_{kk}\right)^{2}},
\end{equation}
where the coefficients $\{a_m\}_{m=0}^6$ are polynomials on the parameters $\{\phi_{kk}, \beta_k, N\}$ and are given by 
\[\Scale[0.9]{a_{0}=\beta_{k}^{3},}\]\vspace{-22pt}
\[\Scale[0.9]{a_{1}=\beta_{k}^{2}\left(\beta_{k}\left(6N-1\right)-4\phi_{kk}\right),}\]
\[\Scale[0.9]{  a_{2}=\beta_{k}\left(\beta_{k}^{2}\left(15N^{2}-6N+1\right)+4\beta_{k}\left(1-4N\right)\phi_{kk}+5\phi_{kk}^{2}\right),}\]
\[\Scale[0.9]{a_{3}=2N\beta_{k}^{3}\left(10N^{2}-7N+2\right)+\beta_{k}^{2}\phi_{kk}\left(-24N^{2}+11N-1\right)+\beta_{k}\phi_{kk}^{2}\left(10N-3\right)-2\phi_{kk}^{3},}\]
\[\Scale[0.9]{a_{4}=\beta_{k}N\left(N\beta_{k}^{2}\left(15N^{2}-16N+6\right)+\beta_{k}\phi_{kk}\left(-16N^{2}+10N-2\right)+\left(5N-2\right)\phi_{kk}^{2}\right),}\]
\[\Scale[0.9]{a_{5}=\beta_{k}N^{2}\left(N\beta_{k}^{2}\left(6N^{2}-9N+4\right)+\beta_{k}\phi_{kk}\left(-4N^{2}+3N-1\right)+\phi_{kk}^{2}\right),}\]
\begin{equation}
    \label{eqn:am_coefficients}
    \small{a_{6}=\beta_{k}^{3}\left(N-1\right)^{2}N^{4}.}
\end{equation}
Because \eqref{condition_continuity_Mk} is satisfied, the denominator of (\ref{Mk_derivative}) is always positive. We then show (see the Mathematica file \emph{Gumbel\_N.nb}) that under \eqref{condition_existence_beta}, i.e., if $N\geq2$ and for each $k\in\left\{ b,s\right\}$, $\left(\phi_{kk},\beta_{k}\right)$ satisfies either 
\begin{equation*}
\left(\phi_{kk}\leq0\textnormal{ and }\beta_{k}>0\right)\textnormal{ or }\left(\phi_{kk}>0\textnormal{ and }\beta_{k}>f\left(N\right)\phi_{kk}\right),
\end{equation*}
then $a_{m}>0$ for all $m=0,\dots,6$, where $f(N)$ is given by \eqref{def:RNphi_kk_proof}. From (\ref{Mk_derivative}), it follows that 
\begin{equation}
    \label{eqn:ineDMk_neg}
    \frac{\partial M_{k}\left(z_{b},z_{s},0,\varphi_{2};\psi\right)}{\partial z_{k}}<0\  \textnormal{ for all } z_{k}. 
\end{equation}
It follows that $M_{k}\left(z_{b},z_{s},0,\varphi_{2};\psi\right)$ is strictly decreasing in the variable $z_k$ for all $z_k\in\R$.

\textit{Proof of (i-c)}. This claim follows from applying L'hopital rule in \eqref{Mk_phi10}.

Step (ii): We show that a second-order condition is satisfied. From \eqref{eqn:S_k_res_gumbel}, \eqref{eqn:R_k_res_gumbel} and the substitution $z_{k}=\frac{u_{k}-u_{k}^{0}}{\beta_{k}}$, we obtain for $i,r,j\in\{1,\dots,N\}$ and $k\in\{b,s\}$,
\begin{equation}
    \label{d2Tiuruj}
   \frac{\partial^{2}T^{i}_k\left(u_{k}^{0},u_{k},...,u_{k}\right)}{\partial u^{r}\partial u^{j}}=\frac{e^{z_{k}}}{\left(1+Ne^{z_{k}}\right)^{3}\beta_{k}^{2}}\cdot\begin{cases}
\left(1+\left(N-1\right)e^{z_{k}}\right)\left(1+\left(N-2\right)e^{z_{k}}\right) & i=j=r\\
-\left(1+\left(N-2\right)e^{z_{k}}\right)e^{z_{k}} & i=j,i\neq r\\
-\left(1+\left(N-2\right)e^{z_{k}}\right)e^{z_{k}} & i\neq j,j\neq r,i=r\\
2e^{2z_{k}}-\left(1+Ne^{z_{k}}\right)e^{z_{k}} & i\neq j=r\\
2e^{2z_{k}} & i\neq j,j\neq r,i\neq r
\end{cases}.
\end{equation}
In the supplementary file \emph{Gumbel\_N.nb}, we show that the matrix $D_{(x_b^1,x_s^1)}^{2}\pi^{1}|_{\varphi_1=0}$, as given by \eqref{eqn:s_SOC_b}, can be written as 
\begin{equation}\label{d2pi1xb1xs1_varphi10}
    D_{(x_b^1,x_s^1)}^{2}\pi^{1}\Big|_{\varphi_1=0}= \textnormal{diag}\left(
\frac{\partial^{2}\pi^{1}}{\partial\left(x_{b}^{1}\right)^{2}},
 \frac{\partial^{2}\pi^{1}}{\partial\left(x_{s}^{1}\right)^{2}}\right)\Big|_{\varphi_{1}=0},
\end{equation}
where
\begin{equation}\label{d2pi1d2k1}
    \frac{\partial^{2}\pi^{1}}{\partial\left(x_{k}^{1}\right)^{2}}\Big|_{\varphi_{1}=0}\text{=}\frac{\sum_{m=0}^{7}s_{m}e^{mz_{k}}}{e^{z_{k}}\left(\beta_{k}\left(\left(N-1\right)e^{z_{k}}+1\right)\left(Ne^{z_{k}}+1\right)-e^{z_{k}}\phi_{kk}\right)^{3}}.
\end{equation}
The coefficients $\{s_m\}_{m=0}^7$ are polynomials on the parameters $\{\phi_{kk}, \beta_k, N\}$ and are given by 
\[\Scale[0.85]{s_{0}=-\beta_{k}^{4},}\]\vspace{-22pt}
\[\Scale[0.85]{s_{1}=\beta_{k}^{3}\left(5\phi_{kk}+\beta_{k}\left(1-7N\right)\right),}\]
\[\Scale[0.85]{  s_{2}=-3\beta_{k}^{2}\left(\beta_{k}^{2}N\left(7N-2\right)+\beta_{k}\phi_{kk}\left(2-9N\right)+3\phi_{kk}^{2}\right),}\]
\[\Scale[0.85]{s_{3}=\beta_{k}\left(5\beta_{k}^{3}N^{2}\left(3-7N\right)+4\beta_{k}^{2}\left(N\left(15N-7\right)+1\right)\phi_{kk}+3\beta_{k}\left(3-11N\right)\phi_{kk}^{2}+7\phi_{kk}^{3}\right),}\]
\[\Scale[0.85]{s_{4}=5\beta_{k}^{4}N^{3}\left(4-7N\right)+2\beta_{k}^{3}N\left(N\left(35N-26\right)+7\right)\phi_{kk}+\beta_{k}^{2}\left(\left(26-45N\right)N-5\right)\phi_{kk}^{2}+\beta_{k}\left(13N-4\right)\phi_{kk}^{3}-2\phi_{kk}^{4},}\]
\[\Scale[0.85]{s_{5}=\beta_{k}\left(3\beta_{k}^{3}\left(5-7N\right)N^{4}+3\beta_{k}^{2}\left(N\left(15N-16\right)+6\right)N^{2}\phi_{kk}+\beta_{k}\left(\left(25-27N\right)N-10\right)N\phi_{kk}^{2}+\left(6N^{2}-4N+1\right)\phi_{kk}^{3}\right),}\]
\[\Scale[0.85]{s_{6}=\beta_{k}N\left(\beta_{k}^{3}\left(6-7N\right)N^{4}+\beta_{k}^{2}\left(N(15N-22)+10\right)N^{2}\phi_{kk}+\beta_{k}\left(-6N^{2}+8N-5\right)N\phi_{kk}^{2}+\phi_{kk}^{3}\right),}\]
\begin{equation}
    \label{eqn:sm_coefficients}
    \small{s_{7}=-\beta_{k}^{3}\left(N-1\right)N^{4}\left(\beta_{k}N^{2}+2\left(-N+1\right)\phi_{kk}\right).}
\end{equation}
Because \eqref{condition_continuity_Mk} is satisfied, the denominator of (\ref{d2pi1d2k1}) is always positive. We then show (see the Mathematica file \emph{Gumbel\_N.nb}) that under \eqref{condition_existence_beta}, i.e., if $N\geq2$ and for each $k\in\left\{ b,s\right\}$, $\left(\phi_{kk},\beta_{k}\right)$ satisfies either 
\begin{equation*}
\left(\phi_{kk}\leq0\textnormal{ and }\beta_{k}>0\right)\textnormal{ or }\left(\phi_{kk}>0\textnormal{ and }\beta_{k}>f\left(N\right)\phi_{kk}\right),
\end{equation*}
then $s_{m}<0$ for all $m=0,\dots,6$, where $f(N)$ is given by \eqref{def:RNphi_kk_proof}. It follows that $D_{(x_b^1,x_s^1)}^{2}\pi^{1}|_{\varphi_1=0}$ is negative definite. By continuity there exists $\tilde{\epsilon}>0$ such that for all $\varphi_{1}\in B_{\tilde{\epsilon}}\left(0,0\right)$, $$D_{(x_b^1,x_s^1)}^{2}\pi^{1}\left(z_{b}\left(\varphi_{1}\right),z_{s}\left(\varphi_{1}\right),\varphi_{1},\varphi_{2};\psi\right)$$ is negative definite. Therefore, a second condition for \eqref{FOCs_z_proof2} is satisfied.

\end{proof}

{\bf Preliminary result for the proof of Proposition~\ref{prop:FOC_gumbel_collusion}.} 
Using the same notation we introduced in \eqref{notation_prop32}, the following lemma provides the FOC of \eqref{pim} as a function of $x_k$.

\begin{lemma}[FOC of CE]\label{lemma:foc_ce}
The symmetric collusive equilibrium $\vp^\text{C}$ and $\vx^\text{C}$ are solutions of \eqref{xki} and of the following two equations
\begin{equation}
\begin{split}
    p_k + x_k\left( \frac{\partial \phi_k}{\partial x_k} - \frac{1}{S_k + (N-1)R_k}\right) + x_l\frac{\partial \phi_l}{\partial x_k} & = 0, \ \textnormal{ for } k,l\in\{b,s\}, k\neq l
\end{split} .   \label{M_FOC_b}
\end{equation}
\end{lemma}

The proof of this Lemma does not require assumptions I and II in Section \ref{sect:platform_competition}. Thus, the FOC given by (\ref{M_FOC_b}) is applicable to idiosyncratic preferences other than Gumbel distribution and to more general externality functions $\phi_k(\vx)$.

\begin{proof}[Proof of Lemma~\ref{lemma:foc_ce}]


The FOC of \eqref{pim} w.r.t.~$x_k$ is 
\begin{align}
    \frac{\partial \Pi_{\text{tot}}}{\partial x_k} & = N\left(p_k + x_k\frac{\partial p_k}{\partial x_k} + x_l\frac{\partial p_l}{\partial x_k}\right) = 0, \ \textnormal{ for } k,l\in\{b,s\}, k\neq l. \label{FOC_M_b}
\end{align}

To solve \eqref{FOC_M_b}, we need the expressions of 
\begin{equation}
\frac{\partial p_k}{\partial x_l},\ \textnormal{ for each } k, l \in  \{b, s\}.
\label{dpdx_M}
\end{equation}
We determine those four partial derivatives $\frac{\partial\vp^1}{\partial \vx^1}$ in \eqref{dpdx_M} using the definition of $T_k$ in \eqref{def:Qki}. By \eqref{xki_Tki}, for $k \in \{b, s\}$, the vectors of market shares and prices, $(x_k,\dots, x_k)$ and $(p_k,\dots, p_k)$ satisfy all the following
\begin{equation}
T^i_k(u^0_k, u_k, u_k, \dots, u_k) = x_k, \ \textnormal{ for } i\in \{1, 2, \cdots, N\}.    \label{eqn:T_func_M}
\end{equation}

Using the relationship $
u_k = \phi_k(\vx) - p_k,
$
it follows that 
$$
\frac{\partial u_k}{\partial x_l} = \frac{\partial \phi_k}{\partial x_l} - \frac{\partial p_k}{\partial x_l}.
$$

Taking derivative w.r.t.~$x_b$ and $x_s$ in \eqref{eqn:T_func_M} gives us
\begin{equation}\label{eqn:pkxl_collusion}
    \left( \sum_{j=1}^N \frac{\partial T_k^i}{\partial u^j}\right) \left( \frac{\partial \phi_k}{\partial x_l} - \frac{\partial p_k}{\partial x_l}\right)  = \delta_{kl}, \text{ for } i\in\{1, \cdots, N\},\ k,l \in\{b,s\}, 
\end{equation}
where $\delta_{kl} = 1$ when $k=l$ and $\delta_{kl}=0$ when $k\neq l$. From \eqref{notation_prop32} and \eqref{eqn:pkxl_collusion}, 
\begin{equation}\label{eqn:pkxl_collusion2}
    \frac{\partial p_k}{\partial x_l} = \frac{\partial \phi_k}{\partial x_l} - \frac{\delta_{kl}}{S_k + (N-1)R_k}.
\end{equation}
Plugging \eqref{eqn:pkxl_collusion} into \eqref{FOC_M_b} gives us 
\begin{align}
    p_k + x_k\left(\frac{\partial \phi_k}{\partial x_k} - \frac{1}{S_k + (N-1)R_k}\right) + x_l\frac{\partial \phi_l}{\partial x_k}  = 0, \ \textnormal{ for } k,l\in\{b,s\}, k\neq l. \label{FOC_M_b2}
\end{align}
\end{proof}

Lemma~\ref{lemma:foc_ce} is general enough to accommodate idiosyncratic preferences other than Gumbel distribution and general externality functions $\phi_k(\vx)$. Next, we use Assumptions I and II from Section 2 and Lemma~\ref{lemma:foc_ce} to prove Proposition~\ref{prop:FOC_gumbel_collusion}.

\begin{proof}[{\bf Proof of Proposition~\ref{prop:FOC_gumbel_collusion}}]

We want to rewrite the FOC given by \eqref{M_FOC_b} using Assumptions I and II from Section \ref{sect:platform_competition}. We plugging \eqref{eqn:Sk}, \eqref{eqn:Rk}, \eqref{eqn:proof_xki_omega} and \eqref{eqn:relationship_xp_z} from the proof of Proposition~\ref{prop:FOC_gumbel} into \eqref{M_FOC_b} to obtain 

\begin{equation}
     \vbeta \vz=\left(\Phi -H^C(\vz)\right)\Omega(\vz) - \vu_0 ,\label{FOCM_z_proof2}
    \end{equation}
    where $H^C(\vz)$ is a $2\times2$ matrix defined as 
\begin{equation}
    \label{def:HC_matrix_proof}
  H(\vz):= \left[\begin{matrix}
      \frac{\beta_b(1+Ne^{z_b})^2}{e^{z_b}} - \phi_{bb} & -\phi_{sb}\\
      -\phi_{bs} & \frac{\beta_s(1+Ne^{z_s})^2}{e^{z_s}} - \phi_{ss}
  \end{matrix}\right].
\end{equation}

Denoting $\vz^C$ to be the solution to \eqref{FOCM_z_proof2} and using \eqref{eqn:relationship_xp_z}, we conclude the proposition by noting that the symmetric equilibrium solution of $\eqref{pim}$ is given by $\vx^C = \Omega(\vz^C)$ and $\vp^C = \Phi \Omega(\vz^C) - \vbeta\vz^C -\vu_0$.
\end{proof}

{\bf Preliminary results for the proof of Proposition~\ref{coro:mono_uniqueness}.}
The following Lemma shows the second-order condition of \eqref{pim} as a function of $x_k$. 

\begin{lemma}[SOC of CNE]\label{lemma:soc_ce}
The second-order condition of \eqref{pim} is given by 
\begin{equation}\label{eqn:s_SOC_b_ce}
 \frac{\partial^{2}\Pi_{tot}}{\partial x_{m}\partial x_{k}}=N\left(\frac{\partial p_{k}}{\partial x_{m}}+\delta_{km}\frac{\partial p_{k}}{\partial x_{k}}+\delta_{ml}\frac{\partial p_{l}}{\partial x_{k}}+x_{k}\frac{\partial^{2}p_{k}}{\partial x_{m}\partial x_{k}}+x_{l}\frac{\partial^{2}p_{l}}{\partial x_{m}\partial x_{k}}\right), \textnormal{ for } k,l,m\in \{b,s\}, k\neq l,
\end{equation}
where $\frac{\partial p_k}{\partial x_l}$ is given by \eqref{eqn:pkxl_collusion2},
\begin{equation}\label{d2pkxm1xl1_ce}
    \frac{\partial^{2}p_{k}}{\partial x_{m}\partial x_{l}}=\frac{\partial^{2}\phi_{k}}{\partial x_{m}\partial x_{l}}+\frac{\delta_{km}\delta_{kl}}{\left(S_{k}+\left(N-1\right)R_{k}\right)^{3}}\left(\sum_{j=1}^{N}\sum_{r=1}^{N}\frac{\partial^{2}T_{k}^{i}\left(u_{k}^{0},u_{k},\cdots,u_{k}\right)}{\partial u^{r}\partial u^{j}}\right), \textnormal{ for } k,l,m\in \{b,s\},
\end{equation}
and $S_k$, $R_k$ are given by \eqref{notation_prop32}.
\end{lemma}

The proof of this Lemma does not require assumptions I and II of Section \ref{sect:platform_competition}. Thus, the SOC given by (\ref{eqn:s_SOC_b_ce}) is applicable to idiosyncratic preferences other than Gumbel distribution and to more general externality functions $\phi_k(\vx)$.

\begin{proof}[Proof of Lemma~\ref{lemma:soc_ce}] Differentiating the left-hand side of \eqref{FOC_M_b} w.r.t.~$x_m$, for $m\in \{b,s\}$, easily yields \eqref{eqn:s_SOC_b_ce}. To obtain \eqref{d2pkxm1xl1_ce}, we differentiate \eqref{eqn:pkxl_collusion} w.r.t.~$x_m$. For $m,k,l\in\{b,s\}$, we obtain 
\begin{equation}\label{d2pkxmxl_ce_proof}
    \sum_{j=1}^{N}\sum_{r=1}^{N}\frac{\partial^{2}T_{k}^{i}}{\partial u^{r}\partial u^{j}}\left(\frac{\partial\phi_{k}}{\partial x_{m}}-\frac{\partial p_{k}}{\partial x_{m}}\right)\left(\frac{\partial\phi_{k}}{\partial x_{l}}-\frac{\partial p_{k}}{\partial x_{l}}\right)+\left(\sum_{j=1}^{N}\frac{\partial T_{k}^{i}}{\partial u^{j}}\right)\left(\frac{\partial^{2}\phi_{k}}{\partial x_{m}\partial x_{l}}-\frac{\partial^{2}p_{k}}{\partial x_{m}\partial x_{l}}\right)=0,
\end{equation}
where the derivatives of $T_k^i$ are evaluated at $\left(u_{k}^{0},u_{k},...,u_{k}\right)$ and $u_k=\phi_k(\vx)-p_k$. Plugging \eqref{notation_prop32} and \eqref{eqn:pkxl_collusion2} into \eqref{d2pkxmxl_ce_proof}, yields \eqref{d2pkxm1xl1_ce}.
\end{proof} 

Lemma~\ref{lemma:soc_ce} is general enough to accommodate idiosyncratic preferences other than Gumbel distribution and general externality functions $\phi_k(\vx)$. Next, we use Assumptions I and II from Section \ref{sect:platform_competition}, Proposition~\ref{prop:FOC_gumbel_collusion} and Lemma~\ref{lemma:soc_ce} to prove Proposition~\ref{coro:mono_uniqueness}

\begin{proof}[{\bf Proof of Proposition~\ref{coro:mono_uniqueness}}]
The proof has two main steps: (i) Verifying sufficient conditions for \eqref{eqn:gumbel_mono_b} to have a unique solution; (ii) Establishing that a second order condition is satisfied.

Step (i): Note that \eqref{eqn:gumbel_mono_b} is equivalent to
\begin{equation}
\begin{split}
    \frac{2\phi_{bb}}{\beta_b(N + e^{-z_b})} + \frac{\phi_{bs} + \phi_{sb}}{\beta_b(N + e^{-z_s})} - \frac{u_b^0}{\beta_b} - (1+Ne^{z_b}) &= z_b \ \textnormal{ and} \\
\frac{\phi_{bs} + \phi_{sb}}{\beta_s(N + e^{-z_b})} +    \frac{2\phi_{ss}}{\beta_s(N + e^{-z_s})} - \frac{u_s^0}{\beta_s} -(1+Ne^{z_s}) &= z_s. \label{eqn:gumbel_mono_proof}
\end{split}
\end{equation}

For each $k\in\{b,s\}$, $k\neq l$, we denote 
\begin{align}
    F_k(z_b, z_s)&:= \frac{2\phi_{kk}}{\beta_k(N + e^{-z_k})} + \frac{\phi_{bs} + \phi_{sb}}{\beta_k(N + e^{-z_l})} - \frac{u_k^0}{\beta_k} - (1+Ne^{z_k}).
\end{align}
We want to find sufficient conditions for \eqref{eqn:gumbel_mono_proof} to have a unique solution. By bounding each term of $F_k(z_b, z_s)$ independently, we can identify an upper bound for $F_k(z_b, z_s)$ in $\sR^2$, indeed,
\begin{align}
    F_k(z_b, z_s) 
    \leq \frac{|2\phi_{kk}|}{\beta_kN} + \frac{|\phi_{bs} + \phi_{sb}|}{\beta_kN} - \frac{u_k^0}{\beta_k} - 1 := v_k.
    \label{eqn:upper_bd_Fb}
\end{align}
Similarly, if $z_k \leq v_k$, a lower bound for $F_k(z_b, z_s)$ is
\begin{align}
    F_k(z_b, z_s) 
    \geq -\frac{|2\phi_{kk}|}{\beta_kN} - \frac{|\phi_{bs} + \phi_{sb}|}{\beta_kN} - \frac{u_k^0}{\beta_k} - (1 + Ne^{v_k}) := w_k.
    \label{eqn:lower_bd_Fb}
\end{align}
We denote a vector-valued function $\mF(z_b, z_s) := (F_b(z_b, z_s), F_s(z_b, z_s))^T$. By combining \eqref{eqn:upper_bd_Fb} and \eqref{eqn:lower_bd_Fb}, we conclude that $\mF(z_b, z_s)$ maps the area $[w_b, v_b]\times [w_s, v_s]$ into $[w_b, v_b]\times [w_s, v_s]$. It is clear that $\mF(z_b, z_s)$ is a continuous function. By using Brouwer's Fixed-Point Theorem, there is a fixed point for the function $\mF(z_b, z_s)$ in the area $[w_b, v_b]\times [w_s, v_s]$, and this concludes that there is a solution for \eqref{eqn:gumbel_mono_b} in the area $[w_b, v_b]\times [w_s, v_s]$.

We now prove the uniqueness of the solution. We first consider $\phi_{bs}=\phi_{sb} =0$, then \eqref{eqn:gumbel_mono_b} becomes two decoupled equations. In particular, for each $k\in \{b,s\}$, $z_k^\text{C}$ is the solution to 
\begin{equation}\label{Mkm_phi10}
M_{k}^\text{C}\left(z_{b},z_{s},0,\varphi_{2};\psi\right):= \frac{2\phi_{kk}}{N+e^{-z_{k}}}-\beta_{k}\left(1+Ne^{z_{k}}\right)-u_k^0-\beta_{k}z_{k} = 0,
\end{equation}
where we are setting $\varphi_1=(\phi_{bs},\phi_{sb})=0$, $\varphi_2=(\phi_{bb},\phi_{ss})$ and $\psi=\left(\beta_{b},\beta_{s},u_{b}^{0},u_{s}^{0},N\right)$. From (\ref{Mkm_phi10}), the function $M_{k}^\text{C}\left(z_{b},z_{s},0,\varphi_{2};\psi\right)$ is continuous for all $z_k\in \mathbb{R}$. Moreover, \begin{equation}
    \label{Mkm_limits}
    \lim_{z_{k}\to-\infty}M_{k}^\text{C}\left(z_{b},z_{s},0,\varphi_{2};\psi\right)=\infty \ \textnormal{ and } \ \lim_{z_{k}\to\infty}M_{k}^\text{C}\left(z_{b},z_{s},0,\varphi_{2};\psi\right)=-\infty.
\end{equation}
The partial derivative of $M_{k}^\text{C}\left(z_{b},z_{s},0,\varphi_{2};\psi\right)$ w.r.t. $z_k$ can be written as 
\begin{equation}\label{Mkm_derivative}
    \frac{\partial M_{k}^\text{C}\left(z_{b},z_{s},0,\varphi_{2};\psi\right)}{\partial z_{k}}=\frac{2e^{z_{k}}\phi_{kk}-\beta_{k}\left(Ne^{z_{k}}+1\right)^{3}}{\left(Ne^{z_{k}}+1\right)^{2}}.
\end{equation}
Given that $\beta_k>0$, if $\phi_{kk}\leq 0$, then $M_{k}^\text{C}\left(z_{b},z_{s},0,\varphi_{2};\psi\right)$ is strictly decreasing w.r.t. $z_k$ for all $z_k\in\mathbb{R}$. Now, suppose that $\phi_{kk}>0$. The function $z_k\mapsto 2e^{z_{k}}/\left(Ne^{z_{k}}+1\right)^{3}$
has a unique maximum over $\mathbb{R}$ at $z_k^0 = \log\frac{1}{2N}$ and such maximum is given by $\frac{8}{27N}$. Therefore, the numerator of (\ref{Mkm_derivative}) is strictly negative whenever $\beta_k>\frac{8\phi_{kk}}{27N}$. Thus, if either $\phi_{kk}\leq 0$ or ($\phi_{kk}>0$ and $\beta_k>\frac{8\phi_{kk}}{27N}$), then
\begin{equation}\label{Mkm_derivative_neg}
    \frac{\partial M_{k}^\text{C}\left(z_{b},z_{s},0,\varphi_{2};\psi\right)}{\partial z_k}<0 \ \textnormal{ for all } z_k.
\end{equation}
It follows that there is a unique solution for (\ref{Mkm_phi10}). Now, notice that 
\begin{equation}
    \label{Mkm_det}
\textnormal{det}\left(\frac{\partial (M_b^C,M_s^C)}{\partial (z_b,z_s)}\right)\Big|_{\left(z_{b}^\text{C},z_{s}^\text{C},0,\varphi_{2};\psi\right)}=\frac{\partial M_{b}^\text{C}\left(z_{b}^\text{C},z_{s}^\text{C},0,\varphi_{2};\psi\right)}{\partial z_{b}}\frac{\partial M_{s}^\text{C}\left(z_{b}^\text{C},z_{s}^\text{C},0,\varphi_{2};\psi\right)}{\partial z_{s}}>0.
\end{equation}
By the Implicit Function Theorem, there exists $\epsilon>0$ and a unique continuous function $$(z_b^\text{C} (\varphi_1) , z_s^\text{C} (\varphi_1)): B_\epsilon (0, 0) \longrightarrow \sR^2$$ such that $(z_b^\text{C} (0, 0) , z_s^\text{C} (0, 0)) = (z_b^\text{C}, z_s^\text{C})$. Moreover, for all $\varphi_1 \in B_{\epsilon}(0, 0)$,
\begin{equation*}
 M^\text{C}(z_b^C(\varphi_1), z_s^C(\varphi_1), \varphi_1, \varphi_2; \psi) :=\left[\begin{array}{c}
M_{b}^\text{C}(z_b^C(\varphi_1), z_s^C(\varphi_1), \varphi_1, \varphi_2; \psi)\\
M_{s}^\text{C}(z_b^C(\varphi_1), z_s^C(\varphi_1), \varphi_1, \varphi_2; \psi)
\end{array}\right]=0.
\end{equation*}

Step (ii): We show that a second-order condition is satisfied. Combining \eqref{d2Tiuruj}, \eqref{eqn:s_SOC_b_ce} and \eqref{d2pkxm1xl1_ce}, we show (supplementary file \emph{Gumbel\_N.nb}) that 
\begin{equation}\label{d2pi1xb1xs1_varphi10_tot}
    D_{(x_b^1,x_s^1)}^{2}\Pi_{tot}= \begin{bmatrix}
      -Ne^{-z_b} \left(\beta_b \left(N e^{z_b}+1\right)^3-2 e^{z_b} \phi_{bb}\right)  & N (\phi_{bs}+\phi_{sb})\\
        N (\phi_{bs}+\phi_{sb}) & -Ne^{-z_s} \left(\beta_s \left(N e^{z_s}+1\right)^3-2 e^{z_s} \phi_{ss}\right)
    \end{bmatrix}.
\end{equation}

Following the same argument from the paragraph above \eqref{Mkm_derivative_neg}, if $\phi_{bs}=\phi_{sb}=0$, $N\geq2$ and for each $k\in\left\{ b,s\right\}$, $\left(\phi_{kk},\beta_{k}\right)$ satisfies either 
\begin{equation*}
\left(\phi_{kk}\leq0\textnormal{ and }\beta_{k}>0\right)\textnormal{ or }\left(\phi_{kk}>0\textnormal{ and }\beta_k>\frac{8\phi_{kk}}{27N}\right),
\end{equation*}
then $D_{(x_b^1,x_s^1)}^{2}\Pi_{tot}|_{\varphi_1=0}$ is negative definite. By continuity there exists $\tilde{\epsilon}>0$ such that for all $\varphi_{1}\in B_{\tilde{\epsilon}}\left(0,0\right)$, $$D_{(x_b^1,x_s^1)}^{2}\Pi_{tot}\left(z_{b}\left(\varphi_{1}\right),z_{s}\left(\varphi_{1}\right),\varphi_{1},\varphi_{2};\psi\right)$$ is negative definite. Therefore, a second condition for \eqref{FOCM_z_proof2} is satisfied.

\end{proof}

\begin{proof}[{\bf Proof of Proposition \ref{prop:sign_zk}}]
Suppose that $N\geq 2$ and for each $k\in\{b,s\}$, $(\phi_{kk},\beta_k)$ satisfies \eqref{condition_existence_beta}. Assume that $\varphi_1 = (\phi_{bs},\phi_{sb}) = 0$. From the proof of Proposition \ref{prop:existence_gumbel}, $M_k$ does not depend on $z_j$ for $j\neq k$ (see (\ref{Mk_phi10})). Moreover,  
\begin{equation}\label{eqn:Mk_ninfty}
  \lim_{z_{k}\to-\infty}M_{k}\left(z_{b},z_{s},0,\varphi_{2};\psi\right)=\infty, 
\end{equation}
and $M_{k}\left(z_{b},z_{s},0,\varphi_{2};\psi\right)$ is strictly decreasing in $z_k$ for all $z_k\in\R$. From \eqref{Mk_phi10}, we compute  
\begin{equation}\label{Mk_phi100}
\begin{split}
M_{k}\left(z_{b},z_{s},0,\varphi_{2};\psi\right)\Big|_{z_k = 0}&=-\frac{\beta_{k}^{2}\left(N+1\right)^{3}-2\phi_{kk}\beta_{k}\left(N+1\right)^{2}+2\phi_{kk}^{2}}{\left(N+1\right)\left(\beta_{k}N\left(N+1\right)-\phi_{kk}\right)}-u_{k}^{0}\\
&=-\frac{A_k}{\left(N+1\right)\left(\beta_{k}N\left(N+1\right)-\phi_{kk}\right)},
\end{split}
\end{equation}
where $A_k$ is a polynomial of order 2 in $\beta_{k}$ given by 
\begin{equation}
    \label{def:A}
A_k:=\beta_{k}^{2}\left(N+1\right)^{3}+\beta_{k}\left(N+1\right)^{2}\left(Nu_{k}^{0}-2\phi_{kk}\right)-\phi_{kk}\left(\left(N+1\right)u_{k}^{0}-2\phi_{kk}\right).
\end{equation}
The largest root of $A_k$ is
\begin{equation}\label{eqn:gamma}
    \gamma(N,\phi_{kk},u_k^0):= \frac{\left(2\phi_{kk}-Nu_{k}^{0}\right)+\sqrt{\left(2\phi_{kk}-Nu_{k}^{0}\right)^{2}+4\phi_{kk}\left(u_{k}^{0}-\frac{2\phi_{kk}}{N+1}\right)}}{2\left(N+1\right)}.
\end{equation}
We also denote the smallest root of $A_k$ by $\gamma_{-}$. Note that by (\ref{condition_existence_beta}), the denominator of (\ref{Mk_phi100}) is always positive.

Step (i): If $\beta_k>\gamma(N,\phi_{kk},u_k^0)$, by (\ref{Mk_phi100}), \eqref{def:A} and \eqref{eqn:gamma}, then (\ref{Mk_phi100}) is strictly negative. From (\ref{eqn:Mk_ninfty}), the fact that $M_{k}\left(z_{b},z_{s},0,\varphi_{2};\psi\right)$ is strictly decreasing in $z_k$ for all $z_k\in\R$, and that $z_k^*$ is the unique solution of $M_{k}\left(z_{b},z_{s},0,\varphi_{2};\psi\right)=0$, then $z_k^*<0$.

Step (ii): If $\gamma_{-}<\beta_k<\gamma(N,\phi_{kk},u_k^0)$, then (\ref{Mk_phi100}) is strictly positive. From (\ref{eqn:Mk_ninfty}), the fact that $M_{k}\left(z_{b},z_{s},0,\varphi_{2};\psi\right)$ is strictly decreasing in $z_k$ for all $z_k\in\R$, and that $z_k^*$ is the unique solution of $M_{k}\left(z_{b},z_{s},0,\varphi_{2};\psi\right)=0$, then $z_k^*>0$.

Note that from the definition of $\gamma_{-}$ as the smallest root of the second degree polynomial $A_k$, if $\phi_{kk}\leq 0$, then $\gamma_{-}\leq 0$. If $\phi_{kk}>0$, 
then $f(N)\phi_{kk}>\gamma_{-}$ (see the supplementary file \emph{Gumbel\_N.nb}). It follows that (\ref{condition_existence_beta}) combined with $\beta_k<\gamma(N,\phi_{kk},u_k^0)$ imply that $z_k^*>0$. 

To end the proof of this proposition, note that from the Proof of Proposition \ref{prop:existence_gumbel}, there is $\epsilon>0$ and a unique continuous function $$\left(z_{b}^{*}\left(\cdot\right),z_{s}^{*}\left(\cdot\right)\right):B_{\epsilon}\left(0,0\right)\longrightarrow\mathbb{R}^{2}$$ such that for all $\varphi_{1}\in B_{\epsilon}\left(0,0\right)$, $
M\left(z_{b}^{*}\left(\varphi_{1}\right),z_{s}^{*}\left(\varphi_{1}\right),\varphi_{1},\varphi_{2};\psi\right)=0.$  By continuity, for $\hat{\varepsilon}=|z_k^*(0,0)|/2$, there exists $\delta>0$ such that for all $\varphi_{1}\in B_{\min(\delta, \varepsilon)}\left(0,0\right)$,
\begin{equation*}
    z_{k}^{*}(0,0)-\frac{\left|z_{k}^{*}(0,0)\right|}{2}<z_{k}^{*}\left(\varphi_{1}\right)<z_{k}^{*}(0,0)+\frac{\left|z_{k}^{*}(0,0)\right|}{2}.
\end{equation*}
Thus, when $z_{k}^{*}(0,0)<0$, by step (i) we obtain $z_{k}^{*}\left(\varphi_{1}\right)<\frac{z_{k}^{*}(0,0)}{2}<0$. When $z_{k}^{*}>0$, by (ii) we obtain $z_{k}^{*}\left(\varphi_{1}\right)>\frac{z_{k}^{*}(0,0)}{2}>0$. 

\end{proof}

\begin{proof}[{\bf Proof of Corollary \ref{coro:sign_zk}}]
    We show that case (ii) in Proposition \ref{prop:sign_zk} is not feasible for large values of $u_k^0$. We first assume $\phi_{kk}\leq 0$ and $\beta_k>0$.  From \eqref{eqn:gamma}, 
\begin{equation}
    \label{sign_ofgamma}
    \textnormal{sign}\left(\gamma\left(N,\phi_{kk},u_{k}^{0}\right)\right)=\textnormal{sign}\left(\frac{2\phi_{kk}}{N+1}-u_{k}^{0}\right).
\end{equation}
Set $u_{k,1}^0:= 2\phi_{kk}/(N+1)$. It follows that if $\phi_{kk}\leq 0 $ and $u_k^0\geq u_{k,1}^0$, then $\gamma\left(N,\phi_{kk},u_{k}^{0}\right)\leq 0$. Thus, (ii) in Proposition \ref{prop:sign_zk} is not feasible as $\beta_k>0$.

Next, we assume that $\phi_{kk}>0$ and $\beta_{k}>f\left(N\right)\phi_{kk}$. We verify  (see \emph{Gumbel\_N.nb}) that for all $\phi_{kk}>0$ and $N\geq 2$,
\begin{equation}\label{uk02_condition}
\begin{cases}
\gamma\left(N,\phi_{kk},u_{k}^{0}\right)\leq f\left(N\right)\phi_{kk} & \textnormal{for any }u_{k}^{0}\geq u_{k,2}^{0},\\
f\left(N\right)\phi_{kk}<\gamma\left(N,\phi_{kk},u_{k}^{0}\right) & \textnormal{for any }u_{k}^{0}< u_{k,2}^{0},
\end{cases}
\end{equation}
where 
\begin{equation}
    \label{prop41_uk02}
    u_{k,2}^{0}:=-\frac{2 \left(N^4-2 N^3-2 N^2+2 N+2\right) \phi_{kk}}{N^3 \left(2 N^3+N^2-3 N-2\right)}.
\end{equation}
In particular, if $u_{k}^{0}\geq u_{k,2}^{0}$, then case (ii) in Proposition \ref{prop:sign_zk}  is not feasible as $\beta_{k}>f\left(N\right)\phi_{kk}$.

We finish the proof with a piecewise definition for $\tilde{u}_k^0$,
\begin{equation}
    \label{def:tilde_uk0}
    \tilde{u}_{k}^{0}:= \begin{cases}
        u_{k,1}^{0}& \textnormal{if } \phi_{kk}\leq 0,\\
        u_{k,2}^{0} & \textnormal{if } \phi_{kk}>0,
    \end{cases}
\end{equation}
where $u_{k,1}^{0} = 2\phi_{kk}/(N+1)$ and $u_{k,2}^{0}$ is given by \eqref{prop41_uk02}. 

\end{proof}

\begin{proof}[{\bf Proof of Corollary \ref{coro:sign_zk_N}}]

For each $k\in\{b,s\}$, let $u_k^0\in \sR$, $\Phi\in \mathbb{R}^{2x2}$ and $\beta_k>0$. From (\ref{FOCs_z2}), as $N\to\infty$, the FOC of \eqref{pi} becomes
\begin{equation*}
    -u_{k}^{0}-\beta_{k}\left(z_{k}+1\right)= 0, \textnormal{ for each } k\in\{b,s\}.
\end{equation*}
Thus, $\lim_{N\to\infty}z_{k}^{*}=-\left(\frac{u_{k}^{0}}{\beta_{k}}+1\right)$. Moreover, $\lim_{N\to\infty}Nx_{k}^{*}=1$ and $\lim_{N\to\infty}p_{k}^{*}=\beta_{k}$.
The solution of the equation $z_k^*=0$ when $N\to\infty$ is $\beta_k = -u_k^0$ if $u_k^0<0$. If $u_k^0\geq 0$, then $\lim_{N\to\infty}z_{k}^{*}<0$.

\end{proof}

\begin{proof}[{\bf Proof of Proposition \ref{prop:dpkduk0}}]
Recall from \eqref{eqn:relationship_xp_z} in the Proof of Proposition \ref{prop:FOC_gumbel}, $$\boldsymbol{p}^*=\Phi\Omega\left(\boldsymbol{z}^*\right)-\boldsymbol{\beta}\boldsymbol{z}^*-\boldsymbol{u}_{0},$$ where $\Omega(\vz^*) = (\omega(z_b^*), \omega(z_s^*))^T$, $\omega(\cdot)$ and $\vz^*$ are given by \eqref{eqn:proof_xki_omega} and \eqref{FOCs_z2}, respectively. We want to compute the following quantity when $\varphi_1=0$, 
\begin{equation}\label{eqn:dp_duk0}
      \frac{\partial p_{k}^{*}}{\partial u_{k}^{0}}\Big|_{\varphi_1=0}=\left[\frac{\phi_{kk}e^{z_{k}^{*}}}{\left(Ne^{z_{k}^{*}}+1\right)^{2}}-\beta_{k}\right]\frac{\partial z_{k}^{*}}{\partial u_{k}^{0}}-1, \quad   k\in \{b,s\}.  
\end{equation}
By \eqref{Mk_phi10}, $z_{k}^{*}$ is uniquely characterized by $M_{k}\left(z_{b}^*,z_{s}^*,0,\varphi_{2};\psi\right)+u_{k}^{0}-u_{k}^{0}=0$. It follows that 
\begin{equation}\label{eqn:dzk_uk0}
    \begin{split}
\frac{\partial z_{k}^{*}}{\partial u_{k}^{0}}=\left[\frac{\partial\left(M_{k}\left(z_{b}^*,z_{s}^*,0,\varphi_{2};\psi\right)+u_{k}^{0}\right)}{\partial z_{k}}\right]^{-1}=\left[\frac{\partial M_{k}\left(z_{b}^*,z_{s}^*,0,\varphi_{2};\psi\right)}{\partial z_{k}}\right]^{-1},
    \end{split}
\end{equation}
where $\frac{\partial M_{k}\left(z_{b}^*,z_{s}^*,0,\varphi_{2};\psi\right)}{\partial z_{k}}$ is given by \eqref{Mk_derivative}. After plugging (\ref{eqn:dzk_uk0}) into (\ref{eqn:dp_duk0}), we show (see \emph{Gumbel\_N.nb}) that  
\begin{equation}\begin{split}\label{eqn:dp_duk02}
      \frac{\partial p_{k}^{*}}{\partial u_{k}^{0}}\Big|_{\varphi_1=0}&=-\frac{n_{p,u}\left(z_{k}^{*},\beta_{k},\phi_{kk},N\right)}{d_{p,u}\left(z_{k}^{*},\beta_{k},\phi_{kk},N\right)},
\end{split}
\end{equation}
where $n_{p,u}\left(z_{k}^{*},\beta_{k},\phi_{kk},N\right)$ and $d_{p,u}\left(z_{k}^{*},\beta_{k},\phi_{kk},N\right)$ can be written as polynomials in $e^{z_{k}^{*}}$ in the following way, 
\begin{equation}\label{eqn:npu_dpu}
\begin{split}
n_{p,u}\left(z_{k}^{*},\beta_{k},\phi_{kk},N\right)&:= \sum_{m=1}^{5}n_{p,u,m}e^{mz_{k}^{*}} ,\quad  \textnormal{ and }\\
d_{p,u}\left(z_{k}^{*},\beta_{k},\phi_{kk},N\right)&:= \sum_{m=0}^{6}d_{p,u,m}e^{m z_{k}^{*}}.
\end{split}
\end{equation}
Moreover, the coefficients $n_{p,u,m}$ are as follows:
\begin{equation}\label{npum_coefficients}
    \begin{array}{c}
n_{p,u,1}=\beta_{k}^{2}\left(\beta_{k}-\phi_{kk}\right),\\
n_{p,u,2}=2\beta_{k}\left(2\beta_{k}^{2}N-2\beta_{k}N\phi_{kk}+\phi_{kk}^{2}\right),\\
n_{p,u,3}=6\beta_{k}^{3}N^{2}-N\beta_{k}^{2}\left(6N+1\right)\phi_{kk}+\phi_{kk}^{2}\beta_{k}\left(4N-1\right)-\phi_{kk}^{3},\\
n_{p,u,4}=2\beta_{k}N^{2}\left(\beta_{k}-\phi_{kk}\right)\left(2\beta_{k}N-\phi_{kk}\right), \textnormal{ and } \\
n_{p,u,5}=\beta_{k}N^{2}\left(\beta_{k}^{2}N^{2}-\phi_{kk}N\beta_{k}\left(N+1\right)+\phi_{kk}^{2}\right).
\end{array}
\end{equation}
The coefficients $d_{p,u,m}$ are given by \[\Scale[0.85]{
d_{p,u,0}=\beta_{k}^{3},}\]
\[\Scale[0.85]{d_{p,u,1}=\beta_{k}^{2}\left(\beta_{k}\left(6N-1\right)-4\phi_{kk}\right),}\]
\[\Scale[0.85]{d_{p,u,2}=\beta_{k}\left(\beta_{k}^{2}\left(15N^{2}-6N+1\right)+4\beta_{k}\left(1-4N\right)\phi_{kk}+5\phi_{kk}^{2}\right),}\]
\[\Scale[0.85]{d_{p,u,3}=N\beta_{k}^{3}\left(20N^{2}-14N+4\right)+\beta_{k}^{2}\left(-24N^{2}+11N-1\right)\phi_{kk}+\beta_{k}\left(10N-3\right)\phi_{kk}^{2}-2\phi_{kk}^{3},}\]
\[\Scale[0.85]{d_{p,u,4}=\beta_{k}N\left(N\beta_{k}^{2}\left(15N^{2}-16N+6\right)+\beta_{k}\left(-16N^{2}+10N-2\right)\phi_{kk}+\left(5N-2\right)\phi_{kk}^{2}\right),}\]
\[\Scale[0.85]{d_{p,u,5}=\beta_{k}N^{2}\left(N\beta_{k}^{2}\left(6N^{2}-9N+4\right)+\beta_{k}\left(-4N^{2}+3N-1\right)\phi_{kk}+\phi_{kk}^{2}\right), \textnormal{ and }}\]
\begin{equation}\label{dp_um_coefficients}
    \small{d_{p,u,6}=\beta_{k}^{3}\left(N-1\right)^{2}N^{4}}.
\end{equation}

Because the expressions determining $n_{p,u}\left(z_{k}^{*},\beta_{k},\phi_{kk},N\right)$ and $d_{p,u}\left(z_{k}^{*},\beta_{k},\phi_{kk},N\right)$ are complex, we focus on finding sufficient conditions for these expressions to have a specific sign for all $z_{k}^{*}$. 

Case (i): $\frac{\partial p_{k}^{*}}{\partial u_k^0}\Big|_{\varphi_{1}=0}<0$. 
We verify in the supplementary file \emph{Gumbel\_N.nb}
that $n_{p,u}$ and $d_{p,u}$ (see \eqref{eqn:npu_dpu}) are positive, if either of the two conditions below, (i-a) or (i-b), hold: 
\begin{itemize}
    \item[(i-a)] $\phi_{kk}\leq 0$, $N\geq 2$ and $\beta_{k}>0$.
    \item[(i-b)] $\phi_{kk}>0$, $N\geq 2$ and $\beta_{k}>g_{p,u}(N, \phi_{kk})$, where $g_{p}(N,\phi_{kk})$ is the largest real root of the third degree polynomial $n_{p,u,5}$ (viewed as a polynomial in $\beta_k$). 
\end{itemize}

    Using the quadratic formula, we verify that $n_{p,u,5}$ (see \eqref{npum_coefficients}) has three real roots and that $g_{p,u}(N, \phi_{kk})$ is linear in $\phi_{kk}$ and can thus be expressed as  $g_{p,u}(N, \phi_{kk})= g_{p,u}\left(N\right)\phi_{kk}$ where
    $$g_{p,u}\left(N\right):=\frac{\left(N+\sqrt{\left(N-1\right)\left(N+3\right)}+1\right)}{2N}.$$

Case (ii): $\frac{\partial p_{k}^{*}}{\partial N}\Big|_{\varphi_{1}=0}>0$. We verify in the supplementary file \emph{Gumbel\_N.nb} that $n_{p,u}$ and $d_{p,u}$ (see \eqref{eqn:npu_dpu}) are negative and positive, respectively, if the condition below, (ii), holds:
\begin{itemize}
    \item[(ii)] $\phi_{kk}>0$, $N\geq 3$ and $f(N)\phi_{kk}<\beta_{k}<f_{p,u}\left(N,\phi_{kk}\right)$, where $f_{p,u}\left(N,\phi_{kk}\right)$ is the largest real root of the third degree polynomial $n_{p,u,2}$ (viewed as a polynomial in $\beta_k$). 
\end{itemize}

    Using the quadratic formula, we verify that $n_{p,u,2}$ (see \eqref{npum_coefficients}) has three real roots and that $f_{p,u}(N, \phi_{kk})$ is linear in $\phi_{kk}$ and can thus be expressed as  $f_{p,u}(N, \phi_{kk})= f_{p,u}\left(N\right)\phi_{kk}$ where
    $$f_{p,u}\left(N\right):=\frac{1}{2}\left(\sqrt{\frac{N-2}{N}}+1\right).$$

We now show that the limits in  \eqref{limitspkuk0} hold. From \eqref{Mk_limits} and \eqref{eqn:ineDMk_neg}, the function $M_{k}\left(z_{b},z_{s},0,\varphi_{2};\psi\right)+u_{k}^{0}$ is strictly decreasing w.r.t.~$z_{k}$ and it approaches $\pm\infty$ as $z_{k}$ approaches $\mp\infty$. Thus, if $z_{k}^{*}$ is the unique solution to $M_{k}\left(z_{b}^*,z_{s}^*,0,\varphi_{2};\psi\right)+u_{k}^{0}=u_{k}^{0}$, as $u_{k}^{0}\to\infty$, $M_{k}\left(z_{b}^*,z_{s}^*,0,\varphi_{2};\psi\right)\to\infty$ and therefore, $z_{k}^{*}\to-\infty$. It follows that $\lim_{u_{k}^{0}\to\infty}z_{k}^{*}=-\infty$. Similarly, $\lim_{u_{k}^{0}\to-\infty}z_{k}^{*}=\infty$. From \eqref{eqn:relationship_xp_z} and \eqref{Mk_phi10}, it follows that
\begin{equation}
\label{def:pk_gumbel_simplified}p_{k}^{*}=\frac{\left(\beta_{k}+e^{z_{k}^{*}}\left(\beta_{k}N-\phi_{kk}\right)\right)\left(\beta_{k}\left(Ne^{z_{k}^{*}}+1\right)^{2}-e^{z_{k}^{*}}\phi_{kk}\right)}{\left(1+Ne^{z_{k}^{*}}\right)\left(\beta_{k}\left(1+\left(N-1\right)e^{z_{k}^{*}}\right)\left(1+Ne^{z_{k}^{*}}\right)-e^{z_{k}^{*}}\phi_{kk}\right)}.
\end{equation}
Taking limits in the above expression yields, 
\begin{equation*}
    \begin{split}
        \lim_{u_{k}^{0}\to-\infty}p_{k}^{*}&=\frac{N}{N-1}\beta_{k}-\frac{\phi_{kk}}{N-1}, \textnormal{ and}\\
\lim_{u_{k}^{0}\to\infty}p_{k}^{*}&=\beta_{k}.
    \end{split}
\end{equation*}

Finally, we show that there exists $\epsilon>0$ such that for any $(\phi_{bs},\phi_{sb})\in B_\epsilon(0)$, (i) and (ii) in Proposition \ref{prop:dpkduk0} hold. From \eqref{eqn:proof_xki_omega}, \eqref{eqn:relationship_xp_z} and \eqref{FOCs_z2}, $\partial p_k^*/\partial u_k^0$ is a rational function w.r.t.~$(\phi_{bs},\phi_{sb})$. Moreover, at  $(\phi_{bs},\phi_{sb})=(0,0)$, the partial derivative $\partial p_k^*/\partial u_k^0$ is given by \eqref{eqn:dp_duk0}. Thus, $\partial p_k^*/\partial u_k^0$ is continuous w.r.t.~$(\phi_{bs},\phi_{sb})$  at  $(0,0)$. We conclude that there exists $\epsilon>0$ such that for any $(\phi_{bs},\phi_{sb})\in B_\epsilon(0)\subset \sR^2$, cases (i) and (ii) in Proposition \ref{prop:dpkduk0} hold true. 
\end{proof}

\begin{proof}[{\bf Proof of Proposition \ref{prop:dpikduk0}}]
Using \eqref{def:pk_gumbel_simplified} from the Proof of Proposition \ref{prop:dpkduk0} and \eqref{eqn:proof_xki_omega}, we can write 
\begin{equation}\label{def:pik_gumbel_simplified}
     \pi_{k}^{*}\Big |_{\varphi_1=0}=\frac{e^{z_{k}^{*}}\left(\beta_{k}+e^{z_{k}^{*}}\left(\beta_{k}N-\phi_{kk}\right)\right)\left(\beta_{k}\left(Ne^{z_{k}^{*}}+1\right)^{2}-e^{z_{k}^{*}}\phi_{kk}\right)}{\left(Ne^{z_{k}^{*}}+1\right)^{2}\left(\beta_{k}\left(\left(N-1\right)e^{z_{k}^{*}}+1\right)\left(Ne^{z_{k}^{*}}+1\right)-e^{z_{k}^{*}}\phi_{kk}\right)}, \quad   k\in \{b,s\}.  
\end{equation}
It follows that $\frac{\partial\pi_{k}^{*}}{\partial u_{k}^{0}}\Big |_{\varphi_1=0}=\frac{\partial\pi_{k}^{*}}{\partial z_{k}^{*}}\frac{\partial z_{k}^{*}}{\partial u_{k}^{0}}$, where $\frac{\partial z_{k}^{*}}{\partial u_{k}^{0}}$ is given by \eqref{eqn:dzk_uk0}. We show (see \emph{Gumbel\_N.nb}) that  
\begin{equation}\begin{split}\label{eqn:dpi_duk0}
      \frac{\partial\pi_{k}^{*}}{\partial u_{k}^{0}}\Big|_{\varphi_1=0}&=-\frac{n_{\pi,u}\left(z_{k}^{*},\beta_{k},\phi_{kk},N\right)}{d_{\pi,u}\left(z_{k}^{*},\beta_{k},\phi_{kk},N\right)},
\end{split}
\end{equation}
where $n_{\pi,u}\left(z_{k}^{*},\beta_{k},\phi_{kk},N\right)$ and $d_{\pi,u}\left(z_{k}^{*},\beta_{k},\phi_{kk},N\right)$ can be written as polynomials in $e^{z_{k}^{*}}$ in the following way: 
\begin{equation}\label{eqn:npi_u_dpi_u}
\begin{split}
n_{\pi,u}\left(z_{k}^{*},\beta_{k},\phi_{kk},N\right)&:= \sum_{m=1}^{6}n_{\pi,u,m}e^{mz_{k}^{*}} ,\quad  \textnormal{ and }\\
d_{\pi,u}\left(z_{k}^{*},\beta_{k},\phi_{kk},N\right)&:= \sum_{m=0}^{7}d_{\pi,u,m}e^{mz_{k}^{*}}.
\end{split}
\end{equation}
Moreover, the coefficients $n_{\pi,u,m}$ are as follows:
\begin{equation}\label{npi_um_coefficients}
    \begin{array}{c}
n_{\pi,u,1}=\beta_{k}^{3},\\
n_{\pi,u,2}=\beta_{k}^{2}\left(5\beta_{k}N-4\phi_{kk}\right),\\
n_{\pi,u,3}=\beta_{k}\left(10\beta_{k}^{2}N^{2}+2\beta_{k}\left(1-7N\right)\phi_{kk}+5\phi_{kk}^{2}\right),\\
n_{\pi,u,4}=10\beta_{k}^{3}N^{3}+2N\beta_{k}^{2}\left(2-9N\right)\phi_{kk}+\beta_{k}\left(9N-2\right)\phi_{kk}^{2}-2\phi_{kk}^{3}, \\
n_{\pi,u,5}=\beta_{k}N\left(5\beta_{k}^{2}N^{3}+2N\beta_{k}\left(1-5N\right)\phi_{kk}+\left(4N-1\right)\phi_{kk}^{2}\right) \textnormal{ and } \\
n_{\pi,u,6}=\beta_{k}N^{2}\left(\beta_{k}^{2}N^{3}-2\beta_{k}N^{2}\phi_{kk}+\phi_{kk}^{2}\right).
\end{array}
\end{equation}
The coefficients $d_{\pi,u,m}$ are given by \[\Scale[0.85]{
d_{\pi,u,0}=\beta_{k}^{3},}\]
\[\Scale[0.85]{d_{\pi,u,1}=\beta_{k}^{2}\left(\beta_{k}\left(7N-1\right)-4\phi_{kk}\right),}\]
\[\Scale[0.85]{d_{\pi,u,2}=\beta_{k}\left(\beta_{k}^{2}\left(21N^{2}-7N+1\right)+4\beta_{k}\left(1-5N\right)\phi_{kk}+5\phi_{kk}^{2}\right),}\]
\[\Scale[0.85]{d_{\pi,u,3}=N\beta_{k}^{3}\left(35N^{2}-20N+5\right)+\beta_{k}^{2}\left(-40N^{2}+15N-1\right)\phi_{kk}+\beta_{k}\left(15N-3\right)\phi_{kk}^{2}-2\phi_{kk}^{3},}\]
\[\Scale[0.85]{d_{\pi,u,4}=N^{2}\beta_{k}^{3}\left(35N^{2}-30N+10\right)+N\beta_{k}^{2}\left(-40N^{2}+21N-3\right)\phi_{kk}+5N\beta_{k}\left(3N-1\right)\phi_{kk}^{2}-2N\phi_{kk}^{3},}\]
\[\Scale[0.85]{d_{\pi,u,5}=\beta_{k}N^{2}\left(N\beta_{k}^{2}\left(21N^{2}-25N+10\right)+\beta_{k}\left(-20N^{2}+13N-3\right)\phi_{kk}+\left(5N-1\right)\phi_{kk}^{2}\right),}\]
\[\Scale[0.85]{d_{\pi,u,6}=\beta_{k}N^{3}\left(N\beta_{k}^{2}\left(7N^{2}-11N+5\right)+\beta_{k}\left(-4N^{2}+3N-1\right)\phi_{kk}+\phi_{kk}^{2}\right), \textnormal{ and }}\]
\begin{equation}\label{dpi_um_coefficients}
    \small{d_{\pi,u,7}=\beta_{k}^{3}\left(N-1\right)^{2}N^{5}}.
\end{equation}

Because the expressions determining $n_{\pi,u}\left(z_{k}^{*},\beta_{k},\phi_{kk},N\right)$ and $d_{\pi,u}\left(z_{k}^{*},\beta_{k},\phi_{kk},N\right)$ are complex, we focus on finding sufficient conditions for these expressions to have a specific sign for all $z_{k}^{*}$. We verify in the supplementary file \emph{Gumbel\_N.nb}
that $n_{\pi,u}$ and $d_{\pi,u}$ (see \eqref{eqn:npi_u_dpi_u}) are positive, if either of the two conditions below, (i-a) or (i-b), hold: 
\begin{itemize}
    \item[(i-a)] $\phi_{kk}\leq 0$, $N\geq 2$ and $\beta_{k}>0$.
    \item[(i-b)] $\phi_{kk}>0$, $N\geq 2$ and $\beta_{k}>g_{\pi,u}(N, \phi_{kk})$, where $g_{\pi,u}(N,\phi_{kk})$ is the largest real root of the third degree polynomial $n_{\pi,u,6}$ (viewed as a polynomial in $\beta_k$). 
\end{itemize}

    Using the quadratic formula, we verify that $n_{\pi,u,6}$ (see \eqref{npi_um_coefficients}) has three real roots and that $g_{\pi,u}(N, \phi_{kk})$ is linear in $\phi_{kk}$ and can thus be expressed as  $g_{\pi,u}(N, \phi_{kk})= g_{\pi,u}\left(N\right)\phi_{kk}$ where
    $$g_{\pi,u}\left(N\right):=\sqrt{\frac{N-1}{N^{3}}}+\frac{1}{N}.$$

We now show that the limits in  \eqref{limitspikuk0} hold. From the Proof of Proposition \ref{prop:dpikduk0}, we have that $\lim_{u_{k}^{0}\to\infty}z_{k}^{*}=-\infty$ and $\lim_{u_{k}^{0}\to-\infty}z_{k}^{*}=\infty$ (see the paragraph above \eqref{def:pk_gumbel_simplified}). Taking limits in \eqref{def:pik_gumbel_simplified} yields, 
\begin{equation*}
    \begin{split}
        \lim_{u_{k}^{0}\to-\infty}\pi_{k}^{*}&=\frac{1}{N-1}\beta_{k}-\frac{\phi_{kk}}{(N-1)N}, \textnormal{ and}\\
\lim_{u_{k}^{0}\to\infty}\pi_{k}^{*}&=0.
    \end{split}
\end{equation*}

Finally, we show that there exists $\epsilon>0$ such that for any $(\phi_{bs},\phi_{sb})\in B_\epsilon(0)$, the first part in Proposition \ref{prop:dpikduk0} holds. From \eqref{eqn:proof_xki_omega}, \eqref{eqn:relationship_xp_z} and \eqref{FOCs_z2}, $\partial \pi_k^*/\partial u_k^0$ is a rational function w.r.t.~$(\phi_{bs},\phi_{sb})$. Moreover, at  $(\phi_{bs},\phi_{sb})=(0,0)$, the partial derivative $\partial \pi_k^*/\partial u_k^0$ is given by \eqref{eqn:dpi_duk0}. Thus, $\partial \pi_k^*/\partial u_k^0$ is continuous w.r.t.~$(\phi_{bs},\phi_{sb})$  at  $(0,0)$. We conclude that there exists $\epsilon>0$ such that for any $(\phi_{bs},\phi_{sb})\in B_\epsilon(0)\subset \sR^2$, the first part of Proposition \ref{prop:dpikduk0} holds true. 
\end{proof}

\begin{proof}[{\bf Proof of Proposition \ref{prop:dCSkuk0}}]
Plugging \eqref{def:pk_gumbel_simplified} from the Proof of Proposition \ref{prop:dpkduk0} and \eqref{eqn:proof_xki_omega} into \eqref{eqn:equilibrium_CS}, we obtain
\begin{equation}\label{def:CSk_gumbel_simplified}
\begin{split}
     CS_{k}^{*}\Big |_{\varphi_1=0}=&\frac{-\beta_{k}^{2}\left(1+Ne^{z_{k}^{*}}\right)^{3}+\beta_{k}\phi_{kk}e^{z_{k}^{*}}\left(\left(2N-1\right)e^{z_{k}^{*}}+3\right)\left(1+Ne^{z_{k}^{*}}\right)-2e^{2z_{k}^{*}}\phi_{kk}^{2}}{\left(1+Ne^{z_{k}^{*}}\right)\left(\beta_{k}\left(1+\left(N-1\right)e^{z_{k}^{*}}\right)\left(1+Ne^{z_{k}^{*}}\right)-e^{z_{k}^{*}}\phi_{kk}\right)}\\
     &+\mu_{k}+\beta_{k}\left(\ln\left(N+1\right)+\gamma\right),\quad k\in \{b,s\}.
     \end{split}
\end{equation}
It follows that $\frac{\partial CS_{k}^{*}}{\partial u_{k}^{0}}\Big |_{\varphi_1=0}=\frac{\partial CS_{k}^{*}}{\partial z_{k}^{*}}\frac{\partial z_{k}^{*}}{\partial u_{k}^{0}}$, where $\frac{\partial z_{k}^{*}}{\partial u_{k}^{0}}$ is given by \eqref{eqn:dzk_uk0}. We show (see \emph{Gumbel\_N.nb}) that  
\begin{equation}\begin{split}\label{eqn:dCS_duk0}
      \frac{\partial CS_{k}^{*}}{\partial u_{k}^{0}}\Big|_{\varphi_1=0}&=\frac{n_{CS,u}\left(z_{k}^{*},\beta_{k},\phi_{kk},N\right)}{d_{CS,u}\left(z_{k}^{*},\beta_{k},\phi_{kk},N\right)},
\end{split}
\end{equation}
where $n_{CS,u}\left(z_{k}^{*},\beta_{k},\phi_{kk},N\right)$ and $d_{CS,u}\left(z_{k}^{*},\beta_{k},\phi_{kk},N\right)$ can be written as polynomials in $e^{z_{k}^{*}}$ in the following way, 
\begin{equation}\label{eqn:nCS_u_dCS_u}
\begin{split}
n_{CS,u}\left(z_{k}^{*},\beta_{k},\phi_{kk},N\right)&:= \sum_{m=1}^{5}n_{CS,u,m}e^{mz_{k}^{*}} ,\quad  \textnormal{ and }\\
d_{CS,u}\left(z_{k}^{*},\beta_{k},\phi_{kk},N\right)&:= \sum_{m=0}^{6}d_{CS,u,m}e^{mz_{k}^{*}}.
\end{split}
\end{equation}
Moreover, the coefficients $n_{CS,u,m}$ are as follows:
\begin{equation}\label{nCS_um_coefficients}
    \begin{array}{c}
n_{CS,u,1}=\beta_{k}^{2}\left(\beta_{k}-2\phi_{kk}\right),\\
n_{CS,u,2}=2\beta_{k}\left(2\beta_{k}^{2}N+\beta_{k}\left(1-4N\right)\phi_{kk}+2\phi_{kk}^{2}\right),\\
n_{CS,u,3}=6\beta_{k}^{3}N^{2}+\beta_{k}^{2}\left(-12N^{2}+5N-1\right)\phi_{kk}+\beta_{k}\left(8N-3\right)\phi_{kk}^{2}-2\phi_{kk}^{3},\\
n_{CS,u,4}=2\beta_{k}N\left(2\beta_{k}^{2}N^{2}+\beta_{k}\left(-4N^{2}+2N-1\right)\phi_{kk}+\left(2N-1\right)\phi_{kk}^{2}\right) \textnormal{ and } \\
n_{CS,u,5}=\beta_{k}N^{2}\left(\beta_{k}^{2}N^{2}+\beta_{k}\left(-2N^{2}+N-1\right)\phi_{kk}+\phi_{kk}^{2}\right).
\end{array}
\end{equation}
The coefficients $d_{CS,u,m} = d_{p,u,m}$ for all $m\in \{0,\dots, 6\}$ where $d_{p,u,m}$ is given by \eqref{dp_um_coefficients}. Because the expressions determining $n_{CS,u}\left(z_{k}^{*},\beta_{k},\phi_{kk},N\right)$ and $d_{CS,u}\left(z_{k}^{*},\beta_{k},\phi_{kk},N\right)$ are complex, we focus on finding sufficient conditions for these expressions to have a specific sign for all $z_{k}^{*}$. 

Case (i): $\frac{\partial CS_{k}^{*}}{\partial u_k^0}\Big|_{\varphi_{1}=0}>0$. 
We verify in the supplementary file \emph{Gumbel\_N.nb}
that $n_{CS,u}$ and $d_{CS,u}$ (see \eqref{eqn:nCS_u_dCS_u}) are positive, if either of the two conditions below, (i-a) or (i-b), hold: 
\begin{itemize}
    \item[(i-a)] $\phi_{kk}\leq 0$, $N\geq 2$ and $\beta_{k}>0$.
    \item[(i-b)] $\phi_{kk}>0$, $N\geq 2$ and $\beta_{k}>2 \phi_{kk}$. 
\end{itemize}

Case (ii): $\frac{\partial CS_{k}^{*}}{\partial u_k^0}\Big|_{\varphi_{1}=0}<0$. We verify in the supplementary file \emph{Gumbel\_N.nb} that $n_{CS,u}$ and $d_{CS,u}$ (see \eqref{eqn:nCS_u_dCS_u}) are negative and positive, respectively, if the condition below, (ii-a), holds:
\begin{itemize}
    \item[(ii-a)] $\phi_{kk}>0$, $N\geq 2$ and  $f\left(N\right)\phi_{kk}<\beta_{k}<f_{CS,u}\left(N,\phi_{kk}\right)$, where $f_{CS,u}\left(N,\phi_{kk}\right)$ is the unique real root of the third degree polynomial $n_{CS,u,3}$ (viewed as a polynomial in $\beta_k$). 
\end{itemize}
    
    We next verify that $n_{CS,u,3}$ indeed has a unique real root and that $f_{CS,u}\left(N,\phi_{kk}\right)$ is linear in $\phi_{kk}$ and can thus be expressed as $f_{CS,u}\left(N,\phi_{kk}\right)=f_{CS,u}\left(N\right)\phi_{kk}$. We note that $n_{CS,u,3}/(6N^2)$ has the standard form
\begin{equation}\label{eqn:normalized_n4p_CS}
    \frac{n_{CS,u,3}}{6N^2} = \beta_k^3+b_{2}\beta_k^2+b_{1}\beta_k+b_{0},
\end{equation}
where $b_0$, $b_1$ and $b_2$ depend on $N$ and $\phi_{kk}$. We first clarify why this polynomial has a unique real root using Cardano's condition. We define $ t_k := \frac{1}{3}(3b_1-b_2^2)$, $s_k :=\frac{1}{27}(2b_2^3-9b_2b_1+27b_0)$ and $\Delta_k := (s_k/2)^2+(t_k/3)^3$. Equivalently, 
\begin{equation}\label{eqn:cardano_nCSu3}
    \begin{split}
        t_k & =-\frac{\left(144N^{4}-264N^{3}+103N^{2}-10N+1\right)\phi_{kk}^{2}}{108N^{4}}, \\
        s_k &= -\frac{\left(1728N^{6}-4752N^{5}+4356N^{4}-1106N^{3}+192N^{2}-15N+1\right)\phi_{kk}^{3}}{2916N^{6}}, \textnormal{ and} \\
        \Delta_k &:= \frac{\left(4608N^{6}-11136N^{5}+6944N^{4}+408N^{3}-97N^{2}+18N-1\right)\phi_{kk}^{6}}{139968N^{8}}. 
    \end{split}
\end{equation}
For each $k\in\{b,s\}$, if $\phi_{kk}> 0$ and $N\geq 2$, then $\Delta_k>0$. It thus follows that $n_{CS,u,3}$ has a unique real root given by 
\begin{equation}\label{eqn:cardano_nCSu3_roots}
    \begin{split}
        \alpha: = \textnormal{Car}(s_k,\Delta_k) - \frac{b_2}{3},
    \end{split}
\end{equation}
where $b_2 = \left(-12N^{2}+5N-1\right)\phi_{kk}/(6N^2)$ and $\textnormal{Car}(\cdot,\cdot)$ is given
\begin{equation}\label{def:CardanoRoot_nCSu3}
    \textnormal{Car}(s_k,\Delta_k):=\left(-\frac{s_k}{2}+\sqrt{\Delta_k}\right)^{1/3} + \left(-\frac{s_k}{2}-\sqrt{\Delta_k}\right)^{1/3}.
\end{equation}
For $\phi_{kk}>0$, it is not difficult to see that $\alpha$ is linear in $\phi_{kk}$, so we can write  
\begin{equation}\label{def:fCSu}
    \begin{split}
 f_{CS,u}\left(N\right)&=\alpha/\phi_{kk},
    \end{split}
\end{equation}
where $\alpha$ is given by \eqref{eqn:cardano_nCSu3_roots}. 

Finally, we show that there exists $\epsilon>0$ such that for any $(\phi_{bs},\phi_{sb})\in B_\epsilon(0)$, (i) and (ii) in Proposition \ref{prop:dCSkuk0} hold. From \eqref{eqn:proof_xki_omega}, \eqref{eqn:relationship_xp_z} and \eqref{FOCs_z2}, $\partial CS_k^*/\partial u_k^0$ is a rational function w.r.t.~$(\phi_{bs},\phi_{sb})$. Moreover, at  $(\phi_{bs},\phi_{sb})=(0,0)$, the partial derivative $\partial CS_k^*/\partial u_k^0$ is given by \eqref{eqn:dCS_duk0}. Thus, $\partial CS_k^*/\partial u_k^0$ is continuous w.r.t.~$(\phi_{bs},\phi_{sb})$  at  $(0,0)$. We conclude that there exists $\epsilon>0$ such that for any $(\phi_{bs},\phi_{sb})\in B_\epsilon(0)\subset \sR^2$, cases (i) and (ii) in Proposition \ref{prop:dCSkuk0} hold true. 
\end{proof}

\begin{proof}[{\bf Proof of Proposition \ref{prop:sign_zkm}}] Suppose that $N\geq 2$ and for each $k\in\{b,s\}$, $(\phi_{kk},\beta_k)$ satisfies (\ref{eqn:condition_unique_mono}). Assume that $\varphi_1=(\phi_{bs},\phi_{sb})= 0$. From the proof of Proposition \ref{coro:mono_uniqueness}, $M_k^\text{C}$ does not depend on $z_j$ for $j\neq k$ (see (\ref{Mkm_phi10})). Moreover,  
\begin{equation}\label{eqn:Mk_ninfty_collusion}
  \lim_{z_{k}\to-\infty}M_{k}^\text{C}\left(z_{b},z_{s},0,\varphi_{2};\psi\right)=\infty,  
\end{equation}
and $M_{k}^\text{C}\left(z_{b},z_{s},0,\varphi_{2};\psi\right)$ is strictly decreasing in $z_k$ for all $z_k\in\mathbb{R}$. From \eqref{Mkm_phi10}, we compute 
\begin{equation}\label{Mk_phi100_collusion}
M_{k}^\text{C}\left(z_{b},z_{s},0,\varphi_{2};\psi\right)\Big|_{z_k =0}=\frac{2\phi_{kk}}{N+1}-\beta_{k}\left(1+N\right)-u_{k}^{0}.
\end{equation}
The only root of (\ref{Mk_phi100_collusion}) as a polynomial in $\beta_k$ is given by 
\begin{equation}\label{def:gammaC_proof}
    \gamma^\text{C}(N,\phi_{kk},u_k^0):=\frac{2\phi_{kk}-u_{k}^{0}\left(N+1\right)}{\left(N+1\right)^{2}}.
\end{equation}

Observation (i): The condition $\beta_k>\gamma^\text{C}(N,\phi_{kk},u_k^0)$ implies that (\ref{Mk_phi100_collusion}) is strictly negative. We combine this fact with (\ref{eqn:Mk_ninfty_collusion}) and the facts that $M_{k}^\text{C}\left(z_{b},z_{s},0,\varphi_{2};\psi\right)$ is strictly decreasing in $z_k$  $\forall z_k\in\sR$ and $z_k^\text{C}$ is the unique solution of $M_{k}^\text{C}\left(z_{b},z_{s},0,\varphi_{2};\psi\right)=0$ to conclude that $z_k^\text{C}<0$.

Observation (ii): If $\beta_k<\gamma^\text{C}(N,\phi_{kk},u_k^0)$, then (\ref{Mk_phi100_collusion}) is strictly positive and thus $z_k^\text{C}>0$.

The Proof of Proposition \ref{coro:mono_uniqueness} implies that there is $\epsilon>0$ and a unique continuous function $$\left(z_{b}^\text{C}\left(\cdot\right),z_{s}^\text{C}\left(\cdot\right)\right):B_{\epsilon}\left(0,0\right)\longrightarrow\mathbb{R}^{2}$$ such that for all $\varphi_{1}\in B_{\epsilon}\left(0,0\right)$, $
\mM^\text{C}\left(z_{b}^\text{C}\left(\varphi_{1}\right),z_{s}^\text{C}\left(\varphi_{1}\right),\varphi_{1},\varphi_{2};\psi\right)=0.$  By this continuity, for $\hat{\varepsilon}=|z_k^\text{C}(0,0)|/2$, there exists $\delta>0$ such that for all $\varphi_{1}\in B_{\min(\delta, \varepsilon)}\left(0,0\right)$,
\begin{equation*}
    z_{k}^\text{C}(0,0)-\frac{\left|z_{k}^\text{C}(0,0)\right|}{2}<z_{k}^\text{C}\left(\varphi_{1}\right)<z_{k}^\text{C}(0,0)+\frac{\left|z_{k}^\text{C}(0,0)\right|}{2}.
\end{equation*}
Thus, when $z_{k}^\text{C}(0,0)<0$, by observation (i) we obtain $z_{k}^\text{C}\left(\varphi_{1}\right)<\frac{z_{k}^\text{C}(0,0)}{2}<0$. When $z_{k}^\text{C}>0$, by observation (ii) we obtain $z_{k}^\text{C}\left(\varphi_{1}\right)>\frac{z_{k}^\text{C}(0,0)}{2}>0$. Concluding the proof of the Proposition. 

\end{proof}

\begin{proof}[{\bf Proof of Corollary \ref{coro:gammavsgammaC}}]
Note that from \eqref{eqn:gamma} and \eqref{def:gammaC_proof}, 
\begin{equation}\label{eqn:gammaCvsgamma}
\begin{array}{l}
\gamma\left(N,\phi_{kk},u_{k}^{0}\right)-\gamma^\text{C}\left(N,\phi_{kk},u_{k}^{0}\right)\\
=\frac{2\left(N-1\right)\phi_{kk}+\left(-N^{2}+N+2\right)u_{k}^{0}+\sqrt{N^{2}\left(N+1\right)^{2}\left(u_{k}^{0}\right)^{2}+4\left(N^{2}-1\right)\phi_{kk}^{2}-4\left(N-1\right)\left(N+1\right)^{2}u_{k}^{0}\phi_{kk}}}{2\left(N+1\right)^{2}}
\end{array}.
\end{equation}

From (\ref{sign_ofgamma}) and the Proof of Corollary \ref{coro:sign_zk}, if $\gamma\left(N,\phi_{kk},u_{k}^{0}\right)\geq0$ then, 
\begin{equation*}
    \textnormal{either } (\phi_{kk}<0 \textnormal{ and } u_{k}^{0}\leq\frac{2\phi_{kk}}{N+1}) \textnormal{ or } \phi_{kk}\geq0.
\end{equation*}

Suppose that $(\phi_{kk}<0 \textnormal{ and } u_{k}^{0}\leq\frac{2\phi_{kk}}{N+1})$ . By \eqref{eqn:gammaCvsgamma}, $\gamma\left(N,\phi_{kk},u_{k}^{0}\right)\geq \gamma^\text{C}\left(N,\phi_{kk},u_{k}^{0}\right)$. On the other hand,  if $\phi_{kk}\geq0$ then (\ref{eqn:gammaCvsgamma}) implies that $\gamma\left(N,\phi_{kk},u_{k}^{0}\right)\geq\gamma^\text{C}\left(N,\phi_{kk},u_{k}^{0}\right)$. 

\end{proof}

\begin{proof}[{\bf Proof of Corollary \ref{coro:sign_zkm}}]
First assume $\phi_{kk}\leq 0$ and $\beta_k>0$.  From the definition of $\gamma^\text{C}(N,\phi_{kk},u_k^0)$ in \eqref{def:gammaC_proof}, if $u_{k}^{0}\geq2\phi_{kk}/(N+1)$, then $\gamma^\text{C}(N,\phi_{kk},u_k^0)\leq 0 $. In this case, statement (ii) of Proposition \ref{prop:sign_zkm} is not feasible as $\beta_k>0$.

Now assume that $\phi_{kk}>0$ and $\beta_{k}>\frac{8\phi_{kk}}{27N}$. The unique $\tilde{u}_{k}^\text{C}$ such that  $\gamma^\text{C}(N,\phi_{kk},\tilde{u}_{k}^\text{C})= \frac{8\phi_{kk}}{27N}$ is given by 
\begin{equation*}
    \tilde{u}_{k}^{0}=-\frac{2\left(N\left(4N-19\right)+4\right)\phi_{kk}}{27N\left(N+1\right)}.
\end{equation*}
If $u_{k}^{0}\geq\tilde{u}_{k}^\text{C}$, then $ \gamma^\text{C}(N,\phi_{kk},\tilde{u}_{k}^0)\leq \frac{8\phi_{kk}}{27N}$. Thus, statement (ii) in Proposition \ref{prop:sign_zkm}  is not feasible as $\beta_k>\frac{8\phi_{kk}}{27N}$.

We finish the proof with a piecewise definition for $\tilde{u}_k^\text{C}$,
\begin{equation}
    \label{def:tilde_uk0m}
    \tilde{u}_{k}^\text{C}:= \begin{cases}
        \frac{2\phi_{kk}}{N+1}& \textnormal{if } \phi_{kk}\leq 0,\\
        -\frac{2\left(N\left(4N-19\right)+4\right)\phi_{kk}}{27N\left(N+1\right)} & \textnormal{if } \phi_{kk}>0.
    \end{cases}
\end{equation}
\end{proof}

\begin{proof}[{\bf Proof of Proposition \ref{prop:CNEvsCE}}] First, we set $\varphi_1=(\phi_{bs},\phi_{sb})=0$. From (\ref{Mk_phi10}) and (\ref{Mkm_phi10}), we obtain 
\begin{equation}\label{eqn:Mk_Mkm}
M_{k}\left(z_{b},z_{s},0,\varphi_{2};\psi\right)-M_{k}^\text{C}\left(z_{b},z_{s},0,\varphi_{2};\psi\right)=\frac{\beta_{k}\left(N-1\right)e^{z_{k}}\left(\beta_{k}\left(Ne^{z_{k}}+1\right)^{2}-e^{z_{k}}\phi_{kk}\right)}{\beta_{k}\left(\left(N-1\right)e^{z_{k}}+1\right)\left(Ne^{z_{k}}+1\right)-e^{z_{k}}\phi_{kk}}
\end{equation}

Step (i): If $\phi_{kk}\leq0$ or $(\phi_{kk}>0$ and $\beta_{k}>f\left(N\right)\phi_{kk})$, then \eqref{eqn:Mk_Mkm} is strictly positive for all $z_{k}\in\mathbb{R}$. The proofs of Propositions  \ref{prop:existence_gumbel} and \ref{coro:mono_uniqueness} imply that the functions $M_{k}\left(z_{b},z_{s},0,\varphi_{2};\psi\right)$ and $M_{k}^\text{C}\left(z_{b},z_{s},0,\varphi_{2};\psi\right)$  are independent of $z_l$ for $l\neq k$ and are strictly decreasing on $z_k$. It follows that $z_{k}^{*}>z_{k}^\text{C}$. 

Again, from the Proofs of Proposition \ref{prop:existence_gumbel} and \ref{coro:mono_uniqueness}, there is $\epsilon>0$ and unique continuous functions $\left(z_{b}^{*}\left(\cdot\right),z_{s}^{*}\left(\cdot\right)\right):B_{\epsilon}\left(0,0\right)\longrightarrow\mathbb{R}^{2}$ and $\left(z_{b}^\text{C}\left(\cdot\right),z_{s}^\text{C}\left(\cdot\right)\right):B_{\epsilon}\left(0,0\right)\longrightarrow\mathbb{R}^{2}$ such that for all $\varphi_{1}\in B_{\epsilon}\left(0,0\right)$, 
\begin{equation}
M\left(z_{b}^{*}\left(\varphi_{1}\right),z_{s}^{*}\left(\varphi_{1}\right),\varphi_{1},\varphi_{2};\psi\right)=0=M^\text{C}\left(z_{b}^\text{C}\left(\varphi_{1}\right),z_{s}^\text{C}\left(\varphi_{1}\right),\varphi_{1},\varphi_{2};\psi\right).    
\end{equation}
By continuity, there is $\hat{\epsilon}>0$ such that $z_{k}^{*}\left(\varphi_{1}\right)>z_{k}^\text{C}\left(\varphi_{1}\right)$ for all $\varphi_{1}\in B_{\hat{\epsilon}}\left(0,0\right)$.

Step (ii): By \eqref{eqn_pz} and  \eqref{eqn_pm},  $x_{k}^*=\omega\left(z_{k}^*\right)$ and $x_{k}^\text{C}=\omega\left(z_{k}^\text{C}\right)$ where $\omega\left(z\right)=\frac{1}{e^{-z}+N}$. Note that $\omega'\left(z\right)>0$ for all $z\in\mathbb{R}$. Then, $x_{k}^{*}=\omega\left(z_{k}^{*}\right)>\omega\left(z_{k}^\text{C}\right)=x_{k}^\text{C}$ for all $\varphi_{1}\in B_{\hat{\epsilon}}\left(0,0\right)$. 

Step (iii): From \eqref{def:zk}, the function $\boldsymbol{p}\left(\boldsymbol{z}\right)=\Phi\Omega\left(\boldsymbol{z}\right)-\boldsymbol{\beta}\boldsymbol{z}-\boldsymbol{u}_{0}$ determines the equilibrium prices $p_k^*$ and $p_k^\text{C}$. Note that 
\begin{equation}\label{eqn:pk_derivative}
\frac{\partial p_{k}}{\partial z_{k}}=\phi_{kk}\omega'\left(z_{k}\right)-\beta_{k}=\frac{\phi_{kk}e^{z_{k}}-\beta_{k}\left(Ne^{z_{k}}+1\right)^{2}}{\left(Ne^{z_{k}}+1\right)^{2}}.    
\end{equation}
 The function $z_k\mapsto \frac{e^{z_k}}{\left(Ne^{z_k}+1\right)^{2}}$
has a unique maximum over $\mathbb{R}$ at $z_k^0 = \log\left(\frac{1}{N}\right)$ and such maximum is given by $\frac{1}{4N}$. Then, (\ref{eqn:pk_derivative}) is strictly negative for all $z_k\in\mathbb{R}$ when either $ \phi_{kk}\leq 0$ or $( \phi_{kk}>0$ and $\beta_k>\frac{\phi_{kk}}{4N})$. Moreover, from (\ref{def:RNphi_kk_proof}), $f\left(N\right)\phi_{kk}>\frac{\phi_{kk}}{4N}$ for all $N\geq 2$ and $ \phi_{kk}>0$. Then, for $\varphi_{1}=0$, $p_{k}^{*}=p_{k}\left(z_{k}^{*}\right)<p_k\left(z_{k}^\text{C}\right)=p_{k}^\text{C}$. By step (i) and continuity, there is $\hat{\epsilon}>0$ such that $p_{k}^{*}\left(z_k^*(\varphi_{1})\right)<p_{k}^\text{C}\left(z_k^C(\varphi_{1})\right)$ for all $\varphi_{1}\in B_{\hat{\epsilon}}\left(0,0\right)$.
    
\end{proof}

\begin{proof}[{\bf Proof of Proposition \ref{prop:dpk_dN}}]
Recall from \eqref{eqn:relationship_xp_z} in the Proof of Proposition \ref{prop:FOC_gumbel}, $$\boldsymbol{p}^*=\Phi\Omega\left(\boldsymbol{z}^*\right)-\boldsymbol{\beta}\boldsymbol{z}^*-\boldsymbol{u}_{0},$$ where $\Omega(\vz^*) = (\omega(z_b^*), \omega(z_s^*))^T$, $\omega(\cdot)$ and $\vz^*$ are given by \eqref{eqn:proof_xki_omega} and \eqref{FOCs_z2}, respectively. We want to compute the following quantity when $\varphi_1=0$, 
\begin{equation}\label{eqn:dp_dN}
    \begin{split}
       \frac{\partial p_{k}^{*}}{\partial N}=\left(\phi_{kk}\omega^{\prime}\left(z_{k}^{*}\right)-\beta_{k}\right)\frac{\partial z_{k}^{*}}{\partial N}+\phi_{kk}\frac{\partial\omega\left(z_{k}^{*}\right)}{\partial N}+\phi_{kj}\left(\omega^{\prime}(z_{j}^{*})\frac{\partial z_{j}^{*}}{\partial N}+\frac{\partial\omega(z_{j}^{*})}{\partial N}\right)
    \end{split}, \  k,j\in\{b,s\}\ k\neq j.
\end{equation}
By \eqref{eqn:proof_xki_omega}, $\omega(z_k^*)= 1/(e^{-z_k^*}+N)$, then
\begin{equation}
    \begin{split}
\omega^{\prime}\left(z_{k}^{*}\right) &=\frac{e^{-z_{k}^{*}}}{\left(e^{-z_{k}^{*}}+N\right)^{2}}\ \textnormal{ and }\\
        \frac{\partial\omega\left(z_{k}^{*}\right)}{\partial N}&=\frac{-1}{\left(e^{-z_{k}^{*}}+N\right)^{2}}.
    \end{split}
\end{equation}

After differentiating both sides of the two equations in (\ref{FOCs_z2}) w.r.t.~$N$, we can write $\partial z_k^*/\partial N$ at $\varphi_1=(\phi_{bs},\phi_{sb})=0$,
\begin{equation}\label{eqn:dzdN_derivation}
    \left[\begin{array}{c}
\frac{\partial z_{s}^*}{\partial N}\\
\frac{\partial z_{b}^*}{\partial N}
\end{array}\right]_{\varphi_1=0}=\frac{1}{J_{N}}\left[\begin{array}{cc}
-d & b\\
c & -a
\end{array}\right]\left[\begin{array}{c}
\frac{\partial\left[h_{b}-\phi_{ss}d_{b}L_{b}\right]}{\partial N}\omega\left(z_{b}^*\right)+\left(2\phi_{bb}-\phi_{ss}d_{b}L_{b}+h_{b}\right)\frac{\partial \omega\left(z_{b}^*\right)}{\partial N}\\
\frac{\partial\left[h_{s}-\phi_{bb}d_{s}L_{s}\right]}{\partial N}\omega\left(z_{s}^*\right)+\left(2\phi_{ss}-\phi_{bb}d_{s}L_{s}+h_{s}\right)\frac{\partial \omega\left(z_{s}^*\right)}{\partial N}
\end{array}\right],
\end{equation}
where $J_N:= ad - bc$ and 
\begin{equation}\label{eqn:abcd}
    \begin{split}
a&=\omega\left(z_{b}^*\right)\left(\frac{\partial h_{b}}{\partial z_{s}}-\phi_{ss}\frac{\partial\left[d_{b}L_{b}\right]}{\partial z_{s}}\right),\\
b&=\omega\left(z_{b}^*\right)\left(\frac{\partial h_{b}}{\partial z_{b}}-\phi_{ss}\frac{\partial\left[d_{b}L_{b}\right]}{\partial z_{b}}\right)+\left(2\phi_{bb}-\phi_{ss}d_{b}L_{b}+h_{b}\right)\omega'\left(z_{b}^*\right)-\beta_b,\\
c&=\omega\left(z_{s}^*\right)\frac{\partial\left[h_{s}-\phi_{bb}d_{s}L_{s}\right]}{\partial z_{s}}+\left(2\phi_{ss}-\phi_{bb}d_{s}L_{s}+h_{s}\right)\omega'\left(z_{s}^*\right)-\beta_s, \textnormal{ and}\\
d&=\omega\left(z_{s}^*\right)\frac{\partial\left[h_{s}-\phi_{bb}d_{s}L_{s}\right]}{\partial z_{b}}.
    \end{split}
\end{equation}
Moreover, recall that $L_k$, $d_k$ and $h_k$ for $k\in\{b,s\}$ are given by (\ref{eqn_z_Lk}). After plugging (\ref{eqn:abcd}) into (\ref{eqn:dzdN_derivation}) and then plugging the resulting expression into (\ref{eqn:dp_dN}), we show (see \emph{Gumbel\_N.nb}) that  
\begin{equation}\begin{split}
\label{eqn:dpk_dN_proof}
    \frac{\partial p_{k}^{*}}{\partial N}\Big|_{\varphi_{1}=0}&=\frac{n_p\left(z_{k}^{*},\beta_{k},\phi_{kk},N\right)}{d\left(z_{k}^{*},\beta_{k},\phi_{kk},N\right)},
\end{split}
\end{equation}
where $n_p\left(z_{k}^{*},\beta_{k},\phi_{kk},N\right)$ and $d\left(z_{k}^{*},\beta_{k},\phi_{kk},N\right)$ can be written as polynomials in $e^{z_{k}^{*}}$ in the following way, 
\begin{equation}\label{eqn:np_d}
\begin{split}
n_p\left(z_{k}^{*},\beta_{k},\phi_{kk},N\right)&:= \sum_{m=2}^{6}n_{m,p}e^{mz_{k}^{*}},\quad  \textnormal{ and }\\
d\left(z_{k}^{*},\beta_{k},\phi_{kk},N\right)&:= \sum_{m=0}^{6}d_{m}e^{mz_{k}^{*}}.
\end{split}
\end{equation}
Moreover, the coefficients $n_{m,p}$ are as follows:
\begin{equation}\label{nmp_coefficients}
    \begin{array}{c}
n_{2,p}=\beta_{k}^{3}\left(\phi_{kk}-\beta_{k}\right),\\
n_{3,p}=-\beta_{k}^{2}\left(4\beta_{k}^{2}N+\beta_{k}\phi_{kk}\left(1-4N\right)+2\phi_{kk}^{2}\right),\\
n_{4,p}=\beta_{k}\left(-6\beta_{k}^{3}N^{2}+2\beta_{k}^{2}\phi_{kk}N\left(3N-1\right)+\beta_{k}\phi_{kk}^{2}\left(3-4N\right)+\phi_{kk}^{3}\right),\\
n_{5,p}=-\beta_{k}\left(4\beta_{k}^{3}N^{3}+\beta_{k}^{2}\phi_{kk}N^{2}\left(1-4N\right)+\beta_{k}\phi_{kk}^{2}N\left(2N-3\right)+\phi_{kk}^{3}\right), \textnormal{ and } \\
n_{6,p}=\beta_{k}^{3}N^{4}\left(\phi_{kk}-\beta_{k}\right).
\end{array}
\end{equation}
The coefficients $d_m$ are given by \[\Scale[0.8]{
d_{0}=\beta_{k}^{3},}\]
\[\Scale[0.8]{d_{1}=\beta_{k}^{2}\left(\beta_{k}\left(6N-1\right)-4\phi_{kk}\right),}\]
\[\Scale[0.8]{d_{2}=\beta_{k}\left(\beta_{k}^{2}\left(15N^{2}-6N+1\right)+4\beta_{k}\left(1-4N\right)\phi_{kk}+5\phi_{kk}^{2}\right),}\]
\[\Scale[0.8]{d_{3}=\beta_{k}^{3}N\left(20N^{2}-14N+4\right)+\beta_{k}^{2}\phi_{kk}\left(-24N^{2}+11N-1\right)+\beta_{k}\phi_{kk}^{2}\left(10N-3\right)-2\phi_{kk}^{3},}\]
\[\Scale[0.8]{d_{4}=\beta_{k}N\left(\beta_{k}^{2}N\left(15N^{2}-16N+6\right)+\beta_{k}\phi_{kk}\left(-16N^{2}+10N-2\right)+\left(5N-2\right)\phi_{kk}^{2}\right),}\]
\[\Scale[0.8]{d_{5}=\beta_{k}N^{2}\left(\beta_{k}^{2}N\left(6N^{2}-9N+4\right)+\beta_{k}\phi_{kk}\left(-4N^{2}+3N-1\right)+\phi_{kk}^{2}\right), \textnormal{ and }}\]
\begin{equation}\label{dm_coefficients}
    \small{d_{6}=\beta_{k}^{3}\left(N-1\right)^{2}N^{4}}.
\end{equation}

Because the expressions determining $n_p\left(z_{k}^{*},\beta_{k},\phi_{kk},N\right)$ and $d\left(z_{k}^{*},\beta_{k},\phi_{kk},N\right)$ are complex, we focus on finding sufficient conditions for these expressions to have a specific sign for all $z_{k}^{*}$. Note that not all the coefficients in $\left\{ d_{m}\right\} _{m=0}^{6}$ can be simultaneously negative because $d_{0}>0$. Thus, we focus on finding the regions for which $d_{m}>0$ for each $m=0,\dots,6$.

Case (i): $\frac{\partial p_{k}^{*}}{\partial N}\Big|_{\varphi_{1}=0}<0$. 
We verify in the supplementary file \emph{Gumbel\_N.nb}
that $n_p$ and $d$ (see \eqref{eqn:np_d}) are negative and positive, respectively, if either of the two conditions below, (i-a) or (i-b), hold: 
\begin{itemize}
    \item[(i-a)] $\phi_{kk}>0$, $N\geq 2$ and $\beta_{k}>\phi_{kk}$.
    \item[(i-b)] $\phi_{kk}\leq 0$, $N\geq 2$ and $\beta_{k}>g_p(N, \phi_{kk})$, where $g_{p}(N,\phi_{kk})$ is the largest real root of the third degree polynomial $n_{5,p}/\beta_k$ (viewed as a polynomial in $\beta_k$). 
\end{itemize}

    We next verify that $n_{5,p}/\beta_k$ has three real roots and that $g_p(N, \phi_{kk})$ is linear in $\phi_{kk}$ and can thus be expressed as  $g_p(N, \phi_{kk})= g_{p}\left(N\right)\phi_{kk}$.
    We note that $n_{5,p}/(\beta_k(4N^3))$ has the standard form  
\begin{equation}\label{eqn:normalized_n5p}
    \frac{n_{5,p}}{\beta_k(4N^3)} = \beta_k^3+c_{2}\beta_k^2+c_{1}\beta_k+c_{0},
\end{equation}
where $c_0$, $c_1$ and $c_2$ depend on $N$ and $\phi_{kk}$.
Following the proof of Proposition \ref{prop:existence_gumbel}, we use Cardano's formula to define
\begin{equation}\label{eqn:cardano_n5p}
    \begin{split}
        t_k &:= \frac{1}{3}(3 c_{1} - c_{2}^2)= -\frac{\left(16 N^2-32 N+37\right) \phi_{kk}^2}{48 N^2}, \\
        s_k &:= \frac{1}{27}(2 c_{2}^3 - 9 c_{2}c_{1} + 27 c_{0})= -\frac{\left(64 N^3-192 N^2+264 N-271\right) \phi_{kk}^3}{864 N^3}, \textnormal{ and} \\
        \Delta_k &:= (s_k/2)^2 + (t_k/3)^3 = -\frac{\left(64 N^4-96 N^3+52 N^2+108 N-211\right) \phi_{kk}^6}{27648 N^6}. 
    \end{split}
\end{equation}
For each $k\in\{b,s\}$, if $\phi_{kk}\neq 0$ and $N\geq 2$, then  $\Delta_k<0$. It thus follows that $n_{5,p}/\beta_k$ has three simple real roots given by 
\begin{equation}\label{eqn:cardano_n5p_roots}
    \begin{split}
        \alpha_{j}: = 2\sqrt{-\frac{t_k}{3}}\cos\left(\frac{\theta+2j\pi}{3}\right) - \frac{c_{2}}{3}, \ \textnormal{ for } j=0,1,2,\\
        0<\theta<\pi, \textnormal{ and } \cos(\theta) = \frac{-s_k/2}{\sqrt{-(t_k/3)^3}}.
    \end{split}
\end{equation}
From \eqref{eqn:cardano_n5p} and \eqref{eqn:cardano_n5p_roots}, $\theta = \theta(N)$ is a function of $N$ and is independent of $\phi_{kk}$. Moreover, for $\phi_{kk}<0 $, $\sqrt{-t_k/3}$ and $c_{2} = -\frac{(4 N-1) \phi_{kk}}{4 N}$ are linear functions of $\phi_{kk}$. Thus, for  $\phi_{kk}<0 $, the three simple real roots of $n_{5,p}/\beta_k$ can be written as 
\begin{equation}\label{eqn:cardano_n5p_roots2}
    \begin{split}
        \alpha_{j} = w_{j}(N) \phi_{kk}.
    \end{split}
\end{equation}
We can thus write $g_p(N, \phi_{kk})= g_{p}\left(N\right)\phi_{kk}$, where 
\begin{equation}\label{def:rp1}
    g_p(N) = \max_{j=0,1,2} \{w_j(N)\}.
\end{equation}

From \eqref{nmp_coefficients} and \eqref{eqn:normalized_n5p}, 
$$\lim_{N\to\infty} \frac{n_{5,p}}{\beta_k(4N^3)} = \beta_k^3-\phi_{kk}\beta_k^2.$$
Thus, for $\phi_{kk}<0$, the quantity $g_{p}\left(N\right)\phi_{kk}$ approaches $0$ as $N\to\infty$.

Case (ii): $\frac{\partial p_{k}^{*}}{\partial N}\Big|_{\varphi_{1}=0}>0$. We verify in the supplementary file \emph{Gumbel\_N.nb} that $n_p$ and $d$ (see \eqref{eqn:np_d}) are positive if either of the two conditions below, (ii-a) or (ii-b), hold:
\begin{itemize}
    \item[(ii-a)] $\phi_{kk}>0$, $N\geq 3$ and $f(3)\phi_{kk}<\beta_{k}<\frac{2}{3}\phi_{kk}$, where $f(3)$ is given by \eqref{def:RNphi_kk_proof}.
    \item[(ii-b)] $\phi_{kk}>0$, $N\geq 4$ and  $f\left(N\right)\phi_{kk}<\beta_{k}<f_{p}\left(N,\phi_{kk}\right)$, where $f_{p}\left(N,\phi_{kk}\right)$ is the unique real root of the third degree polynomial $n_{4,p}/\beta_k$ (viewed as a polynomial in $\beta_k$). 
\end{itemize}
    
    We next verify that $n_{4,p}/\beta_k$ indeed has a unique real root and that $f_{p}\left(N,\phi_{kk}\right)$ is linear in $\phi_{kk}$ and can thus be expressed as $f_{p}\left(N,\phi_{kk}\right)=f_{p}\left(N\right)\phi_{kk}$. We note that $n_{4,p}/(\beta_k(-6N^2))$ has the standard form
\begin{equation}\label{eqn:normalized_n4p}
    \frac{n_{4,p}}{\beta_k(-6N^2)} = \beta_k^3+b_{2}\beta_k^2+b_{1}\beta_k+b_{0},
\end{equation}
where $b_0$, $b_1$ and $b_2$ depend on $N$ and $\phi_{kk}$. We first clarify why this polynomial has a unique real root using Cardano's condition. We define
\begin{equation}\label{eqn:cardano_n4p}
    \begin{split}
        \tilde{t}_k &:= -\frac{\left(18 N^2-48 N+29\right) \phi_{kk}^2}{54 N^2}, \\
        \tilde{s}_k &:=  -\frac{\left(108 N^3-432 N^2+630 N-85\right) \phi_{kk}^3}{1458 N^3}, \textnormal{ and} \\
        \tilde{\Delta}_k &:= \frac{\left(72 N^4-168 N^3-88 N^2+556 N-171\right) \phi_{kk}^6}{34992 N^6}. 
    \end{split}
\end{equation}
For each $k\in\{b,s\}$, if $\phi_{kk}> 0$ and $N\geq 4$, then $\tilde{\Delta}_k>0$. It thus follows that $n_{4,p}/\beta_k$ has a unique real root given by 
\begin{equation}\label{eqn:cardano_n4p_roots}
    \begin{split}
        \tilde{\alpha}_{j}: = \textnormal{Car}(\tilde{s}_k,\tilde{\Delta}_k) - \frac{b_2}{3},
    \end{split}
\end{equation}
where $b_2 = (1- 3N)\phi_{kk}/(3N)$ and $\textnormal{Car}(\cdot,\cdot)$ is given
\begin{equation}\label{def:CardanoRoot}
    \textnormal{Car}(s_k,\Delta_k):=\left(-\frac{s_k}{2}+\sqrt{\Delta_k}\right)^{1/3} + \left(-\frac{s_k}{2}-\sqrt{\Delta_k}\right)^{1/3}.
\end{equation}
For $\phi_{kk}>0$, it is not difficult to see that $\tilde{\alpha}_j$ is linear in $\phi_{kk}$, so we can write  
\begin{equation}\label{eqn:cardano_n4p_roots2}
    \begin{split}
 f_{p}\left(N\right)&=\tilde{\alpha}_j/\phi_{kk},
    \end{split}
\end{equation}
where $\tilde{\alpha}_j$ is given by \eqref{eqn:cardano_n4p_roots}. Finally, for any $\phi_{kk}>0$, by \eqref{eqn:normalized_n4p}
\begin{equation*}
   \lim_{N\to\infty} \frac{n_{4,p}}{\beta_k(-6N^2)} = \beta_k^2(\beta_k-\phi_{kk}).
\end{equation*}
It thus follows that $f_{p}\left(N\right)\phi_{kk} $ converges to $\phi_{kk}$ as $N\to \infty$.

From (ii-a) and (ii-b), it follows that $\frac{\partial p_{k}^{*}}{\partial N}|_{\varphi_{1}=0}>0$ whenever $N\geq3$, $\phi_{kk}>0$ and $f\left(N\right)\phi_{kk}<\beta_{k}<f_{p}\left(N\right)\phi_{kk}$, where the definition of $f_p(N)$ is extended to include the case $N=3$ as follows,
\begin{equation}
    \label{def:rp2} f_{p}\left(N\right) := \begin{cases}
\frac{2}{3} & N=3,\\
\tilde{\alpha}_j/\phi_{kk} & N\geq4,
\end{cases}
\end{equation}
where $\tilde{\alpha}_j$ is given by \eqref{eqn:cardano_n4p_roots}.

Finally, we show that there exists $\epsilon>0$ such that for any $(\phi_{bs},\phi_{sb})\in B_\epsilon(0)$, (i) and (ii) in Proposition \ref{prop:dpk_dN} hold. From \eqref{eqn:dp_dN}, \eqref{eqn:dzdN_derivation} and \eqref{eqn:abcd}, $\partial p_k^*/\partial N$ is a rational function w.r.t.~$(\phi_{bs},\phi_{sb})$. Moreover, at  $(\phi_{bs},\phi_{sb})=(0,0)$, the partial derivative $\partial p_k^*/\partial N$ is given by \eqref{eqn:dpk_dN_proof}. Thus, $\partial p_k^*/\partial N$ is continuous w.r.t.~$(\phi_{bs},\phi_{sb})$  at  $(0,0)$. We conclude that there exists $\epsilon>0$ such that for any $(\phi_{bs},\phi_{sb})\in B_\epsilon(0)\subset \sR^2$, cases (i) and (ii) in Proposition \ref{prop:dpk_dN} hold true. 
\end{proof}

\begin{proof}[{\bf Proof of Proposition \ref{prop:dNxk_dN}}] Using the same ideas from the Proof of Proposition \ref{prop:dpk_dN}, we want to compute the following quantity when $\varphi_1 = 0$,
\begin{equation}
    \label{eqn:dNxdN_proof}
    \frac{\partial\left(Nx_{k}^{*}\right)}{\partial N}=x_{k}^{*}+N\frac{\partial x_{k}^{*}}{\partial N}.
\end{equation}
By \eqref{eqn:proof_xki_omega}, $x_k^*= \frac{1}{e^{-z_k^*}+N}$. Then, 
\begin{equation}\label{eqn:dx_dN}
    \begin{split}
        \frac{\partial x_{k}^{*}}{\partial N}&=\frac{-1}{\left(e^{-z_{k}^{*}}+N\right)^{2}}\left(1-e^{-z_{k}^{*}}\frac{\partial z_{k}^{*}}{\partial N}\right).
    \end{split}
\end{equation}
Using (\ref{eqn:dzdN_derivation}) and (\ref{eqn:dx_dN}), we show (see \emph{Gumbel\_N.nb}) that 
\begin{equation}
\label{eqn:dNxdN_proof2}
    \frac{\partial \left(Nx_{k}^{*}\right)}{\partial N}\Big|_{\varphi_{1}=0}=\frac{n_{Nx}\left(z_{k}^{*},\beta_{k},\phi_{kk},N\right)}{d_{Nx}\left(z_{k}^{*},\beta_{k},\phi_{kk},N\right)},
\end{equation}
where 
\begin{equation}\label{eqn:n_Nx_d}
\begin{split}
n_{Nx}\left(z_{k}^{*},\beta_{k},\phi_{kk},N\right)&:= \sum_{m=1}^{6}n_{m,Nx}e^{mz_{k}^{*}}\quad  \textnormal{ and }\\ 
d_{Nx}\left(z_{k}^{*},\beta_{k},\phi_{kk},N\right)&:= \sum_{m=0}^{7}d_{m,Nx}e^{mz_{k}^{*}}.
\end{split}
\end{equation}
Moreover, the coefficients $n_{m,Nx}$ are as follows: \[\Scale[0.8]{
n_{1,Nx}=\beta_{k}^{3},}\]
\[\Scale[0.8]{n_{2,Nx}=\beta_{k}^{2}\left(\beta_{k}\left(5N-1\right)-4\phi_{kk}\right),}\]
\[\Scale[0.8]{n_{3,Nx}=\beta_{k}\left(\beta_{k}^{2}\left(10N^{2}-4N+1\right)+2\beta_{k}\left(2-7N\right)\phi_{kk}+5\phi_{kk}^{2}\right),}\]
\[\Scale[0.8]{n_{4,Nx}=\beta_{k}^{3}N\left(10N^{2}-6N+3\right)+\beta_{k}^{2}\phi_{kk}\left(-18N^{2}+10N-1\right)+3\beta_{k}\phi_{kk}^{2}\left(3N-1\right)-2\phi_{kk}^{3},}\]
\[\Scale[0.8]{n_{5,Nx}=\beta_{k}N\left(\beta_{k}^{2}N\left(5N^{2}-4N+3\right)+2\beta_{k}\phi_{kk}\left(-5N^{2}+4N-1\right)+\left(4N-3\right)\phi_{kk}^{2}\right),}\]
\begin{equation}\label{n_Nx_coefficients}
\small{n_{6,Nx}=\beta_{k}^{2}N^{2}\left(\beta_{k}N\left(N^{2}-N+1\right)-\left(2N^{2}-2N+1\right)\phi_{kk}\right)}.
\end{equation}
The coefficients $d_{m,Nx}$ are given by \[\Scale[0.8]{d_{0,Nx}=\beta_{k}^{3},}\]
\[\Scale[0.8]{d_{1,Nx}=\beta_{k}^{2}\left(\beta_{k}\left(7N-1\right)-4\phi_{kk}\right),}\]
\[\Scale[0.8]{d_{2,Nx}=\beta_{k}\left(\beta_{k}^{2}\left(21N^{2}-7N+1\right)+4\beta_{k}\left(1-5N\right)\phi_{kk}+5\phi_{kk}^{2}\right),}\]
\[\Scale[0.8]{d_{3,Nx}=5\beta_{k}^{3}N\left(7N^{2}-4N+1\right)+\beta_{k}^{2}\phi_{kk}\left(-40N^{2}+15N-1\right)+3\beta_{k}\left(5N-1\right)\phi_{kk}^{2}-2\phi_{kk}^{3},}\]
\[\Scale[0.8]{d_{4,Nx}=N\left(5\beta_{k}^{3}N\left(7N^{2}-6N+2\right)+\beta_{k}^{2}\phi_{kk}\left(-40N^{2}+21N-3\right)+5\beta_{k}\left(3N-1\right)\phi_{kk}^{2}-2\phi_{kk}^{3}\right),}\]
\[\Scale[0.8]{d_{5,Nx}=\beta_{k}N^{2}\left(\beta_{k}^{2}N\left(21N^{2}-25N+10\right)+\beta_{k}\left(-20N^{2}+13N-3\right)\phi_{kk}+\left(5N-1\right)\phi_{kk}^{2}\right),}\]
\[\Scale[0.8]{d_{6,Nx}=\beta_{k}N^{3}\left(\beta_{k}^{2}N\left(7N^{2}-11N+5\right)+\beta_{k}\left(-4N^{2}+3N-1\right)\phi_{kk}+\phi_{kk}^{2}\right),}\]
\begin{equation}\label{d_Nx_coefficients}
\small{d_{7,Nx}=\beta_{k}^{3}\left(N-1\right)^{2}N^{5}}.
\end{equation}

Because the expressions determining $n_{Nx}\left(z_{k}^{*},\beta_{k},\phi_{kk},N\right)$ and $d_{Nx}\left(z_{k}^{*},\beta_{k},\phi_{kk},N\right)$ are complex, we focus instead on finding sufficient conditions for these expressions to have a specific sign for all $z_{k}^{*}$. Note that not all coefficients in $\left\{ d_{m,Nx}\right\} _{m=0}^{7}$ can be simultaneously negative, as $d_{0,Nx}>0$. Thus, we focus on finding the regions for which $d_{m,Nx}>0$ for each $m=0,\dots,7$. By the previous argument, we also focus on finding regions where the coefficients $n_{m,Nx}>0$ for all $m=1,\dots,6$. We verify in the supplementary file \emph{Gumbel\_N.nb} that the coefficients $n_{m,Nx}$ and $d_{m,Nx}$ in (\ref{n_Nx_coefficients}) and \eqref{d_Nx_coefficients}, respectively, are positive if either of the two conditions below, (a) or (b), hold:
\begin{itemize}
    \item[(a)] $\phi_{kk}\leq0$, $N\geq 2$ and $\beta_{k}>0$.
    \item[(b)]  $\phi_{kk}>0$, $N\geq 2$ and $\beta_{k}>g_x(N)\phi_{kk}$, where $g_x(N)\phi_{kk}$ is the unique non-zero real root of $n_{6,Nx}$ (viewed as a polynomial in $\beta_k$) and $g_x(N)$ is given by 
    \begin{equation}
        g_x(N):=\frac{\left(2N^{2}-2N+1\right)}{N\left(N^{2}-N+1\right)}.\label{def:gxN_proof}
    \end{equation}
\end{itemize}
We remark that finding the non-zero root of $n_{6,Nx}$ leads to a solution of a linear equation (see \eqref{n_Nx_coefficients}), so the uniqueness of the root and its linearity in $\phi_{kk}$ are obvious.
We also show (see \emph{Gumbel\_N.nb}) that if $N\geq2$, $\beta_{k}>0$ and $\phi_{kk}>0$, then $g_x(N)\phi_{kk}\geq f\left(N\right)\phi_{kk}$.

Finally, we show that there exists $\epsilon$ such that for any $(\phi_{bs},\phi_{sb})\in B_\epsilon(0)$, Proposition \ref{prop:dNxk_dN} holds. Note that from  \eqref{eqn:dzdN_derivation} and \eqref{eqn:abcd}, $\partial (Nx_k^*)/\partial N$ is a rational function w.r.t.~$(\phi_{bs},\phi_{sb})$. Moreover, at  $(\phi_{bs},\phi_{sb})=(0,0)$, $\partial (Nx_k^*)/\partial N$ is given by \eqref{eqn:dNxdN_proof2}. Thus, $\partial (Nx_k^*)/\partial N$ is continuous w.r.t.~$(\phi_{bs},\phi_{sb})$  at  $(0,0)$. We conclude that there exists $\epsilon>0$ such that for any $(\phi_{bs},\phi_{sb})\in B_\epsilon(0)\subset \sR^2$, $\partial (Nx_k^*)/\partial N>0$ provided that either (a) or (b) above hold.
\end{proof}

\begin{proof}[{\bf Proof of Proposition \ref{prop:dCSk_dN_positive}}] Using the same ideas from the Proof of Proposition \ref{prop:dpk_dN}, we want to compute the following quantity when $\varphi_1 = 0$,
\begin{equation}
    \label{eqn:dCSdN_proof}
    \frac{\partial CS_{k}^{*}}{\partial N}\Big|_{\varphi_{1}=0}=\beta_{k}\left(\frac{1}{N+1}+\frac{\partial z_{k}^{*}}{\partial N}\Big|_{\varphi_{1}=0}\right).
\end{equation}
Using equations \eqref{eqn:dp_dN}, (\ref{eqn:dzdN_derivation}) and (\ref{eqn:dx_dN}), we show (see \emph{Gumbel\_N.nb}) that
\begin{equation}\begin{split}
\label{eqn:dCSdN_proof2}
    \frac{\partial CS_{k}^{*}}{\partial N}\Big|_{\varphi_{1}=0}&=\frac{n_{CSk}\left(z_{k}^{*},\beta_{k},\phi_{kk},N\right)}{d_{CSk}\left(z_{k}^{*},\beta_{k},\phi_{kk},N\right)},
\end{split}
\end{equation}
where $n_{CSk}\left(z_{k}^{*},\beta_{k},\phi_{kk},N\right)$ and $d_{CSk}\left(z_{k}^{*},\beta_{k},\phi_{kk},N\right)$ can be written as polynomials in $e^{z_{k}^{*}}$ in the following way, 
\begin{equation}\label{eqn:n_CS_d}
\begin{split}
n_{CSk}\left(z_{k}^{*},\beta_{k},\phi_{kk},N\right)&:=\sum_{m=0}^{6}n_{m,CSk}e^{mz_{k}^{*}},\quad  \textnormal{ and }\\
d_{CSk}\left(z_{k}^{*},\beta_{k},\phi_{kk},N\right)&:=\sum_{m=0}^{6}d_{m,CSk}e^{mz_{k}^{*}}.
\end{split}
\end{equation}
Moreover, the coefficients $n_{m,CSk}$ and $d_{m,CSk}$ are polynomials on $\left\{ \phi_{kk},\beta_{k},N\right\}$ and are defined, respectively, as \[\Scale[0.8]{n_{0,CSk}=\beta_{k}^{4},}\]
\[\Scale[0.8]{n_{1,CSk}=\beta_{k}^{3}\left(\beta_{k}\left(6N-1\right)-4\phi_{kk}\right),}\]
\[\Scale[0.8]{n_{2,CSk}=\beta_{k}^{2}\left(\beta_{k}^{2}\left(15N^{2}-5N+2\right)+2\beta_{k}\left(1-9N\right)\phi_{kk}+5\phi_{kk}^{2}\right),}\]
\[\Scale[0.8]{n_{3,CSk}=\beta_{k}\left(\beta_{k}^{3}N\left(20N^{2}-10N+8\right)+\beta_{k}^{2}\phi_{kk}\left(-32N^{2}+6N+2\right)+\beta_{k}\phi_{kk}^{2}\left(14N+1\right)-2\phi_{kk}^{3}\right),}\]
\[\Scale[0.8]{n_{4,CSk}=\beta_{k}\left(\beta_{k}^{3}N^{2}\left(15N^{2}-10N+12\right)+\beta_{k}^{2}\phi_{kk}\left(-28N^{3}+6N^{2}+5N-1\right)+\beta_{k}\phi_{kk}^{2}\left(13N^{2}+2N-4\right)-2\left(N+1\right)\phi_{kk}^{3}\right),}\]
\[\Scale[0.8]{n_{5,CSk}=\beta_{k}^{2}N\left(\beta_{k}^{2}N^{2}\left(6N^{2}-5N+8\right)+\beta_{k}\phi_{kk}\left(-12N^{3}+2N^{2}+4N-2\right)+\phi_{kk}^{2}\left(4N^{2}+N-4\right)\right),}\]
\begin{equation}\label{n_CS_coefficients}
\small	{n_{6,CSk}=\beta_{k}^{3}N^{2}\left(\beta_{k}N^{2}\left(N^{2}-N+2\right)+\left(N-1-2N^{3}\right)\phi_{kk}\right)}.
\end{equation}
The coefficients $d_{m,CSk}$ are given by \[\Scale[0.8]{d_{0,CSk}=\beta_{k}^{3}\left(N+1\right),}\]
\[\Scale[0.8]{d_{1,CSk}=\beta_{k}^{2}\left(N+1\right)\left(\beta_{k}\left(6N-1\right)-4\phi_{kk}\right),}\]
\[\Scale[0.8]{d_{2,CSk}=\beta_{k}\left(N+1\right)\left(\beta_{k}^{2}\left(15N^{2}-6N+1\right)+4\beta_{k}\left(1-4N\right)\phi_{kk}+5\phi_{kk}^{2}\right),}\]
\[\Scale[0.8]{d_{3,CSk}=\left(N+1\right)\left(\beta_{k}^{3}N\left(20N^{2}-14N+4\right)+\beta_{k}^{2}\phi_{kk}\left(-24N^{2}+11N-1\right)+\beta_{k}\phi_{kk}^{2}\left(10N-3\right)-2\phi_{kk}^{3}\right),}\]
\[\Scale[0.8]{d_{4,CSk}=\beta_{k}N\left(N+1\right)\left(\beta_{k}^{2}N\left(15N^{2}-16N+6\right)+\beta_{k}\phi_{kk}\left(-16N^{2}+10N-2\right)+\left(5N-2\right)\phi_{kk}^{2}\right),}\]
\[\Scale[0.8]{d_{5,CSk}=\beta_{k}N^{2}\left(N+1\right)\left(\beta_{k}^{2}N\left(6N^{2}-9N+4\right)+\beta_{k}\phi_{kk}\left(-4N^{2}+3N-1\right)+\phi_{kk}^{2}\right),}\]
\begin{equation}\label{d_CS_coefficients}
    \small{d_{6,CSk}=\beta_{k}^{3}\left(N-1\right)^{2}N^{4}\left(N+1\right)}.
\end{equation}

Because the expressions determining $n_{CSk}\left(z_{k}^{*},\beta_{k},\phi_{kk},N\right)$ and $d_{CSk}\left(z_{k}^{*},\beta_{k},\phi_{kk},N\right)$ are complex, we focus instead on finding sufficient conditions for these expressions to have a specific sign for all $z_{k}^{*}$. Note that not all coefficients in $\left\{ d_{m,CSk}\right\}_{m=0}^{6}$ can be simultaneously negative, as $d_{0,CSk}>0$. Thus, we focus on finding the regions for which $d_{m,CSk}>0$ for each $m=0,\dots,6$. 

Case (i):  $\frac{\partial CS_{k}^{*}}{\partial N}\Big|_{\varphi_{1}=0}>0$. We verify in the supplementary file \emph{Gumbel\_N.nb} that $n_{m,CSk}>0$ and $d_{m,CSk}>0$ for all $m=0,\dots,6$ (see \eqref{n_CS_coefficients} and \eqref{d_CS_coefficients}), if either of the two conditions below, (a-i) or (b-i), hold: 
\begin{itemize}
    \item[(a-i)] $\phi_{kk}\leq0$, $N\geq 2$ and $\beta_{k}>0$.
    \item[(b-i)] $\phi_{kk}>0$, $N\geq2$ and $\beta_{k}>g_{CS}(N)\phi_{kk}$, 
    where $g_{CS}(N)\phi_{kk}$ is the unique non-zero real root of $n_{6,CSk}$ (viewed as a polynomial in $\beta_k$) and $g_{CS}(N)$ is given by
    \begin{equation}\label{eqn:gCSN_proof}
    g_{CS}(N):=\frac{2N^{3}-N+1}{N^{2}\left(N^{2}-N+2\right)}.
    \end{equation}
\end{itemize}
We remark that finding the non-zero root of $n_{6,CSk}$ leads to a solution of a linear equation (see \eqref{n_CS_coefficients}), so the uniqueness of the root and its linearity in $\phi_{kk}$ are obvious.

Case (ii): $\frac{\partial CS_{k}^{*}}{\partial N}\Big|_{\varphi_{1}=0}<0$. First, we want to show that for all $\phi_{kk}>0$,

$$n_{CSk}\left(z_{k}^{*},\beta_{k},\phi_{kk},N\right)=\sum_{m=0}^{6}n_{m,CSk}e^{mz_{k}^{*}}<0.$$
Using (\ref{def:RNphi_kk_proof}) and \eqref{eqn:gCSN_proof}, we show (see \emph{Gumbel\_N.nb}) that for all $N\geq2$ and $\phi_{kk}>0$, 
\begin{equation}\label{RNphivsr}
    f\left(N\right)\phi_{kk}<g_{CS}(N)\phi_{kk}.
\end{equation}
For $N\geq 7$, $\phi_{kk}>0$ and $\beta_k\in(f\left(N\right)\phi_{kk},g_{CS}(N)\phi_{kk})$, we show (see \emph{Gumbel\_N.nb}) that 
\begin{equation}
    \label{caseiiCSK}
    \begin{split}
    n_{0,CS{k}} &> 0 \textnormal{ for all } k=0,\cdots 5, \textnormal{ and} \\  
    n_{m,CS{6}} &< 0.
    \end{split}
\end{equation}
From \eqref{caseiiCSK}, if $N\geq 7$, $\phi_{kk}>0$ and $\beta_k\in(f\left(N\right)\phi_{kk},g_{CS}(N)\phi_{kk})$, then 
\begin{equation}
    \label{eqn:numeratorCSk}
    \begin{split}
&\sum_{m=0}^{6}n_{m,CSk}e^{mz_{k}^{*}}\\
&=\underbrace{n_{0,CS{k}}+n_{1,CS{k}}e^{z_{k}^{*}}+n_{2,CS{k}}e^{2z_{k}^{*}}+n_{3,CS{k}}e^{3z_{k}^{*}}+n_{4,CS{k}}e^{4z_{k}^{*}}+n_{5,CS{k}}e^{5z_{k}^{*}}}_{>0}+\underbrace{n_{6,CS{k}}e^{6z_{k}^{*}}}_{<0}.
 \end{split}
\end{equation}
From Proposition \ref{prop:sign_zk}, if $\beta_{k}<\gamma\left(N,\phi_{kk},u_{k}^{0}\right)$, then $z_{k}^{*}>0$. Moreover, by hypothesis, $z_{k}^{*}<\frac{1}{5}\ln 2$, then $e^{5z_{k}^{*}}<2$. It follows from \eqref{eqn:numeratorCSk} that
\begin{equation}
    \label{eqn:numeratorCSK2}
    \begin{split}
\sum_{m=0}^{6}n_{m,CSk}e^{mz_{k}^{*}}&<n_{0,CS_{k}}+2\left(n_{1,CS_{k}}+n_{2,CS_{k}}+n_{3,CS_{k}}+n_{4,CS_{k}}+n_{5,CS_{k}}\right)+n_{6,CS_{k}}\\
&=\beta_k\sum_{m=0}^{3}y_{m,k}\beta_k^m,
    \end{split}
\end{equation}
where $y_{k,0}=-4\left(N+2\right)\phi_{kk}^{3}$,
\[\Scale[0.92]{y_{k,1}=4\left(2N^{3}+7N^{2}+6N+1\right)\phi_{kk}^{2},}\]
\[\Scale[0.92]{y_{k,2}=-\left(2N^{5}+24N^{4}+51N^{3}+45N^{2}+18N+2\right)\phi_{kk},}\]
\begin{equation}\label{y_coefficients}
\small{y_{k,3}=\left(N^{6}+11N^{5}+22N^{4}+36N^{3}+34N^{2}+18N+3\right)}.
\end{equation}
We show (see \emph{Gumbel\_N.nb}) that for any $N\geq7$, $\phi_{kk}>0$, $z_{k}^{*}<\frac{1}{5}\ln 2$ and $$f\left(N\right)\phi_{kk}<\beta_{k}<\min\{f_{CS}(N,\phi_{kk}),\gamma(N,\phi_{kk},u_k^0)\},$$ the right-hand side of \eqref{eqn:numeratorCSK2} is negative and $d_{CSk}\left(z_{k}^{*},\beta_{k},\phi_{kk},N\right)>0$ for all $z_{k}^{*}$, where $f_{CS}(N,\phi_{kk})$ is the largest real root of the third degree polynomial $\sum_{m=0}^{3}y_{m,k}\beta_k^m$. 
Thus, $\frac{\partial CS_{k}^{*}}{\partial N}\Big|_{\varphi_{1}=0}<0$.

We next verify that $\sum_{m=0}^{3}y_{m,k}\beta_k^m$ has three real roots and the largest real root, $f_{CS}(N,\phi_{kk})$, is linear in $\phi_{kk}$ and can thus be expressed as $f_{CS}(N,\phi_{kk})=f_{CS}(N)\phi_{kk}$. We define
\begin{equation}\label{eqn:cardano_ym}
    \begin{split}
        \breve{t}_k &:= -\frac{(N+1)t\left(N\right)\phi_{kk}^{2}}{3\left(N^{6}+11N^{5}+22N^{4}+36N^{3}+34N^{2}+18N+3\right)^{2}}, \\
        \breve{s}_k &:= -\frac{2s\left(N\right)\phi_{kk}^{3}}{27\left(N^{6}+11N^{5}+22N^{4}+36N^{3}+34N^{2}+18N+3\right)^{3}}, \textnormal{ and} \\
        \breve{\Delta}_k &:=-\frac{4\delta\left(N\right)\phi_{kk}^{6}}{27\left(N^{6}+11N^{5}+22N^{4}+36N^{3}+34N^{2}+18N+3\right)^{4}}, 
    \end{split}
\end{equation}
where $t\left(N\right)=\sum_{j=0}^{9}a_{j}N^{j}$, $s\left(N\right)=\sum_{j=0}^{15}b_{j}N^{j}$, $\delta\left(N\right)=\sum_{j=0}^{16}c_{j}N^{j}$ and 
\[\Scale[0.67]{\left\{ a_{0},\dots,a_{9}\right\} =\left\{ -32,-328,-1124,-1516,-671,577,740,364,68,4\right\} ,}\]\vspace{-22pt}
\[\Scale[0.67]{\left\{ b_{0},\dots,b_{15}\right\} =\left\{ 872,10098,45378,102078,118692,39084,-73701,-98271,-22581,50841,61656,34542,11412,2214,216,8\right\} ,}\]
\[\Scale[0.67]{\left\{ c_{0},\dots,c_{16}\right\} =\left\{ -816,-9464,-42275,-92522,-97149,-4202,113985,130285,40882,-40699,-51820,-24692,-4259,1118,729,136,8\right\} .}\]
From \eqref{eqn:cardano_ym}, it follows that for any $N\geq7$, $\phi_{kk}>0$, then $\breve{\Delta}_k<0$. It thus follows that $\sum_{m=0}^{3}y_{m,k}\beta_k^m$ has three simple real roots given by 
\begin{equation}\label{eqn:cardano_y_m_roots}
    \begin{split}
        \breve{\alpha}_{j}: = 2\sqrt{-\frac{ \breve{t}_k}{3}}\cos\left(\frac{\theta+2j\pi}{3}\right) - \frac{y_{2,k}}{3y_{3,k}}, \ \textnormal{ for } j=0,1,2,\\
        0<\theta<\pi, \textnormal{ and } \cos(\theta) = \frac{-\breve{s}_k/2}{\sqrt{-(\breve{t}_k/3)^3}}.
    \end{split}
\end{equation}
From \eqref{y_coefficients} and \eqref{eqn:cardano_y_m_roots},  for $\phi_{kk}<0 $, $\theta = \theta(N)$ is a function of $N$ and is independent of $\phi_{kk}$. Moreover, $\sqrt{-\frac{ \breve{t}_k}{3}}$ and $\frac{y_{2,k}}{3y_{3,k}}$ are linear in $\phi_{kk}$. Thus, for  $\phi_{kk}<0 $, the three simple real roots of $\sum_{m=0}^{3}y_{m,k}\beta_k^m$ can be written as 
\begin{equation}\label{eqn:cardano_y_m_roots2}
    \begin{split}
        \breve{\alpha_{j}} = w_{j}(N) \phi_{kk}.
    \end{split}
\end{equation}
We can thus write $f_{CS}(N, \phi_{kk})= f_{CS}\left(N\right)\phi_{kk}$, where 
\begin{equation}\label{def:f_CS}
    f_{CS}(N) = \max_{j=0,1,2} \{w_j(N)\}.
\end{equation}

Finally, we show that there exists $\epsilon>0$ such that for any $(\phi_{bs},\phi_{sb})\in B_{\epsilon}(0)$, (i) and (ii) in Proposition \ref{prop:dCSk_dN_positive} hold. From  \eqref{eqn:dzdN_derivation} and \eqref{eqn:abcd}, $\partial (CS_k^*)/\partial N$ is a rational function w.r.t.~$(\phi_{bs},\phi_{sb})$. Moreover, at  $(\phi_{bs},\phi_{sb})=(0,0)$, $\partial (CS_k^*)/\partial N$ is given by \eqref{eqn:dCSdN_proof2}. Thus, $\partial (CS_k^*)/\partial N$ is continuous w.r.t.~$(\phi_{bs},\phi_{sb})$  at  $(0,0)$. We conclude that there exists $\epsilon>0$ such that for any $(\phi_{bs},\phi_{sb})\in B_\epsilon(0)\subset \sR^2$, (i) and (ii) in Proposition \ref{prop:dCSk_dN_positive} hold true. 

\end{proof}

\begin{proof}[{\bf Proof of Proposition \ref{prop:dpik_dN_negative}}] We want to compute the following quantity when $\varphi_1 = 0$,
\begin{equation}
    \label{eqn:dpikdN_proof}
   \frac{\partial\pi_{k}^{*}}{\partial N}=\frac{\partial p_{k}^{*}}{\partial N}x_{k}^{*}+p_{k}^{*}\frac{\partial x_{k}^{*}}{\partial N}.
\end{equation}
Using equations \eqref{eqn:dp_dN}, (\ref{eqn:dzdN_derivation}) and (\ref{eqn:dx_dN}), we show (see \emph{Gumbel\_N.nb}) that
\begin{equation}\begin{split}
\label{eqn:dpikdN_proof2}
   \frac{\partial\pi_{k}^{*}}{\partial N}\Big|_{\varphi_{1}=0}&=\frac{n_{\pi k}\left(z_{k}^{*},\beta_{k},\phi_{kk},N,u_k^0\right)}{d_{\pi k}\left(z_{k}^{*},\beta_{k},\phi_{kk},N\right)},
\end{split}
\end{equation}
where $n_{\pi k}\left(z_{k}^{*},\beta_{k},\phi_{kk},N,u_k^0\right)$ and $d_{\pi k}\left(z_{k}^{*},\beta_{k},\phi_{kk},N\right)$ can be written as polynomials in $e^{z_{k}^{*}}$ in the following way, 
\begin{equation}\label{eqn:n_pik_d}
\begin{split}
n_{\pi k}\left(z_{k}^{*},\beta_{k},\phi_{kk},N,u_k^0\right)&:= \sum_{m=2}^{7}n_{m,\pi k}e^{mz_{k}^{*}},\quad  \textnormal{ and }\\
d_{\pi k}\left(z_{k}^{*},\beta_{k},\phi_{kk},N\right)&:= \sum_{m=0}^{7}d_{m,\pi k}e^{mz_{k}^{*}}.
\end{split}
\end{equation}
Moreover, the coefficients $n_{m,\pi k}$ are as follows: \[\Scale[0.8]{n_{2,\pi k}=\beta_{k}^{3}\left(u_{k}^{0}+\beta_{k}z_{k}^{*}\right),}\]
\[\Scale[0.8]{n_{3,\pi k}=\beta_{k}^{2}\left(\beta_{k}^{2}\left(\left(5N-2\right)z_{k}^{*}-1\right)+\beta_{k}\left(\left(5N-2\right)u_{k}^{0}-2z_{k}^{*}\phi_{kk}\right)-2u_{k}^{0}\phi_{kk}\right),}\]
\[\Scale[0.8]{n_{4,\pi k}=\beta_{k}\left[\beta_{k}^{3}\left(10N^{2}z_{k}^{*}-2N\left(4z_{k}^{*}+2\right)+z_{k}^{*}\right)+\beta_{k}^{2}\left(\left(10N^{2}-8N+1\right)u_{k}^{0}+\left(z_{k}^{*}+1-6Nz_{k}^{*}\right)\phi_{kk}\right)\right]}\]
\[\Scale[0.8]{+\beta_{k}\left[\beta_{k}\left(u_{k}^{0}\left(1-6N\right)\phi_{kk}+z_{k}^{*}\phi_{kk}^{2}\right)+u_{k}^{0}\phi_{kk}^{2}\right],}\]
\[\Scale[0.8]{n_{5,\pi k}=\beta_{k}\left[\beta_{k}^{3}N\left(10N^{2}z_{k}^{*}-12Nz_{k}^{*}-6N+3z_{k}^{*}\right)+\beta_{k}^{2}\left(N\left(10N^{2}-12N+3\right)u_{k}^{0}+\left(2Nz_{k}^{*}+4N-1\right)\phi_{kk}-6N^{2}z_{k}^{*}\phi_{kk}\right)\right],}\]
\[\Scale[0.8]{+\beta_{k}\left[\beta_{k}\phi_{kk}\left(N\left(2-6N\right)u_{k}^{0}+\left(Nz_{k}^{*}+z_{k}^{*}+2\right)\phi_{kk}\right)+\left(N+1\right)u_{k}^{0}\phi_{kk}^{2}\right],}\]
\[\Scale[0.8]{n_{6,\pi k}=\beta_{k}\left[\beta_{k}^{3}N^{2}\left(5N^{2}z_{k}^{*}-2N\left(4z_{k}^{*}+2\right)+3z_{k}^{*}\right)+\beta_{k}^{2}\left(N^{2}\left(5N^{2}-8N+3\right)u_{k}^{0}+\left(N^{2}z_{k}^{*}-2N^{3}z_{k}^{*}+5N^{2}-2N\right)\phi_{kk}\right)\right],}\]
\[\Scale[0.8]{+\beta_{k}\left[\beta_{k}\phi_{kk}N\left(-2N^{2}u_{k}^{0}+Nu_{k}^{0}+\left(z_{k}^{*}+2\right)\phi_{kk}\right)+\left(Nu_{k}^{0}-2\phi_{kk}\right)\phi_{kk}^{2}\right],}\]
\begin{equation}\label{n_pik_coefficients}
\small{n_{7,\pi k}=\beta_{k}^{3}N^{2}\left[\beta_{k}N\left(N^{2}z_{k}^{*}-2Nz_{k}^{*}+z_{k}^{*}-N\right)+Nu_{k}^{0}+2N\phi_{kk}-\phi_{kk}+N^{3}u_{k}^{0}-2N^{2}u_{k}^{0}\right]}.
\end{equation}
The coefficients $d_{m,\pi k}$ are given by \[\Scale[0.8]{d_{0,\pi k}=\beta_{k}^{3},}\]
\[\Scale[0.8]{d_{1,\pi k}=\beta_{k}^{2}\left(\beta_{k}\left(7N-1\right)-4\phi_{kk}\right),}\]
\[\Scale[0.8]{d_{2,\pi k}=\beta_{k}\left(\beta_{k}^{2}\left(21N^{2}-7N+1\right)+4\beta_{k}\left(1-5N\right)\phi_{kk}+5\phi_{kk}^{2}\right),}\]
\[\Scale[0.8]{d_{3,\pi k}=\beta_{k}^{3}N\left(35N^{2}-20N+5\right)+\beta_{k}^{2}\left(-40N^{2}+15N-1\right)\phi_{kk}+\beta_{k}\left(15N-3\right)\phi_{kk}^{2}-2\phi_{kk}^{3},}\]
\[\Scale[0.8]{d_{4,\pi k}=N\left(\beta_{k}^{3}N\left(35N^{2}-30N+10\right)+\beta_{k}^{2}\left(-40N^{2}+21N-3\right)\phi_{kk}+\beta_{k}\left(15N-5\right)\phi_{kk}^{2}-2\phi_{kk}^{3}\right),}\]
\[\Scale[0.8]{d_{5,\pi k}=\beta_{k}N^{2}\left(\beta_{k}^{2}N\left(21N^{2}-25N+10\right)+\beta_{k}\left(-20N^{2}+13N-3\right)\phi_{kk}+\left(5N-1\right)\phi_{kk}^{2}\right),}\]
\[\Scale[0.8]{d_{6,\pi k}=\beta_{k}N^{3}\left(\beta_{k}^{2}N\left(7N^{2}-11N+5\right)+\beta_{k}\left(-4N^{2}+3N-1\right)\phi_{kk}+\phi_{kk}^{2}\right),}\]
\begin{equation}\label{d_pi_coefficients}
    \small{d_{7,\pi k}=\beta_{k}^{3}\left(N-1\right)^{2}N^{5}}.
\end{equation}
Because the expressions determining $n_{\pi k}\left(z_{k}^{*},\beta_{k},\phi_{kk},N,u_k^0\right)$ and $d_{\pi k}\left(z_{k}^{*},\beta_{k},\phi_{kk},N\right)$ are complex, we focus on finding sufficient conditions for these expressions to have a specific sign for all $z_{k}^{*}$.

(i) Case $\frac{\partial \pi_{k}^{*}}{\partial N}\Big|_{\varphi_{1}=0}<0:$ We verify in the supplementary file \emph{Gumbel\_N.nb} that $n_{\pi k}$ and $d_{\pi k}$ (see \eqref{eqn:n_pik_d}) are negative and positive, respectively, if either of the five conditions below, (a-i)-(a-v), hold: 

\begin{itemize}
    \item[(a-i)] $\phi_{kk}<0$, $N\geq2$, $0<\beta_{k}<f_{\pi}\left(N,\phi_{kk}\right)$ and $z_{k}^{*}<r_{\pi,z,1}\left(N,\phi_{kk},u_{k}^{0},\beta_k\right)$ where 
\begin{equation}
    \label{def:rpi1z}
r_{\pi,z,1}\left(N,\phi_{kk},u_{k}^{0},\beta_k\right):=\frac{\sum_{l=0}^{3}n_{\pi z,l}\beta_{k}^{l}}{\beta_{k}^{2}\left(N^{2}-2N^{3}\right)\phi_{kk}+\beta_{k}^{3}\left(5N^{4}-8N^{3}+3N^{2}\right)+\beta_{k}N\phi_{kk}^2}
\end{equation}
with coefficients given by 
\begin{equation}
    \begin{array}{c}
n_{\pi z,0}=\phi_{kk}^{2}\left(-Nu_{k}^{0}+2\phi_{kk}\right),\\
n_{\pi z,1}=N\phi_{kk}\left(2N^{2}u_{k}^{0}-Nu_{k}^{0}-2\phi_{kk}\right),\\
n_{\pi z,2}=N\left(-5N^{3}u_{k}^{0}+8N^{2}u_{k}^{0}-3Nu_{k}^{0}-5N\phi_{kk}+2\phi_{kk}\right),\\
n_{\pi z,3}=4N^{3}.
\end{array}
\end{equation}
Moreover, $f_{\pi}\left(N,\phi_{kk}\right)$ is the largest real root of the third degree polynomial
\begin{equation}\label{def:fpiNphikk}
    4 \beta_{k}^3 N^3+\beta_{k}^2 \left(2  -5 N \right)N\phi_{kk}-2 \beta_{k} N \phi_{kk}^2+2 \phi_{kk}^3 = 0. 
\end{equation}
We clarify below that the above polynomial has three real roots and that $f_{\pi}\left(N,\phi_{kk}\right)$ is linear in $\phi_{kk}$ and can thus be expressed as $f_{\pi}\left(N,\phi_{kk}\right)=f_{\pi}\left(N\right)\phi_{kk}$.

\item[(a-ii)] $\phi_{kk}<0$, $N\geq2$, $\beta_{k}\geq f_{\pi}\left(N\right)\phi_{kk}$ and $z_{k}^{*}<-\frac{u_{k}^{0}}{\beta_{k}}$.

\item[(a-iii)] $\phi_{kk}=0$, $N\geq2$, $\beta_{k}>0$ and $z_{k}^{*}<-\frac{u_{k}^{0}}{\beta_{k}}$.

\item[(a-iv)] $\phi_{kk}>0$, $N\geq2$, $f\left(N\right)\phi_{kk}<\beta_{k}\leq h_{\pi}\left(N\right)\phi_{kk}$ and $z_{k}^{*}<r_{\pi,z,2}\left(N,\phi_{kk},u_{k}^{0},\beta_k\right)$ where 
\begin{align}
h_{\pi}\left(N\right)&:=\frac{(2N-1)}{N^{2}}, \textnormal{ and} \label{def:rpi2}\\
r_{\pi,z,2}\left(N,\phi_{kk},u_{k}^{0},\beta_k\right)&:=\frac{-N^{3}u_{k}^{0}+\beta_{k}N^{2}+2N^{2}u_{k}^{0}-Nu_{k}^{0}-2N\phi_{kk}+\phi_{kk}}{\beta_{k}\left(N^{3}-2N^{2}+N\right)}. \label{def:rpiz2}
\end{align}

\item[(a-v)] $\phi_{kk}>0$, $N\geq2$, and $\beta_{k}>h_{\pi}\left(N\right)\phi_{kk}$ and $z_{k}^{*}<-\frac{u_{k}^{0}}{\beta_{k}}$. 
\end{itemize}

For completeness, we quickly verify that as stated in (a-i), \eqref{def:fpiNphikk} has three real roots and that $f_{\pi}\left(N,\phi_{kk}\right)$ is linear in $\phi_{kk}$. We use Cardano's condition, for $\phi_{kk}<0$, $\hat{\Delta}_k:=(\hat{s}_k/2)^2 + (\hat{t}_k/3)^3<0$, where 
\begin{equation}\label{eqn:cardano_rpi1}
    \begin{split}
        \hat{t}_k &:= -\frac{\left(49 N^2-20 N+4\right) \phi_{kk}^2}{48 N^4}, \textnormal{ and} \\
        \hat{s}_k &:= \frac{(N+2) \left(127 N^2-32 N+4\right) \phi_{kk}^3}{864 N^6}. \\
    \end{split}
\end{equation}
We note that for $\phi_{kk}<0 $, $f_{\pi}\left(N\right)$ satisfies
\begin{equation}\label{def:rpi1}\begin{split}
     f_{\pi}\left(N\right)&: = \max_{j=0,1,2} \left\{\frac{\hat{\alpha}_j}{\phi_{kk}}\right\}, 
\end{split}
\end{equation}
where $\hat{\alpha}_j$ is given by \eqref{eqn:cardano_n5p_roots}, after replacing $(t_k,s_k)$ by $(\hat{t}_k,\hat{s}_k)$.

The five conditions (a-i)-(a-v) can be aggregated into the following single condition: $n_{\pi k}$ and $d_{\pi k}$ are negative and positive, respectively, if 
$\left(\beta_{k},\phi_{kk},N\right)$ satisfy (\ref{condition_existence_beta}) and $z_{k}^{*}<g_{\pi,z}\left(N,\phi_{kk},u_{k}^{0},\beta_{k}\right)$, where 
\begin{equation}
    \label{def:R1piz}
g_{\pi,z}\left(N,\phi_{kk},u_{k}^{0},\beta_{k}\right):=\begin{cases}
r_{\pi,z,1}\left(N,\phi_{kk},u_{k}^{0},\beta_{k}\right) &\textnormal{if } \phi_{kk}<0 \textnormal{ and } 0<\beta_{k}<f_{\pi}\left(N\right)\phi_{kk},\\
r_{\pi,z,2}\left(N,\phi_{kk},u_{k}^{0},\beta_{k}\right) &\textnormal{if } \phi_{kk}>0\textnormal{ and }f\left(N\right)\phi_{kk}<\beta_{k}\leq h_{\pi}\left(N\right)\phi_{kk},\\
-u_{k}^{0}/\beta_{k} &\textnormal{if } \phi_{kk}=0\textnormal{ or }(\phi_{kk}<0\textnormal{ and }\beta_{k}\geq f_{\pi}\left(N\right)\phi_{kk}),\\ 
-u_{k}^{0}/\beta_{k} & \textnormal{if } (\phi_{kk}>0\textnormal{ and }\beta_{k}>h_{\pi}\left(N\right)\phi_{kk}) 
\end{cases}
\end{equation}
and the quantities $r_{\pi,z,1}\left(N,\phi_{kk},u_{k}^{0},\beta_{k}\right)$, $f_{\pi}\left(N\right)$, $h_{\pi}\left(N\right)$ 
and $r_{\pi,z,2}\left(N,\phi_{kk},u_{k}^{0},\beta_{k}\right)$ are given by \eqref{def:rpi1z}, \eqref{def:rpi1}, \eqref{def:rpi2} and \eqref{def:rpiz2}, respectively.

(ii) Case $\frac{\partial\pi_{k}^{*}}{\partial N}|_{\varphi_{1}=0}>0.$  We verify in the supplementary file \emph{Gumbel\_N.nb} that  $n_{\pi k}$ and $d_{\pi k}$ (see \eqref{eqn:n_pik_d}) are positive if either of the two conditions below, (b-i) or (b-ii), hold:
\begin{itemize}
    \item[(b-i)] $\phi_{kk}\leq0$, $N\geq2$, $\beta_{k}>0$ and $z_{k}^{*}>f_{\pi,z}\left(N,\phi_{kk},u_{k}^{0},\beta_{k}\right)$, where
    \begin{equation}
   \label{def:R2piz} f_{\pi,z}\left(N,\phi_{kk},u_{k}^{0},\beta_{k}\right):=r_{\pi,z,2}\left(N,\phi_{kk},u_{k}^{0},\beta_{k}\right),
    \end{equation} 
and $r_{\pi,z,2}\left(N,\phi_{kk},u_{k}^{0},\beta_{k}\right)$ is given by \eqref{def:rpiz2}.
\item[(b-ii)] $\phi_{kk}>0$, $N\geq2$, $\beta_{k}>g_{\pi}\left(N,\phi_{kk}\right)$ and $z_{k}^{*}>f_{\pi,z}\left(N,\phi_{kk},u_{k}^{0},\beta_{k}\right)$, where $g_{\pi}\left(N,\phi_{kk}\right)$ is the unique real root of the third degree polynomial 
\begin{equation}
    \label{def:polynomialgpi}
    \small{\beta_{k}^3 N^3\left(N^2-1\right)+\beta_{k}^2 N\left(-7 N^3 +10 N^2 -5 N +1 \right)\phi_{kk}
    +\beta_{k}N \left(6 N^2 -7 N +3  \right)\phi_{kk}^2-(2 N^2 -2 N +1)\phi_{kk}^3}.
\end{equation}
\end{itemize}
We next verify that the above polynomial indeed has a unique real root and that $g_{\pi}\left(N,\phi_{kk}\right)$ is linear in $\phi_{kk}$ and can thus be expressed as $g_{\pi}\left(N,\phi_{kk}\right)=g_{\pi}\left(N\right)\phi_{kk}$. We use Cardano's condition, $\overline{\Delta}_k:=(\overline{s}_k/2)^2 + (\overline{t}_k/3)^3>0$, where 
\[\Scale[0.95]{ \overline{t}_k := -\frac{\left(31 N^6-119 N^5+179 N^4-135 N^3+54 N^2-10 N+1\right) \phi_{kk}^2}{3 (N-1)^2 N^4 (N+1)^2}, \textnormal{ and}}\]
\[\Scale[0.95]{ \overline{s}_k := -\frac{\left(362 N^9-2013 N^8+4878 N^7-6728 N^6+5781 N^5-3186 N^4+1117 N^3-237 N^2+30 N-2\right) \phi_{kk}^3}{27 (N-1)^3 N^6 (N+1)^3}}.\]
Next, we express this solution. For $\phi_{kk}>0 $, $g_{\pi}\left(N\right)$ satisfies 
\begin{equation}\label{def:rpi}\begin{split}
     g_{\pi}\left(N\right)&: = \frac{\textnormal{Car}(\overline{s}_k,\overline{\Delta}_k) - \frac{\overline{b}_2}{3}}{\phi_{kk}}, 
\end{split}
\end{equation}
where $\overline{b}_2 = -\frac{\left(7 N^3-10 N^2+5 N-1\right) \phi_{kk}}{(N-1) N^2 (N+1)}$ and $\textnormal{Car}(\cdot,\cdot)$ is given by \eqref{def:CardanoRoot}.

Finally, we show that there exists $\epsilon>0$ such that for any $(\phi_{bs},\phi_{sb})\in B_{\epsilon}(0)$, (i) and (ii) in Proposition \ref{prop:dpik_dN_negative} hold. From  \eqref{eqn:dzdN_derivation} and \eqref{eqn:abcd}, $\partial \pi_k^*/\partial N$ is a rational function w.r.t.~$(\phi_{bs},\phi_{sb})$. Moreover, at  $(\phi_{bs},\phi_{sb})=(0,0)$, the derivative $\partial \pi_k^*/\partial N$ is given by \eqref{eqn:dpikdN_proof2}. Thus, $\partial \pi_k^*/\partial N$ is continuous w.r.t.~$(\phi_{bs},\phi_{sb})$  at $(0,0)$. Therefore, there exists $\epsilon>0$ such that for any $(\phi_{bs},\phi_{sb})\in B_\epsilon(0)$, cases (i) and (ii) in Proposition \ref{prop:dpik_dN_negative} hold true. 

\end{proof}

\newpage
\clearpage
\bibliographystyle{apalike}
\bibliography{ref}


\end{document}